\documentclass[11pt]{article}
\usepackage{mathtools}
\usepackage{hyperref}
\usepackage{graphicx}
\usepackage{graphics}
\usepackage{dsfont}
\usepackage{epsfig}
\usepackage{bbm}
\usepackage{amsmath,amssymb,amsthm,amscd}
\usepackage[sort&compress,square,numbers]{natbib} 
\usepackage{float}
\usepackage{braket}
\usepackage[arrow, matrix, curve]{xy}
\usepackage[latin1]{inputenc}
\usepackage{tikz}
\usepackage{mathabx}
\usetikzlibrary{shapes,arrows}
\usepackage{fancybox}
\usepackage{tikz}
\usetikzlibrary{matrix}
\usepackage{calc}
\usepackage{MnSymbol}
\usepackage[OT2,T1]{fontenc}
\DeclareSymbolFont{cyrletters}{OT2}{wncyr}{m}{n}
\DeclareMathSymbol{\Sha}{\mathalpha}{cyrletters}{"58}

\restylefloat{table}

\newcommand{\nn}{{\nonumber}}
\def\ket{\rangle}
\def\bra{\langle}

\def\Mgn[#1]#2{{\overline{\cal M}_{#1,#2}}}
\def\pqs[#1,#2]{{\footnotesize{$\left[\begin{array}{c} #1\\#2  \end{array}\right]$}}} 
\def\pqsu[#1,#2]{\left[\begin{array}{c} #1\\#2  \end{array}\right]} 
\def\pqssu[#1,#2]{{\footnotesize{\left[\begin{array}{c} #1\\#2  \end{array}\right]}}} 
\def\pqh[#1,#2]{{\footnotesize{$\left[\begin{array}{c} #1\\#2  \end{array}\right]$}}} 
\def\pqhu[#1,#2]{\left[\begin{array}{c} #1\\#2  \end{array}\right]}

\def\Mnote#1{{\bf[MP: #1]}}

\newcommand{\ba}{\begin{eqnarray*}}
\newcommand{\ea}{\end{eqnarray*}}
\newcommand{\ban}{\begin{eqnarray}}
\newcommand{\ean}{\end{eqnarray}}
\newcommand{\be}{\begin{equation}}
\newcommand{\ee}{\end{equation}}
\newcommand{\ben}{\begin{equation}}
\newcommand{\een}{\end{equation}}
\newcommand{\beq}{\begin{equation}}
\newcommand{\eeq}{\end{equation}}
\numberwithin{equation}{section}

\def\ket{\rangle}
\def\bra{\langle}

\numberwithin{equation}{section}

\newtheorem{theorem}{Theorem}[section]

\newtheorem{lemma}[theorem]{Lemma}

\begin{document}

\baselineskip=15pt

\begin{titlepage}
\begin{flushright}
\parbox[t]{1.73in}{\flushright 
UPR-1275-T\\
CERN-PH-TH-2015-273}
\end{flushright}

\begin{center}

\vspace*{ 0.9cm}

{\Large \bf \text{Origin of Abelian Gauge Symmetries in}}\\
{\Large \bf \text{  Heterotic/F-theory Duality 
}}

\vspace*{ .2cm}
\vskip 0.5cm

\renewcommand{\thefootnote}{}
\begin{center}
 {Mirjam Cveti\v{c}$^{1,2}$,  Antonella Grassi$^{3}$, Denis Klevers$^4$, Maximilian Poretschkin$^1$, Peng Song$^1$}
\end{center}
\vskip .2cm
\renewcommand{\thefootnote}{\arabic{footnote}}

$\,^1$ {Department of Physics and Astronomy,\\
University of Pennsylvania, Philadelphia, PA 19104-6396, USA} \\[.3cm]

$\,^2$ {Center for Applied Mathematics and Theoretical Physics,\\ University of Maribor, Maribor, Slovenia}\\[.3cm]

$\,^3$ {Department of Mathematics,\\
University of Pennsylvania, Philadelphia, PA 19104-6396, USA} \\[.3cm]

$\,^4$ {Theory Group, Physics Department, CERN,\\
CH-1211, Geneva 23, Switzerland} \\[.3cm]

{cvetic\ \textsf{at}\ cvetic.hep.upenn.edu,  grassi\ \textsf{at}\ math.upenn.edu, denis.klevers\ \textsf{at}\ cern.ch,  mporet \textsf{at}\ sas.upenn.edu, songpeng\ \textsf{at} \ sas.upenn.edu}

 \vspace*{0.0cm}

\end{center}

\vskip 0.0cm

\begin{center} {\bf ABSTRACT } \end{center}
We study aspects of heterotic/F-theory duality for compactifications with Abelian gauge 
symmetries. We consider F-theory on general Calabi-Yau manifolds with a rank one Mordell-Weil 
group of rational sections. By rigorously performing the stable degeneration limit in a class
of toric models, we derive both the Calabi-Yau geometry as well as the
spectral cover describing the vector bundle in the heterotic dual theory. 
We carefully investigate
the spectral cover employing the group law on the elliptic curve in the heterotic theory. We find 
in explicit examples that there are three different
classes of heterotic duals that have U(1) 
factors in their low energy effective theories: 
split spectral covers describing bundles with S$($U$(m)\times$U$(1))$ structure group, spectral 
covers containing torsional sections that seem to give rise to bundles with 
SU$(m)\times\mathbb{Z}_k$ structure group and bundles with purely non-Abelian 
structure groups having a 
centralizer in $\text{E}_8$ containing a U(1) factor. In the former two cases, it is required that the 
elliptic fibration on the heterotic side has a non-trivial Mordell-Weil group.  
While the number of geometrically massless U(1)'s is determined entirely by geometry on the 
F-theory side, on the heterotic side the correct number of U(1)'s is found by taking into account 
a St\"uckelberg mechanism in the lower-dimensional effective theory. In geometry, this 
corresponds to the condition that sections in the two half K3 surfaces that arise in the stable degeneration limit of F-theory can be glued together globally. 

\begin{flushright}
\parbox[t]{1.2in}{November, 2015}
\end{flushright}
\end{titlepage}

\tableofcontents


\newpage

\section{Introduction and Summary of Results}

The study of effective theories of string theory in lower dimensions with minimal supersymmetry 
are both of conceptual and phenomenological relevance. Two very prominent avenues to their 
construction are Calabi-Yau compactifications of the $E_8\times E_8$ heterotic string and of 
F-theory, respectively. The defining data of the two compactifications are seemingly very 
different. While a compactification to $10-2n$ dimensions is defined in the heterotic string by a 
complex $n$-dimensional Calabi-Yau manifold $Z_n$ and a holomorphic, 
semi-stable vector bundle  $V$ \cite{green2012superstring,polchinski1998string}, in F-theory one needs 
to specify a complex $(n+1)$-dimensional elliptically-fibered Calabi-Yau 
manifold $X_{n+1}$ 
\cite{Vafa:1996xn,Morrison:1996pp,Morrison:1996na}. For an elliptic K3-fibered 
$X_{n+1}$ and an elliptically fibered $Z_n$, however, 
both formulations of compactifications of string theory are physically equivalent. The defining 
data of both sides are related to each other by heterotic/F-theory duality 
\cite{Vafa:1996xn,Morrison:1996pp,Morrison:1996na}. 
Most notably, this duality allows making statements about the heterotic vector bundle $V$ in terms of the controllable geometry of the Calabi-Yau manifold $X_{n+1}$ on the 
F-theory side. Studying the structure of the heterotic vector bundle $V$ is crucial for 
understanding the gauge theory sector of the resulting effective theories. In this note, we 
present key steps towards developing the geometrical duality map between heterotic and 
F-theory compactifications with Abelian gauge symmetries in their effective theories. 

Since the advent of F-theory, the  
matching of gauge symmetry and the matter content in the effective theories  
has been studied in heterotic/F-theory duality \cite{Vafa:1996xn,Morrison:1996pp,Morrison:1996na}. Mathematically, the duality astonishingly allows to use 
the data of singular Calabi-Yau manifolds $X_{n+1}$ in F-theory to efficiently construct 
vector bundles $V$ on the heterotic side, which is typically very challenging. The duality can be 
precisely formulated in the so-called stable 
degeneration limit of $X_{n+1}$ \cite{Aspinwall:1997ye}, in which its K3-fibration degenerates 
into two half K3-fibrations $X_{n+1}^\pm$,
\beq
	X_{n+1}\,\,\rightarrow\,\, X_{n+1}^{+}\cup_{Z_n}X_{n+1}^-\,,
\eeq
that intersect in the heterotic Calabi-Yau manifold, $X_{n+1}^+\cap X_{n+1}^-=Z_n$. It
can be shown that $X_{n+1}^\pm$ naturally encode the heterotic vector bundle 
$V$ on elliptically fibered Calabi-Yau manifolds $Z_n$ 
\cite{Friedman:1997yq}. The most concrete 
map between the data of $X_{n+1}$ in stable degeneration and the heterotic side is realized 
if $V$ is described  
by a spectral cover employing the Fourier-Mukai transform 
\cite{Friedman:1997yq,Donagi:1997dp} (for more details see e.g.~\cite{Andreas:1998zf} and references therein). Heterotic/F-theory
duality has been systematically applied using 
toric geometry for the construction of vector bundles $V$ with 
non-Abelian structure groups described both via spectral covers and half K3 fibrations, 
see e.g.~\cite{Candelas:1996su,Berglund:1998ej} for representative works.
More recently, heterotic/F-theory duality has been used to 
study the geometric constraints on both sides of the duality in four-dimensional 
compactifications and to characterize the arising low-energy physics \cite{Anderson:2014gla}, 
see also \cite{Hayashi:2008ba}. Furthermore,  computations of both vector bundle and M5-brane 
superpotentials could be performed by calculation of the F-theory superpotential  using powerful 
techniques from mirror symmetry \cite{Grimm:2009ef,Grimm:2009sy,Jockers:2009ti}.
In addition, the heterotic/F-theory duality has been recently explored for studies of moduli-dependent prefactor of M5-instanton corrections to the superpotential in  F-theory compactifications \cite{Cvetic:2011gp,Cvetic:2012ts}.
The focus of all these works has been on vector bundles $V$ with non-Abelian structure
groups, see however \cite{Aldazabal:1996du,Klemm:1996hh} for first works on aspects of heterotic/F-theory duality with U(1)'s.

In this work, we will apply the simple and unifying description on the F-theory side in terms of 
elliptically fibered Calabi-Yau manifolds $X_{n+1}$ to study explicitly, using stable degeneration, 
the structure of spectral covers yielding heterotic vector bundles that give rise to U(1) 
gauge symmetry  in the lower-dimensional effective theory,  continuing the analysis  explained in  the 2010 talk \cite{GrassiIPMU}.\footnote{We have recently learned that  A.~Braun and S.~Sch\"afer-Nameki have been working on similar techniques.}

Abelian gauge symmetries are desired ingredients for controlling the phenomenology both of 
extensions of the standard model as well as of GUT theories. Recently, there has been 
tremendous progress on the construction of F-theory compactifications with Abelian gauge 
symmetries based on the improved understanding of elliptically fibered Calabi-Yau 
manifold $X_{n+1}$ with higher rank Mordell-Weil group of rational sections, see the 
representative works \cite{Morrison:2012ei,Borchmann:2013jwa,Cvetic:2013nia,Cvetic:2013uta,Borchmann:2013hta,Cvetic:2013jta,Cvetic:2013qsa,Braun:2014qka,Klevers:2014bqa,Cvetic:2015ioa}. In contrast, it has 
been long known that Abelian gauge symmetries in the heterotic theory can for example be constructed by 
considering background bundle $V$ with line bundle components
\cite{green2012superstring}.
The setup we are studying in this work is the duality map between the concrete and known 
geometry of the Calabi-Yau manifold $X_{n+1}$ with a rank one Mordell-Weil group in 
\cite{Morrison:2012ei} on the F-theory side and the data of the Calabi-Yau manifold $Z_n$ and 
the vector bundle $V$ defining the dual heterotic compactification. We will demonstrate, at the 
hand of a number of concrete examples, 
the utility of the F-theory Calabi-Yau manifold $X_{n+1}$ for the construction of vector bundles 
with non-simply connected structure groups that arise naturally in this duality. In particular,
the F-theory side will guide us to the physical interpretation of less familiar or
novel structures in the heterotic vector bundle.

There are numerous key advancements in this direction presented in this work:
\begin{itemize}
\item We rigorously perform the stable degeneration limit of a class of F-theory
Calabi-Yau manifolds   $X_{n+1}$ with U(1) Abelian gauge symmetry  using toric geometry, applying and extending the techniques of \cite{Grassi:2012qw}. We explicitly extract the data of the two half-K3 surfaces inside $X_{n+1}^{\pm}$, the spectral covers  and the heterotic Calabi-Yau manifold $Z_n$. We point out the non-commutativity of the
stable degeneration limit and birational maps, 
such as the one to the Weierstrass model. The stable degeneration limit we perform, which we
denote as ``toric stable degeneration'', preserves the structure of the Mordell-Weil group
of rational sections before and after the limit, which is, in contrast, obscured in the 
stable degeneration limit performed in the Weierstrass model.
We apply our general techniques to Calabi-Yau manifolds with elliptic fiber in
$\text{Bl}_1\mathbb{P}^2(1,1,2)$, which yield one U(1) in F-theory \cite{Morrison:2012ei}.  
\item We illuminate the systematics in the mapping under heterotic/F-theory duality
between F-theory with a Mordell-Weil group and heterotic vector bundles with non-simply 
connected structure groups leading to U(1)'s in their effective theories. 
We find that a single type of
F-theory geometry $X_{n+1}$ can be dual to a whole range of different phenomena in the 
heterotic string, at the hand of numerous concrete examples. We find three different classes of examples 
of how a U(1) gauge group is obtained in the heterotic string: one class of examples has a split 
spectral cover, which is a well-known ingredient for obtaining U(1) gauge groups
in the heterotic literature starting with 
\cite{Blumenhagen:2005ga} and the F-theory literature, 
see e.g.~\cite{Marsano:2009gv,Blumenhagen:2009yv,Grimm:2010ez}; another class of models have 
a spectral cover containing a torsional section of the heterotic Calabi-Yau manifold 
$Z_n$, where, duality suggests that this should describe zero-size instantons of discrete 
holonomy, as considered in \cite{Aspinwall:1998xj}; in a last set of examples, the U(1) arises as 
the commutant inside $\text{E}_8$ of vector bundles with purely non-Abelian structure groups. We analyze 
the emerging spectral covers by explicit computations in the group law on the elliptic curve in 
$Z_n$. In the first two classes of examples, it is 
crucial that  the heterotic elliptic fibration $Z_n$ exhibits rational sections, as also found in 
\cite{Choi:2012pr}.
In addition, in certain examples, the U(1) is only visible in the  half K3 fibration (and in 
$Z_n$), but not in the spectral cover.  
\item Whereas the number of massless U(1)'s on the F-theory side equals the Mordell-Weil rank 
of $X_{n+1}$, it is on the heterotic side a mixture of geometry and effective field theory 
effects: while the analysis of the spectral cover can be performed already in 8D, in 6D and lower 
dimensions U(1)'s can   be lifted from the massless spectrum by a St\"uckelberg effect, 
i.e.~gaugings of axions \cite{green2012superstring}. We understand explicitly in all three classes of examples how these gaugings arise and what is the remaining number of massless U(1) fields.
\end{itemize} 
We note that although our analysis is performed in 8D and 6D,  it is equally applicable also to 
heterotic/F-theory duality for compactifications to 4D.

This paper is organized in the following way: In Section \ref{sec:review+insights}, we provide a 
brief review of the key points of heterotic/F-theory duality as well as a discussion of
the new insights gained in this work into spectral covers and half K3-fibrations 
for vector bundles with non-simply connected structure 
groups. We review and discuss  heterotic/F-theory duality in eight and six dimensions, the 
spectral cover construction for SU($N$) bundles, specializations thereof giving rise to  
U(1) factors in the heterotic string and the St\"uckelberg mechanism rendering certain 
U(1) gauge fields massive. Section \ref{verticalmodel} contains the toric description of 
a class of F-theory models $X_{n+1}$ for which we describe a toric stable degeneration
limit.  We specialize to the toric fiber $\text{Bl}_1\mathbb{P}^2(1,1,2)$ and obtain the half
K3-fibrations as well as the dual heterotic geometry and spectral cover polynomial. In Section 
\ref{section:examples}, we present selected examples of F-theory/heterotic dual compactifications.
We illustrate the three different classes of examples with heterotic vector bundles of
structure groups $\text{S}(\text{U}(n)\times \text{U}(1))$ and $\text{S}(\text{U}(n)\times \mathbb{Z}_k)$, as well as purely 
non-Abelian ones having a centralizer in $\text{E}_8$ with one U(1) factor. There we also illustrate the 
utility of the St\"uckelberg mechanism to correctly match the number of geometrically massless 
U(1)'s on both sides of the duality. In Section \ref{sec:conclusions}, we conclude and discuss 
possibilities for future works. This work has four Appendices: we present the birational map of the quartic in $\mathbb{P}^2(1,1,2)$ to Tate and Weierstrass form in Appendix 
\ref{appendix:TateandWeierstrass}); Appendix \ref{app:noU1s} contains examples
with no U(1) factor, consistently reproducing \cite{Morrison:1996pp}; in Appendix \ref{lemma}
we state the condition for the existence of two independent rational sections and
Appendix \ref{app:noncomm} illustrates explicitly the non-commutativity of 
the stable degeneration limit and the birational map to Weierstrass form.

\section{Heterotic/F-theory Duality and U(1)-Factors}
\label{sec:review+insights}

The aim of this section is two-fold: On the one hand, we review those aspects of 
heterotic/F-theory duality in eight and six dimensions 
that are relevant for the analyses performed in this work. 
On the other hand,  we point out subtleties and new insights into heterotic/F-theory duality 
with Abelian U(1) factors. In particular, we discuss in detail split spectral covers for 
heterotic vector bundles with non-simply connected gauge groups and the heterotic 
St\"uckelberg mechanism.

In Section \ref{sec:hetF8d}, we discuss 
the fundamental duality in 8d, the standard stable degeneration limit in Weierstrass form and the 
principal matching of gauge groups and moduli. There, we also discuss a subtlety in performing
the stable degeneration limit of F-theory models with U(1) factors due to the non-commutativity 
of this limit with the map to the Weierstrass model. Section \ref{SUNellipticcurve} contains
a discussion of the spectral cover construction for SU($N$) bundles as well as of 
split spectral covers giving rise to  $\text{S}(\text{U}(N-1)\times \text{U}(1))$ and $\text{S}(\text{U}(N-1)\times \mathbb{Z}_k)$ bundles. In Section \ref{sec:hetF6d} we briefly review
heterotic/F-theory duality in 6d, before we discuss 
the St\"uckelberg effect in the effective theory of heterotic compactifications with U(1) bundles
as well as the relation to gluing condition of rational sections in Section \ref{sec:Stueckelberg}.

In the review part, we mainly follow \cite{Aspinwall:1996mn,Friedman:1997yq,Andreas:1998zf}, 
to which we refer for further details.

\subsection{Heterotic/F-Theory duality in eight dimensions} 
\label{sec:hetF8d}

The basic statement of heterotic/F-Theory duality is that the heterotic String (in the following, 
we always concentrate on the $E_8\times E_8$ string) compactified on a torus, 
which we denote by $Z_1$, is 
equivalent to F-Theory compactified on an elliptically fibered K3 surface $X_2$. The first 
evidence is that the moduli spaces $\mathcal{M}$ of these two theories coincide and are 
parametrized by
\be
\mathcal{M}=\text{SO}(18,2, \mathbb{Z}) \big\backslash  \text{SO}(18,2, \mathbb{R}) \big/(\text{SO}(18) \times \text{SO}(2))\, \times \mathbb{R}^+.
\ee
From a heterotic perspective this is just the parametrization of the complex and K\"ahler 
structure of the torus $Z_1$ as well as of the 
24 Wilson lines. On the F-Theory side it corresponds to the 
moduli space of algebraic K3 surfaces $X_2$
with Picard number two. The last factor corresponds to the 
vacuum expectation value of the dilaton and the size of the base $\mathbb{P}^1$ of $X_2$,
respectively. 

Lower-dimensional dualities are obtained, applying the adiabatic argument 
\cite{Vafa:1995gm}, by fibering the eight-dimensional duality over a 
base manifold $B_{n-1}$ of complex dimension $n-1$ that is common to both theories of the 
duality.

\subsubsection{The standard stable degeneration limit}   \label{stabledegeneration}

In order to match the moduli on both sides of the duality, the K3 surface $X_2$ has to undergo 
the so-called stable degeneration limit. In this limit it splits into two half K3 surfaces $X_2^+$,
$X_2^-$ as
\beq
	X_{2}\,\,\rightarrow\,\, X_{2}^{+}\cup_{Z_1}X_{2}^-\,.
\eeq
Each of these are an elliptic fibration $\pi_\pm: X_2^{\pm} \longrightarrow \mathbb{P}^1$ over 
a $\mathbb{P}^1$. These two $\mathbb{P}^1$ intersect in precisely one point so that the two 
half K3 surfaces intersect in a common elliptic fiber which is identified with the heterotic elliptic 
curve, $X_2^+\cap X_2^-=Z_1$. 
On the heterotic side, the stable degeneration limit corresponds to the large elliptic fiber
limit of $Z_1$.

\paragraph{Matching the gauge groups}
~\\
\noindent The F-theory gauge group is given by the singularities of the elliptic fibration of $X_2$, 
determining the non-Abelian part $G$, and its rational sections, which correspond to Abelian 
gauge fields \cite{Vafa:1996xn,Morrison:1996na,Bershadsky:1996nh}. In stable degeneration
the non-Abelian gauge group of F-theory is distributed into the two half K3 surfaces $X^\pm_2$ 
and matched with the heterotic side as follows. 

It is a well-known fact that the homology lattice of a half K3 surface $X_2^\pm$ is given in 
general by
\be
H_2(X_2^\pm, \mathbb{Z}) = \Gamma_8 \oplus U
\ee
Here, $U$ contains the classes of the elliptic fiber
as well as of the zero section. $\Gamma_8$ equals the root lattice of E$_8$ and splits into a 
direct sum of two contributions: the first contribution is given by the Mordell-Weil group of the 
rational elliptic surface while the second contribution is given by a sub-lattice which forms,
for the half K3 surfaces $X_{2}^\pm$ at hand, 
the root-lattice of the part  $G_\pm$ of the non-Abelian F-theory gauge group $G= G_+\times G_-$ that is of ADE type. 
In the F-Theory limit all 
fiber components are shrunken to zero size and the half K3 surface develops a singularity of type 
$G_\pm$. The possible ADE-singularities  in the case of complex 
surfaces have been classified by Kodaira  
\cite{kodaira1963compact}. Thus, 
one can always read off the corresponding gauge group from the order of vanishings of $f, g$ and $\Delta$ once the half K3 has been brought into affine Weierstrass normal form
\be \label{eq:WSF}
y^2 z = x^3 + f x z^2 +g z^3, \qquad \Delta = f^3 + 27 g^2\, ,
\ee
with $f$ and $g$ in $\mathcal{O}(4)$ and $\mathcal{O}(6)$ of $\mathbb{P}^1$, 
respectively.
For convenience of the reader, we reproduce Kodaira's classification in Table \ref{Kodairatable}. 

\begin{table}[ht]
\begin{center}
\begin{tabular}{|c|c|c|c|}\hline
order $(f)$ & order $(g)$ & order $(\Delta)$ & singularity \nn \\ \hline
$\ge 0$ & $\ge 0$ & $0$ & none \\
$ 0$ & $ 0$ & $n$ & A$_{n-1}$ \\
$ \ge 1$ & $ 1$ & $2$ & none \\
$ 1$ & $ \ge 2$ & $3$ & A$_{1}$ \\
$ \ge 2$ & $ 2$ & $4$ & A$_{2}$ \\
$ 2$ & $ \ge 3$ & $n+6$ & D$_{n+4}$ \\
$ \ge 2$ & $ 3$ & $n+6$ & D$_{n+4}$ \\
$ \ge 3$ & $ 4$ & $8$ & E$_{6}$ \\
$ 3$ & $ \ge 5$ & $9$ & E$_{7}$ \\
$ \ge 4$ & $ 5$ & $10$ & E$_{8}$ \\ \hline
\end{tabular} 
\caption{The Kodaira classification of singular fibers. Here $f$ and $g$ are the coefficients of the Weierstrass normal form, $\Delta$ is the discriminant as defined in \eqref{eq:Delta} and order refers to their order of vanishing at a particular zero.} \label{Kodairatable}
\end{center}
\end{table}

In contrast, the gauge group on the heterotic side is encoded in two vector bundles $V_1, V_2$ 
that generically carry the structure group $\text{E}_8$. Their respective commutants inside the two 
ten-dimensional $\text{E}_8$ gauge groups of the heterotic string are to be identified with the
F-theory gauge group.  As observed in \cite{Friedman:1997yq}, the moduli 
space of semi-stable $\text{E}_8$-bundles on an elliptic curve $E$ corresponds to the complex 
structure moduli space of a half K3 surface $S$ whose anti-canonical class is given by 
$E$. Furthermore, if $S$ has an ADE singularity of type $\tilde{G}_\pm$ then the structure 
group of $V_1$, $V_2$ is reduced to the centralizer $H_\pm$ of $\tilde{G}_\pm$ within 
$\text{E}_8$, respectively.
In heterotic/F-theory duality, a matching of the gauge group is then
established by identifying $S\equiv X_2^\pm$ yielding $\tilde{G}_\pm\equiv G_\pm$.

Notice that the full eight-dimensional gauge group is given by $G\times \text{U}(1)^{16-\text{rk}(G)}\times \text{U}(1)^4$. Here, the last factor accounts for the reduction of the 
metric and 
the Kalb Ramond $B$-field  along the two one-cycles of the torus in the heterotic string. 
From the F-theory perspective, all U(1) factors arise from the reduction of the $C_3$ field 
along those 2-forms in the full K3 surface $X_2$ that are orthogonal to the zero section and the 
elliptic fiber. In particular, the  $\text{U}(1)^{16-\text{rk}(G)}$ arises from the generators of
the Mordell-Weil group of the half K3 surfaces.
For a derivation in Type IIB string theory, see the recent work \cite{Douglas:2014ywa}.

\paragraph{Matching complex structure and bundle moduli}
~\\
\noindent In this section, we discuss how the heterotic moduli can be recovered from the data of the 
F-theory K3 surface \cite{Morrison:1996pp, Donagi:1998vw}. Here we restrict the discussion to 
the moduli of the heterotic torus $Z_1$ and the vector bundle (i.e.~Wilson line) moduli, ignoring 
the heterotic dilaton modulus.

So far, this discussion has been restricted to the case that the elliptic fibration of the K3 surface 
is described by a Weierstrass model. In this case, the standard stable degeneration procedure 
applies. Given the Weierstrass form \eqref{eq:WSF} for $X_2$ with $f$, $g$ sections of $\mathcal{O}(8)$ 
and $\mathcal{O}(12)$ on $\mathbb{P}^1$,
respectively, we can expand these degree eight and twelve polynomials in the affine 
$\mathbb{P}^1$-coordinate $u$ as
\be \label{eq:fgK3}
f = \sum_{i=0}^8 f_i u^i \, , \qquad g= \sum_{i=0}^{12} g_i u^i \, .
\ee
Then, the two half K3 surfaces $X_2^{\pm}$ arising in the stable degeneration limit, given as the 
Weierstrass models
\be
X^{\pm}: \qquad y^2 z = x^3 + f^{\pm} z + g^{\pm} z^3,
\ee
can be obtained from \eqref{eq:fgK3} by  the split 
\be \label{eq:splitfg}
f^+ = \sum_{i=0}^4 f_i u^i\, , \quad f^- = \sum_{i=4}^8 f_i u^i\, , \qquad g^+ = \sum_{i=0}^{6} g_i u^i\, ,\quad g^-= \sum_{i=6}^{12} g_i u^i\, ,
\ee
The "middle" polynomials $f_4$ and $g_6$ 
correspond to the heterotic elliptic curve, which then reads
\be
y^2 z = x^3 + f_4 x z^2 + g_6 z^3,
\ee
while the "upper" and "lower" coefficients correspond to the moduli of the two E$_8$-bundles.  

\subsubsection{Stable degeneration with other elliptic fiber types}

The focus of the present work are F-theory compactifications with one U(1) gauge 
group arising from elliptically fibered Calabi-Yau manifolds with two rational sections. These  
are naturally constructed using the fiber ambient space $\text{Bl}_1 \mathbb{P}^{(1,1,2)}$ 
\cite{Morrison:2012ei}. More 
precisely, we will consider K3 surfaces given as sections $\chi$ of the anti-canonical bundle 
$-K_{\mathbb{P}^1 \times \text{Bl}_1 \mathbb{P}^{(1,1,2)}}$ of 
$\mathbb{P}^1 \times \text{Bl}_1 \mathbb{P}^{(1,1,2)}$ reading
\beq
	\chi=\sum_i s_i\chi^i\,.
\eeq
Here $s_i$ and $\chi^i$ are sections of the anti-canonical bundles
$-K_{\mathbb{P}^1}=\mathcal{O}(2)$ and $-K_{\text{Bl}_1 \mathbb{P}^{(1,1,2)}}$, respectively.

Then, analogously to the above construction, one 
can perform a stable degeneration limit for these hypersurfaces as well. However, it is crucial to
note here that we can perform the stable degeneration limit in
\textit{two} possible ways, as shown in Figure
\ref{noncommutativediagram}: one way is to first take the Weierstrass normal form $W_\chi$
(upper horizontal arrow) of the full 
$\text{Bl}_1 \mathbb{P}^{(1,1,2)}$-model and then apply the split 
\eqref{eq:splitfg} to obtain two half K3 surfaces (right vertical arrow); a second way is to 
first perform stable degeneration (left vertical arrow), yielding two half K3 surfaces $\chi^\pm$ with elliptic fibers in $\text{Bl}_1 \mathbb{P}^{(1,1,2)}$, and then compute their Weierstrass normal forms $W_{\chi}^\pm$ (lower horizontal arrow). It is important
to realize, however,  that these two possible paths in the diagram \ref{noncommutativediagram} 
do not commute, as explicitly shown in Appendix \ref{app:noncomm}. 
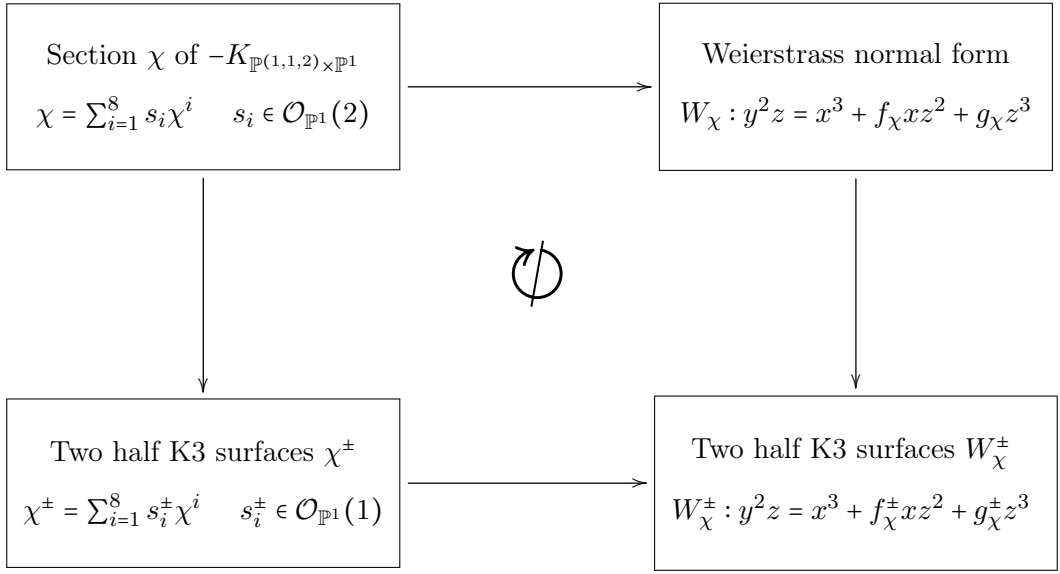
\begin{figure}[t!]
\be
\xymatrix{
\fbox{\parbox{5cm}{\begin{center} \text{Section $\chi$ of }$-K_{\mathbb{P}^{(1,1,2)}\times \mathbb{P}^1} $ \\ \vspace{0.3cm} $\chi = \sum_{i=1}^8 s_i \chi^i$ \quad $s_i \in \mathcal{O}_{\mathbb{P}^1}(2)$ \end{center} }}\ar@{->}[rr]\ar@{->}[dd]  && \fbox{\parbox{5cm}{\begin{center} \text{Weierstrass normal form} \\ \vspace{0.3cm}  $W_{\chi}: y^2 z=x^3 + f_\chi x z^2 +g_\chi z^3$ \end{center}  }}\ar@{->}[dd] \\
& \text{\Huge{ $\ncirclearrowright$}} & \\
\fbox{\parbox{5cm}{\begin{center} \text{Two half K3 surfaces} $\chi^{\pm}$ \\ \vspace{0.3cm}  $\chi^\pm = \sum_{i=1}^8 s_i^\pm \chi^i$ \quad $s_i^\pm \in \mathcal{O}_{\mathbb{P}^1}(1)$ \end{center} }}\ar@{->}[rr]  && \fbox{\parbox{5cm}{\begin{center} Two half K3 surfaces $W_{\chi}^\pm$ \\ \vspace{0.3cm}  $W_{\chi}^{\pm}: y^2 z=x^3 + f_\chi^\pm x z^2 +g_\chi^\pm z^3$ \end{center} }  } } \nn
\ee
\caption{Computing the Weierstrass normal form (horizontal arrows) and taking the stable degeneration limit (vertical arrows) does not commute.} \label{noncommutativediagram}
\end{figure}

We propose and demonstrate in Section \ref{verticalmodel}
that the natural order to perform heterotic/F-theory duality for models with U(1) factors and 
different elliptic fiber types than the Weierstrass model is to first perform stable degeneration 
with the other fiber type (left vertical arrow) and then
compute the Weierstrass model of the resulting half K3-fibrations (lower horizontal arrow) in 
order to analyze the physics of the model.

\subsection{Constructing SU($N$) bundles on elliptic curves and fibrations} \label{SUNellipticcurve}

While the description of the structure group of the vector bundle via half K3 surfaces as reviewed 
above is of high conceptual importance, it is in practice often easier to construct vector 
bundles with the desired structure group directly. In the following section, we review this 
construction for SU($N$) bundles and specializations thereof which has been studied first in \cite{Atiyah01011957} and was further developed in \cite{Donagi:1997dp,Friedman:1997yq,Friedman:1997ih}.

In this section $E$ always denotes an elliptic curve with a marked point $p$. The curve is defined 
over a general field $K$, which does not necessarily have to be algebraically closed. 
It is well-known that an elliptic curve with a point $p$ has a representation in the 
Weierstrass normal form \eqref{eq:WSF}, where $p$ reads $[x:y:z]=[0:1:0]$. 
In general, a degree zero line bundle $\mathcal{L} \longrightarrow E$, i.e.~a U(1)-bundle, takes the form
\be \label{eq:L}
\mathcal{L} = \mathcal{O} (q) \otimes \mathcal{O}(p)^{-1} = \mathcal{O}(q-p)\, ,
\ee
where $q$ denotes another arbitrary rational point on $E$ (note that over $K=\mathbb{C}$ 
every point is rational). Furthermore, we note that there is a bijective map $\phi$ from the 
elliptic curve $E$ onto its Picard group of degree zero which is defined by
\be
\phi: E \longrightarrow \text{Pic}^0(E)\, , \qquad q \mapsto q-p\, . \label{mapEtoPicE}
\ee
In particular, this extends to an isomorphism from the space of line bundles onto Pic$^0(E)$, defined by div$(\mathcal{L})=q-p$. To be more precise, the divisor map `div' is to be applied to a meromorphic section\footnote{This map is independent of the section chosen.} of $\mathcal{L}$.  For later purposes, we also recall that the addition law in Pic$^0(E)$ can be identified with the group law on $E$, which we denote by $\boxplus$, via this isomorphism.

A semi-stable SU$(N)$ vector bundle of degree zero $V$ is then given as the sum\footnote{If two or more points coincide, the situation is a bit more subtle. In this case the bundle is given by $\oplus_{i=1}^N \mathcal{O}(q_i-p)I_{r_i}$, where $r_i$ denotes the multiplicity of the point $q_i$ and $I_r$ is inductively defined by the extension sequence $0 \longrightarrow \mathcal{O} \longrightarrow I_{r-1} \longrightarrow \mathcal{O} \longrightarrow 0$. However, one usually only considers bundles up to S-equivalence which identifies $I_r$ with $\mathcal{O}^{\oplus r}$.} of $N$ holomorphic line bundles $\mathcal{L}_i$, i.e.~we have $V= \oplus_{i=1}^N \mathcal{L}_i= \oplus_{i=1}^N \mathcal{O}(q_i-p)$, such that the determinant of $V$ is trivial. The latter implies that
\be
\otimes_{i=1}^N \mathcal{O}(q_i-p) = \mathcal{O} \qquad \Leftrightarrow \qquad \boxplus_{i=1}^N q_i = 0\, .
\ee
An SU$(N)$ vector bundle is therefore determined by the choice of $N$ points on $E$ that sum up to zero. Any such $N$-tupel is determined by a projectively unique element of $H^0(E, \mathcal{O}(Np))$, i.e.~a function with $N$ zeros and a pole of order $N$ at $p$. 
Thus, the moduli space of SU$(N)$ vector bundles is given by
\be
\mathcal{M}_{\text{SU}(N)} = \mathbb{P} H^0(E, \mathcal{O}(N p))\, .
\ee 
In the affine Weierstrass form of $E$, given by \eqref{eq:WSF}, the coordinates $x$, $y$ have a 
pole of order two and three at $p$, respectively. Accordingly, any element of $\mathbb{P} H^0(E, \mathcal{O}(N p))$ enjoys an expansion
\be
w=c_0 + c_1 x + c_2 y+ c_3 x^2 + \ldots +\begin{cases} c_{N}x^{\frac{N}{2}} \,\,\,\,\,\,\,\,\,\,\, \text{if $N$ is even\,,}\\ c_N x^{\frac{N-3}{2}}y \,\,\, \text{if $N$ is odd\,,}\end{cases}\label{spectralcover}
\ee
with $c_i\in K$.
The section $w$ is called the spectral cover polynomial and has $N$ common points with $E$,
called the \textit{spectral cover}, which define the desired SU$(N)$ bundle. Counting parameters 
of \eqref{spectralcover}, one is lead to the conclusion that
\be
\mathcal{M}_{\text{SU}(N)} = \mathbb{P}^{N-1}\, .
\ee

Finally, a comment on rational versus non-rational points is in order. Generically, $p$ is the only 
point on $E$ over a general field $K$. 
However, in such a situation, it is possible to mark $N$ points in a rational way by the polynomial 
$w=0$ which give rise to an SU$(N)$ bundle in the way just described. Nevertheless, under the 
circumstances that there 
are additional rational points on $E$ and the spectral cover polynomial $w=0$ specializes 
appropriately, the structure group reduces in a certain way, as discussed next.

\subsubsection{Vector bundles with reduced structure groups}

As described in the previous section, the choice of $N$ points on $E$ describes an SU$(N)$ 
bundle. If we consider just an elliptic curve $E$ over $\mathbb{C}$, which is the geometry 
relevant for the construction of heterotic compactifications to eight dimensions, the spectral 
cover \eqref{spectralcover} can be factorized completely. This corresponds to
the 16 possible Wilson lines on $T^2$.  

In contrast, if we consider an elliptic 
curve over a function field, as it arises in elliptic fibrations $Z_n$ of $E$ over a base $B_{n-1}$ 
used for lower-dimensional heterotic compactifcations, the $N$ points are the zeros of 
\eqref{spectralcover}, which defines an $N$-section of the fibration. In non-generic situations,
where subsets of the $N$ sheets of this  $N$-section are well-defined globally, i.e.~are 
monodromy invariant,  the structure group of the vector bundle is reduced. For example, a 
separation into two sets of $k$ and $l$ sheets (with $k+l=N$), respectively, results in the 
structure group  $\text{S}(\text{U}(k)\times \text{U}(l))$. The spectral cover defined by \eqref{spectralcover} is 
called  ``split'' and defines a reducible variety inside $Z_n$, see 
e.g.~\cite{Blumenhagen:2005ga,Marsano:2009gv,Blumenhagen:2009yv,Grimm:2010ez}. In the most extreme 
case, one could have $k=1$ and $l=N-1$. 
In this case, the elliptic fibration of $Z_n$ has to necessarily have another well-defined section
in addition to the section induced by the rational point $p$: it is the one marked by the 
component of the spectral cover $w=0$ with just one sheet \cite{Choi:2012pr}.  Thus, 
the fiber $E$ has a rational point, which we denote by $q$ and one can, as discussed above, 
define a U(1) line bundle $\mathcal{L}$ via \eqref{eq:L}. As this fiberwise well-defined line 
bundle is also well-defined globally, it will induce a line bundle on $Z_n$, whose first 
Chern class is given, up to vertical components, by the difference of the sections induced by $q$ 
and $p$, cf.~\cite{Aspinwall:2005qw}.  The structure group $H$ of the vector bundle 
is in this case given by 
\beq
 H =\text{S}(\text{U}(N-1) \times \text{U}(1))\,. 
\eeq 
We will see later that this situation will be relevant situation for the construction of U(1) 
gauge groups in the heterotic string.  

We emphasize that for a U(1)-bundle alone there is no spectral cover polynomial 
\eqref{spectralcover} that would be able to detect this additional rational 
point. This is due to the fact that there is no function that has only one zero 
on an elliptic curve $E$ . However, if the rational point is accompanied by further points, rational 
or non-rational points over the  field $K$, it can very well be seen by the spectral cover.  For 
instance, one could construct 
a spectral cover from $q$ and $-q$, which would describe a bundle of structure group 
$\text{S}(\text{U}(1)\times \text{U}(1))$.

Finally, it needs to be discussed what interpretation should be given to the case that the rational 
point $q$ on the curve $E$ happens to be torsion of order $k$. In this case the structure group
$H$  
reduces further to $\text{S}(\text{U}(N) \times \mathbb{Z}_k)$. To argue for this, we invoke again
a fiberwise argument. The fiber at a generic point in $B_{n-1}$ admits a line bundle 
$\mathcal{L}=\mathcal{O}(q-p)$ with the property that $\mathcal{L}^k = \mathcal{O}$. This is
clear as the transition functions $g_{ij}$ will be subject  to $g_{ij}^k =1$ in \v{C}ech 
cohomology as $k$ times the Poincar\'{e} dual of its first Chern class is trivial. However, this is 
just the statement that the fiberwise structure group of $\mathcal{L}$ is contained in 
$\mathbb{Z}_k$. Employing that $p$ and $q$ are globally well-defined sections then suggests 
that this argument also holds  on $Z_n$.

\subsection{Heterotic/F-Theory duality in six dimensions}
\label{sec:hetF6d}

Six-dimensional heterotic/F-Theory duality arises by fibering the eight-dimensional duality 
over a common base $B_1=\mathbb{P}^1$, employing the adiabatic argument \cite{Vafa:1995gm}. Thus, the heterotic string gets compactified on an 
elliptically fibered K3 surface $Z_2$ while  F-Theory is compactified on an
elliptically fibered Calabi-Yau threefold $X_3$ over a Hirzebruch surface $\mathbb{F}_n$.
Our presentation will be brief and focused on the later applications in this work. For a more 
detailed discussion we refer to the classical reference 
\cite{Morrison:1996na,Morrison:1996pp,Friedman:1997yq} or the reviews 
\cite{Aspinwall:1996mn,Andreas:1998zf}.

On the F-theory side, the non-Abelian gauge content originates from the codimension one singularities of the elliptic fibration $\pi: X_3\rightarrow \mathbb{F}_n $. The singularity is generically of type $G'$, which gets broken down to $G \subset G'$ by monodromies corresponding to outer automorphisms of the Dynkin diagram of $G'$ \cite{Bershadsky:1996nh}. The resulting gauge symmetry is encoded in the order of vanishing of the coefficients $a_0, a_1, a_2, a_3, a_4, a_6$ in the Tate form of the elliptic fibration
\be
y^2 + a_1 x y +a_2 = x^3 + a_3 x^2 + a_4 x + a_6  \, .
\ee

In addition, we introduce the \textit{Tate vector} $\vec{t}_X$ which encodes the orders of vanishing of the coefficients $a_i$ along the divisor defined by the local coordinate $X$:
\be
\vec{t}_X = (\text{ord}_X \left(a_0 \right),\,\text{ord}_X \left(a_1 \right),\, \text{ord}_X \left(a_2 \right),\,\text{ord}_X \left(a_3 \right),\, \text{ord}_X \left(a_4 \right),\, \text{ord}_X \left(a_6 \right),\, ord_X \left(\Delta \right))\,.
\ee
The results of the analysis of singularities, known as Tate's algorithm, are summarized in Table 
\ref{Tatetable} \cite{tate1975algorithm,Bershadsky:1996nh}, see, however, \cite{Katz:2011qp} 
for subtleties. 

\begin{table}[ht]
\begin{center}
\begin{tabular}{|c|c|c|c|c|c|c|c|}\hline
Type & Group & $a_1$ & $a_2$ & $a_3$ & $a_4$ & $a_6$ & $\Delta$\\ \hline
$I_0$ & \{e\} & $0$ & $0$ & $0$ & $0$ & $0$ & $0$\\ \hline
$I_1$ & \{e\} & $0$ & $0$ & $1$ & $1$ & $1$ & $1$\\ \hline
$I_2$ & SU(2) & $0$ & $0$ & $1$ & $1$ & $2$ & $2$\\ \hline

$I_3$ & SU(3) & $0$ & $1$ & $1$ & $2$ & $3$ & $3$\\ \hline
$I_{2k}, k \geq 2$ & Sp($k$) & $0$ & $0$ & $k$ & $k$ & $2k$ & $2k$\\ \hline
$I_{2k+1}, k\geq 1$ & Sp(k) & $0$ & $0$ & $k+1$ & $k+1$ & $2k+1$ & $2k+1$\\ \hline
$I_{n}, n\geq 4$ & SU(n) & $0$ & $1$ & $\left[ \frac{n}{2} \right]$ & $\left[ \frac{n+1}{2} \right]$ & $n$ & $n$\\ \hline
$II$ & \{e\} & $1$ & $1$ & $1$ & $1$ & $1$ & $2$\\ \hline
$III$ & SU(2) & $1$ & $1$ & $1$ & $1$ & $2$ & $3$\\ \hline
$IV$ & Sp(1) & $1$ & $1$ & $1$ & $2$ & $2$ & $4$\\ \hline
$IV$ & SU(3) & $1$ & $1$ & $1$ & $2$ & $3$ & $4$\\ \hline
$I_0^{*}$ & G$_2$ & $1$ & $1$ & $2$ & $2$ & $3$ & $6$\\ \hline
$I_0^{*}$ & Spin(7) & $1$ & $1$ & $2$ & $2$ & $4$ & $6$\\ \hline
$I_0^{*}$ & Spin(8) & $1$ & $1$ & $2$ & $2$ & $4$ & $6$\\ \hline
$I_1^{*}$ & Spin(9) & $1$ & $1$ & $2$ & $3$ & $4$ & $7$\\ \hline
$I_1^{*}$ & Spin(10) & $1$ & $1$ & $2$ & $3$ & $5$ & $7$\\ \hline
$I_2^{*}$ & Spin(11) & $1$ & $1$ & $3$ & $3$ & $5$ & $8$\\ \hline
$I_2^{*}$ & Spin(12) & $1$ & $1$ & $3$ & $3$ & $5$ & $8$\\ \hline
$I_{2k-3}^{*}, k\geq 3$ & SO($4k+1$) & $1$ & $1$ & $k$ & $k+1$ & $2k$ & $2k+3$\\ \hline
$I_{2k-3}^{*}, k\geq 3$ & SO($4k+2$) & $1$ & $1$ & $k$ & $k+1$ & $2k+1$ & $2k+3$\\ \hline
$I_{2k-2}^{*}, k\geq 3$ & SO($4k+3$) & $1$ & $1$ & $k+1$ & $k+1$ & $2k+1$ & $2k+4$\\ \hline
$I_{2k-2}^{*}, k\geq 3$ & SO($4k+4$) & $1$ & $1$ & $k+1$ & $k+1$ & $2k+1$ & $2k+4$\\ \hline
$IV^{*}$ & F$_4$ & $1$ & $2$ & $2$ & $3$ & $4$ & $8$\\ \hline
$IV^{*}$ & E$_6$ & $1$ & $2$ & $2$ & $3$ & $5$ & $8$\\ \hline
$III^{*}$ & E$_7$ & $1$ & $2$ & $3$ & $3$ & $5$ & $9$\\ \hline
$II^{*}$ & E$_8$ & $1$ & $2$ & $3$ & $4$ & $5$ & $10$\\ \hline
non-min & - & $1$ & $2$ & $3$ & $4$ & $6$ & $12$\\ \hline
\end{tabular}
\end{center}
\caption{Results from Tate's algorithm.} \label{Tatetable}
\end{table}

On the heterotic side, the gauge theory content is encoded in a vector bundle $V$ where the following discussion restricts itself to the case of $\text{SU}(N)$ bundles. 
The six-dimensional bundle is defined in terms of two pieces of data, the spectral cover curve 
$C$ as well as a line bundle $\mathcal{N}$ which is defined on $C$.  Here, the spectral curve 
$C$ is the 6d analog of the points defined by the section of $\mathbb{P} H^0(E, \mathcal{O}(np))$ which has been discussed in \ref{SUNellipticcurve}. In six dimensions, the elliptic curve 
$Z_1\cong E$ gets promoted to an elliptic fibration, which can again be described by a Weierstrass form
\eqref{eq:WSF} with coordinates $x$, $y$,$ z$ being sections of 
$\mathcal{L}^2$, $\mathcal{L}^2$, $\mathcal{O}$, respectively, for 
$\mathcal{L} = K_{\mathbb{P}^1}^{-1}=\mathcal{O}(-2)$ and
coefficients $f$, $g$ being in $\mathcal{L}^4$, $\mathcal{L}^6$, respectively. Accordingly, the 
coefficients $c_i$ entering the spectral cover  \eqref{spectralcover} are now sections of 
$\mathcal{M} \otimes \mathcal{L}^{-i}$, $\mathcal{M}$ being an arbitrary line bundle on 
$\mathbb{P}^1$ and $C$ is defined as the zero locus of the section of \eqref{spectralcover}.
Thus, $C$ defines an $N$-sheeted ramified covering of $\mathbb{P}^1$, i.e.~a Riemann surface. 
The spectral cover
$C$ defines the isomorphism class of a semi-stable vector bundle above each fiber. The line 
bundle $\mathcal{N}$ describes the possibility to twist the vector bundle without changing its 
isomorphism class. It is usually fixed, up to a twisting class $\gamma$, by the condition $c_1(V)=0$ for an $\text{SU}(N)$ bundle, see \cite{Friedman:1997yq} for more details.

\subsection{Massless U(1)-factors in heterotic/F-theory duality}
\label{sec:Stueckelberg}

As previously discussed, the perturbative heterotic gauge group is obtained by commuting the 
structure group $H$ of the vector bundle $V$ within the two $\text{E}_8$-bundles. 
We propose  three possibilities, how U(1) 
gauge groups can arise from this perspective:
\begin{itemize}
\item $H$ contains a U(1) factor, i.e.~it is of the form $H=H_1 \times  \text{U}(1)$, or $\text{S}(\text{U}(M)\times \text{U}(1))$,
\item  $H$ contains a discrete piece, i.e.~a part taking values in $\mathbb{Z}_k$,
\item or $H$ is non-Abelian and is embedded such that its centralizer in $\text{E}_8$ 
necessarily contains a U(1)-symmetry. 
\end{itemize}
The construction of a vector bundle for  these three different cases employing spectral
covers has been discussed in  Section \ref{SUNellipticcurve}. 

In general, we emphasize that  U(1)-factor which arises from a split spectral cover is usually 
massive due to a St\"uckelberg mass term which is induced by the first Chern class of the U(1) 
background bundle, as we review next. However, if the U(1) term originates from a background 
bundle with non-Abelian structure group there is tautologically no U(1) background factor which 
could produce a mass term and therefore the six-dimensional U(1) field is expected to be 
massless. Finally, we propose, for consistency with heterotic/F-theory duality, 
that a six-dimensional torsional section gives rise to a point-like 
instanton with discrete holonomy, as introduced in \cite{Aspinwall:1998xj}.
Indeed, we will show in several examples in Section \ref{section:examples} that 
all three cases naturally appear in heterotic duals of F-theory compactifications with one U(1) and that a
matching of the corresponding gauge groups is only possible if the arising spectral covers
are interpreted as suggested here.

\subsubsection{The heterotic St\"uckelberg mechanism} \label{massU(1)}

In six and lower dimensions, it is well-known that a geometric St\"uckelberg effect can render a 
U(1) gauge field massive \cite{green2012superstring}. To identify the mass term of the 
six- (or lower-) dimensional U(1), one considers the modified ten-dimensional kinetic term of the Kalb-Ramond 
field $B_2$ which reads, up to some irrelevant proportionality constant, as
\be
\mathcal{L}^{10d}_{\text{kin}} = H \wedge \star_{10d} H, \qquad H= dB_2 - \frac{\alpha'}{4} \left(\omega_{3Y} (A) - \omega_{3L} (\Omega) \right) \,.\label{10dfieldstrength}
\ee
Here, $\star_{10d}$ is the ten-dimensional Hodge-star and $\omega_{3Y}$, $\omega_{3L}$ 
denote the Chern-Simons terms of the gauge field and 
the spin connection, respectively. The  physical effect we want to discuss here arises from the 
former one, which is given explicitly by
\be
\omega_{3Y} = \text{Tr} \left(A \wedge d A + \frac{2}{3} A \wedge A \wedge A \right)\, .
\ee

Now, we perform a dimensional reduction of the kinetic term \eqref{10dfieldstrength} in
the background of a U(1) vector bundle on the heterotic compactification manifold $Z_n$, 
ignoring possible additional non-Abelian vector bundles for simplicity.
On such a background, we can expand the ten-dimensional field strength $F^{10d}_{U(1)}$ 
of the U(1) gauge field as 
\beq
F^{10d}_{U(1)} = F_{U(1)} + \mathcal{F}=F_{U(1)}+k_\alpha\omega^\alpha\,.
\eeq
 Here $\mathcal{F}=\frac{1}{2\pi i}c_1(\mathcal{L})$ is the background field strength, i.e.~the first 
Chern class $c_1(\mathcal{L})$ of the corresponding U(1) line bundle $\mathcal{L}$, and
$F_{U(1)}$ is the lower-dimensional gauge field.  
We have also introduced a basis $\omega^\alpha$, $\alpha=1,\ldots, b_2(Z_n)$,
of harmonic two-forms in $H^{(2)}(Z_n)$, where $b_2(Z_n)$ is the second Betti number of 
$Z_n$, along which we have expanded $\mathcal{F}$ into the flux quanta $k_\alpha$.  
We also expand the ten-dimensional Kalb-Ramond field as
\beq
 B_2=b_2+\rho_\alpha\omega^\alpha \,,
\eeq
where $b_2$ is a lower-dimensional two-form and $\rho_\alpha$ are lower-dimensional axionic
scalars .
We readily insert this reduction ansatz into the ten-dimensional field strength $H$ in 
\eqref{10dfieldstrength}, where we only take into account the gauge part, to arrive, dropping 
unimportant prefactors, 
at the lower-dimensional kinetic term for the axions $\rho_\alpha$ of the form
\be \label{eq:LStueck}
\mathcal{L}_{\text{St\"uck.}} = G^{\alpha\beta}\left(d\rho_\alpha + k_\alpha A_{U(1)} \right)\wedge \star\left(d\rho_\beta + k_\beta A_{U(1)} \right)\, .
\ee
Here we introduced the kinetic metric
\beq
	G^{\alpha\beta}=\int_{Z_n}\omega^\alpha\wedge \star \omega^\beta\,.
\eeq

It is clear from \eqref{eq:LStueck} that a single U(1) gauge field will be massive if we have a 
non-trivial $c_1(\mathcal{L})\neq 0$. However, we note that in the presence of multiple massive 
U(1) gauge fields, appropriate linear combinations of them in the kernel of the mass matrix can 
remain massless U(1) fields. A computation similar to the one above has appeared in 
e.g.~\cite{Blumenhagen:2005ga}, where also the case of multiple U(1)'s is systematically 
discussed.

\subsubsection{U(1)-factors from gluing conditions in half K3-fibrations}
\label{gluing}

We conclude this section by discussing the connection between 
the previous field theoretic considerations that lead to a  
massive U(1) via the St\"uckelberg action \eqref{eq:LStueck} on the heterotic side
and geometric glueing conditions of the sections of half K3 surfaces to global sections of the
two half K3-fibrations $X_n^\pm$ that arise in stable degeneration as well 
as of the full Calabi-Yau manifold $X_n$. 
We illustrate this in six dimensions 
for concreteness, i.e.~for F-theory on a Calabi-Yau threefold $X_3$ and the heterotic string on a 
K3 surface $Z_2$, although the arguments hold more generally.

It is well known that the number of U(1) factors in F-theory is given by the rank of the 
Mordell-Weil group, i.e.~by the number of independent global rational sections of the elliptic 
fibration $X_3$ in addition to the zero section. As discussed in Section \ref{sec:hetF8d}, 
a half K3 surface with ADE singularity of rank $r$ has an $(8-r)$-dimensional Mordell-Weil group. 
Promoting the half K3 surface to a fibration of half K3 surfaces over the base $\mathbb{P}^1$,
such as the threefolds $X_3^\pm$, 
these sections need 
not necessarily give rise to sections of the arising three-dimensional elliptic fibrations. 
Considering the half K3 surfaces arising in the stable degeneration limit of F-theory, there
are those sections which also give rise to sections of e.g.~the full half K3 fibration $X_3^+$.
These sections will induce a U(1)-factor on the heterotic side which is embedded into one 
$\text{E}_8$-bundle and which is generically massive with a mass arising via the St\"uckelberg action
\eqref{eq:LStueck}. If there is also a globally well-defined section of 
the  other half K3 fibration $X_3^-$ and this section glues
with the section in the first half K3 fibration $X^+_3$, then there is a linear combination of 
U(1)'s that remains massless in the St\"uckelberg mechanism  on the heterotic side.   
This is clear from the F-theory perspective, as these two sections can then be glued  along the 
heterotic two-fold $Z_2$ to a section of the full Calabi-Yau threefold $X_3$, i.e.~give rise to an 
element in its Mordell-Weil group and a massless U(1).

\section{Dual Geometries with Toric Stable Degeneration} \label{verticalmodel}

In this section, we describe a toric method in order to study the stable degeneration limit of an elliptically fibered K3 surface.
This stable degeneration limit will be at the heart of the analysis of the examples of heterotic/F-theory dual geometries in Section \ref{section:examples}.
In a first step in Section \ref{sec:K3surface}, we construct an elliptically fibered K3 surface. Afterwards in Section \ref{sec:K3surfacefibrations}, we fiber this K3 surface over another $\mathbb{P}^1$ which is used to investigate the splitting of the K3 surface into two rational elliptic surfaces, as discussed in Section \ref{subsection:stabledegenerationlimit}. In the concluding Section \ref{sec:proof12K3}, we prove that the surfaces arising in the stable degeneration of the K3 surface indeed 
define rational elliptic surfaces, i.e.~half K3 surfaces.

\subsection{Constructing an elliptically fibered K3 surface}
\label{sec:K3surface}

We start by constructing a three-dimensional reflexive polytope $\Delta_3^\circ$ given as the 
convex hull of vertices that are the rows of the following matrix:
\be
\left(\begin{array}{ccc|c}
a_1 & b_1 & 0 & x_1\\
\vdots & & 0 & x_i \\
a_n & b_n & 0 & x_n \\
0 & 0& 1 & U \\
0 & 0 & -1 & V \\
\end{array}\right)\, .
\ee
Here $(a_i\, b_i)$ denote the points of a two-dimensional reflexive polytope $\Delta_2^\circ$, which will specify the geometry of
the elliptic fiber $E$. It is embedded into $\Delta_3^\circ$ in the $xy$-plane, see the first picture in
Figure \ref{fig:Delta3}.  
The last column contains the  homogeneous coordinate associated to a given vertex.
We label the rays of the two-dimensional polytope counter-clockwise by the coordinates $x_1,\dots x_n$.
In addition, we assign the coordinates $U$, $V$ to the points $(0\, 0\, 1)$ and $(0\, 0\, -1)$ which correspond to 
the rays of the fan of the $\mathbb{P}^1$-base. We use the shorthand notation $\mathbb{P}^1_{[U:V]}$ to indicate its homogeneous coordinates. Finally, we use the notation $\rho_H$ for the ray 
with corresponding homogeneous coordinate $H$. We denote by $\Sigma_{3}$ the natural simplicial fan associated to 
$\Delta_3^\circ$ and denote the corresponding toric variety over the fan of $\Delta^\circ$ as $\mathbb{P}_{\Sigma_3}$. Provided  
a fine triangulation of the polytope $\Delta_3^\circ$ has been chosen, the toric ambient space 
$\mathbb{P}_{\Sigma_3}$ will be Gorenstein and terminal.  

A general section $\chi$ of the anti-canonical bundle 
$\mathcal{O}_{\mathbb{P}_{\Sigma_3}}(-K_{\mathbb{P}_{\Sigma_3}})$ defines a smooth elliptically fibered  
K3 surface $X_2$. The ambient space of its elliptic fiber $E$ is the toric variety 
$\mathbb{P}_{\Sigma_2}$ that is constructed from the fan $\Sigma_2$ of the polytope 
$\Delta_2^\circ$ induced by $\Sigma_3$. 
As the toric fibration of $\Sigma_2$ over $\Sigma_{\mathbb{P}^1}$ is direct, the section $\chi$ takes the form 
\be
\chi= { s}_i \eta^i\, \qquad \text{for}\qquad {s}_i = { s}_i^0 U^2 +  { s}_i^1 UV + { s}_i^2 V^2. \label{CYconstraint}
\ee
Here $\eta^i$ are the sections of the anti-canonical bundle of 
$\mathcal{O}_{\mathbb{P}_{\Sigma_2}}(-K_{\mathbb{P}_{\Sigma_2}})$, i.e.~the range of the index $i$ is given by the number of integral points in $\Delta_2$,  
and $s_i^k$, $k=1,2,3$, are constants.
Note that, for a very general\footnote{A point is very general if it lies outside a countable union of closed subschemes of positive codimension} $X_2$,  the dimension $h^{(1,1)}(X)$, 
of the cohomology group $H^{(1,1)}(X_2,\mathbb{C})$  can be computed combinatorically from the pair of reflexive polyhedra $\Delta_3$, $\Delta_3^\circ$ by a 
generalization of the Batyrev's formula \cite{2010arXiv1011.1003B}:
\begin{eqnarray}
h^{(1,1)}(X) &=& l(\Delta^\circ) -n - 1 - \sum_{\Gamma^\circ} l^*(\Gamma^\circ) + \sum_{\Theta^\circ} l^*(\Theta^\circ)l^*({\hat \Theta}^\circ) \, . \label{homologydimension}
\end{eqnarray}
Here $l(\Delta)$  ($l^*(\Delta)$) denote the number of (inner) points of the $n$-dimensional polytope $\Delta$. In addition, $\Gamma$ ($\Gamma^\circ$) denote the codimension one faces of $\Delta$ ($\Delta^\circ$), while $\Theta$ denotes a codimension two face with $\hat \Theta$ being its dual.
\begin{figure}[h!]
\begin{center}
\includegraphics[scale=0.22]{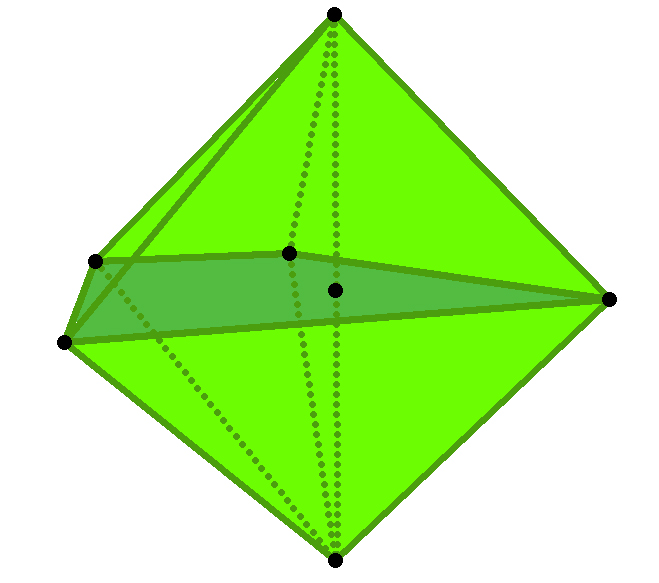}\hspace{1cm}
\includegraphics[scale=0.3]{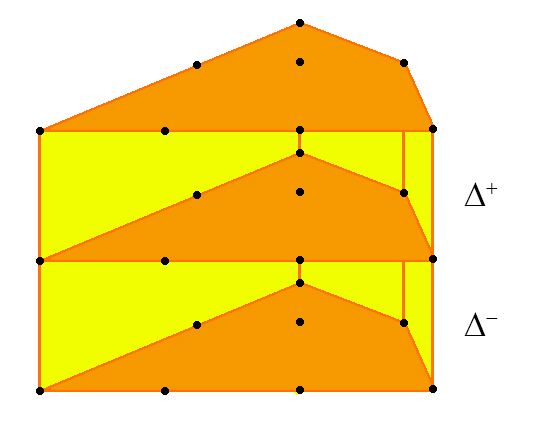}
\end{center}
\caption{On the left we show the reflexive polytope $\Delta_3^\circ$, while its dual $\Delta_3$ is shown on the right. In this example, the ambient space for the elliptic fiber, specified by $\Delta_2^\circ$, is given by $\text{Bl}_1 \mathbb{P}^{(1,1,2)}$.}
\label{fig:Delta3}
\end{figure}
  
\subsection{Constructing K3 fibrations}
\label{sec:K3surfacefibrations}

As a next step, we fiber this ambient space over a second $\mathbb{P}^1_{[\lambda_1,\lambda_2]}$ with homogeneous 
coordinates $\lambda_1$, $\lambda_2$. 
The following construction is such that the generic fiber consists of a smooth K3 surface $X_2$ over a generic 
point of $\mathbb{P}^1_{[\lambda_1,\lambda_2]}$ and a split fiber, i.e.~a splitting into two half K3 
surfaces, over a distinguished point of $\mathbb{P}^1_{[\lambda_1,\lambda_2]}$, as explained below.

The four-dimensional polytope which describes this construction is given by 
\be
\Delta_4 = \left\{ \left(m_1, m_2, m_3, m_4\right) \in \mathbb{Z}^4 \mid \left(m_1, m_2, m_3\right) \in \Delta_3,\, \, -1 \le m_4 \le 1, \begin{cases} m_4 \ge -1 \,\,\, \text{if} \,\,\, m_3 \le 0\,, \\ m_4 \ge m_3 -1 \,\,\, \text{if} \,\,\, m_3 \ge 0\,.  \end{cases}\hspace{-0.3cm} \right\}\,.
\ee
Here, $\Delta_3$ denotes the dual polytope of $\Delta_3^\circ$, cf.~the second picture in Figure 
\ref{fig:Delta3}. 
The faces of $\Delta_4$ are given by the (intersection of the) hyperplanes  
\be
m_4=1\,, \qquad m_4=-1\,, \qquad m_4=-1+m_3\,, \qquad m_3=-1\,, \quad m_3=1\,, \qquad \sum_{j=0}^2 a_i^j m_j=1\, ,
\ee
where the last expression is given by the the defining hyperplanes of $\Delta_2$, the dual of 
$\Delta_2^\circ$. We denote by $\Sigma_4$ the fan associated to the dual polytope $\Delta^\circ_4$ of $\Delta_4$. In particular, the normal vectors of the facets of $\Delta_4$ give the rays of 
$\Sigma_4$. To be explicit, the rays of $\Sigma_4$ are given by the rows of the matrix
\be
\left(\begin{array}{cccc|c}  a_1 & b_1 & 0 &0 & x_1\\
\vdots & & 0&0 & x_i \\
a_n & b_n & 0&0 & x_n \\ 0&0&1&0& U \\ 0&0&0&1 &\lambda_1 \\ 0&0&-1&1 & \mu \\ 0&0&-1&0& V \\ 0&0&0&-1 & \lambda_2  \end{array} \right) \label{Deltatoricpoints}\, .
\ee
We note that the coordinates assigned to its rays as displayed in \eqref{Deltatoricpoints} transform as follows under the $\mathbb{C}^*$-actions
\begin{eqnarray} \label{scalingrelationsdP2}
\left(U: \lambda_1 : \mu : V: \lambda_2 \right) \sim \left(a^{-1}U: ab^{-1} \lambda_1 : a^{-1} bc^{-1}\mu :b^{-1}c V:c^{-1}\lambda_2 \right) 
\end{eqnarray}
with  $a,b,c \in \mathbb{C}^*$.

In analogy to the discussion in the previous section, a section $\chi_4$ of the  anti-canonical bundle 
$-K_{\mathbb{P}_{\Sigma_4}}$  of the toric variety $\mathbb{P}_{\Sigma_4}$
defines a three-dimensional smooth Calabi-Yau manifold $\mathfrak{X}$. In particular, the Calabi-Yau constraint \eqref{CYconstraint} generalizes as
\be
\chi_4 = s_i \eta^i,  
\ee
where the $\eta^i$ are given as before and the coefficients $s_i$  now read
\begin{eqnarray} \label{coefficients4dCYconstraint}
s_i (U,V, \lambda_1, \lambda_2, \mu) &=&s_i^1\lambda_1 \lambda_2 U^2 + s_i^2 \lambda_1^2 \mu U^2 +s_i^3 \lambda_2^2 U V + s_i^4
  \lambda_1 \lambda_2 \mu U V + s_i^5 \lambda_1^2 \mu^2 U V \nn \\ && + s_i^6
  \lambda_2^2 \mu V^2 + s_i^7 \lambda_1 \lambda_2 \mu^2 V^2 + s_i^8
  \lambda_1^2 \mu^3 V^2 \,  
\end{eqnarray}
with constants $s_i^j \in \mathbb{C}$.

We proceed by observing that the projection on the last two columns in \eqref{Deltatoricpoints}  yields the polytope $\Delta_{dP_2}^\circ$ 
of the toric variety $dP_2$, cf.~Figure \ref{fig:dP2}. 
\begin{figure}[t!]
\begin{center}
\includegraphics[scale=0.36]{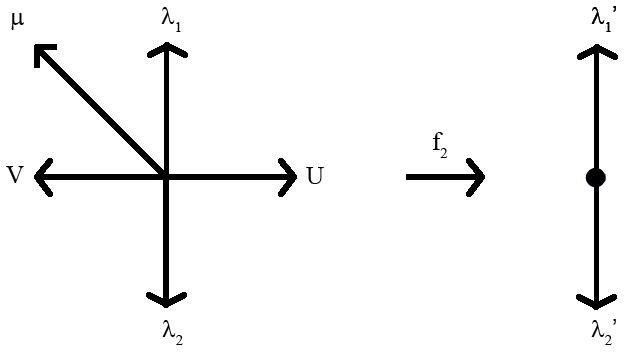}
\end{center}
\caption{The toric morphism $f_2$.}
\label{fig:dP2}
\end{figure}
Denoting the fan of $\Delta_{dP_2}^\circ$  by $\Sigma_{dP_2}$ this projection gives 
rise to a toric map
\be
f_1: \, \,\,\Sigma_4 \longrightarrow \ \Sigma_{dP_2}\,. 
\ee
In addition, $dP_2$ is fibered over the $\mathbb{P}^1_{[\lambda_1':\lambda_2']}$ as can be seen by projecting onto the fourth column of $\Delta_4$, cf.~Figure \ref{fig:dP2}, i.e.~there is a toric map
\be
f_2:\,\,\, \Sigma_{dP_2} \longrightarrow \Sigma_{\mathbb{P}^1}\,, \label{mapf2}
\ee
where $\Sigma_{\mathbb{P}^1}$ is the fan of $\mathbb{P}^1_{[\lambda_1':\lambda_2']}$. Note that this $\mathbb{P}^1$ is isomorphic to $\mathbb{P}^1_{[\lambda_1:\lambda_2]}$.
We denote the composition map of the two by  $f=f_2 \circ f_1$.

In summary, we have the following diagram of toric morphisms and induced maps on $\mathfrak{X}$:
\be
\xymatrixcolsep{1.8pc}\xymatrix{
\mathbb{P}_{\Sigma_4} \ar@{->}[rr]^f \ar@{->}[rd]^{f_1} \ar@{->}[rr]^f && \mathbb{P}^1_{[\lambda_1':\lambda_2']} \\&dP_2 \ar@{->}[ru]^{f_2} &\\&dP_2 \ar@{->}[rd]^{\pi_2} &\\
\mathfrak{X} \ar@{^{(}->}[uuu] \ar@{->}[ur]^{\pi_1} \ar@{->}[rr]^{\pi} && \mathbb{P}^1_{[\lambda_1':\lambda_2']} \ar@{->}[uuu]^{\cong}
}\nn
\ee
Here we denote the toric maps $f_1$, $f_2$, $f$ and their induced morphisms of toric 
varieties by the same symbol, respectively. Note that for a generic point, the fiber of $\pi$ is given by a smooth K3 surface $X_2$. 

In order to prepare for the discussion of the stable degeneration limit, we proceed by discussing the fibration map in more detail. 
For this purpose, we note the correspondence of facets and rays as displayed in Table \ref{tablefacetsrays}.
\begin{table}[t] 
\begin{center}
\begin{tabular}{|c|c|c|}
ray & facet & constraint \\
$\rho_{\lambda_1}$ & $m_4=-1$ & $s_{\lambda_1} = s_i^3 U + s_i^6 \mu$ \\
$\rho_{\mu}$ & $m_4= m_3-1$ & $s_{\mu} = s_i^1 \lambda_1 + s_i^3 V$ \\
$\rho_{\lambda_2}$ & $m_4 =1$ & $s_{\lambda_2} = s_i^2U^2 + s_i^5 UV + s_i^8 V^2$
\end{tabular}
\end{center}
\caption{The correspondence between the rays of $\Delta_{dP_2}^\circ$ and the facets of $\Delta_{dP_2}$. The last column displays the global sections that embed the associated divisor into $\mathbb{P}^1$ and $\mathbb{P}^2$, respectively. The coefficients on the right-hand side refer to equation \eqref{coefficients4dCYconstraint}.}
\label{tablefacetsrays}
\end{table} 
The dual $\Delta_{dP_2}$ of $\Delta_{dP_2}^\circ$ with associated monomials is shown in Figure 
\ref{figure2ddiagram}. These monomials are the global sections of $K_{dP_2}$ and are constructed according to \cite{Batyrev:1994hm}
\be
\chi_{dP_2} = \sum_{P  \in \Delta_{dP_2}} \prod_{P^* \in \Delta_{dP_2}^*} a_P x_{P^*}^{\bra P, P^* \ket + 1}\,.
\ee
Here $x_{P^*}$ denotes the coordinate which is associated to the corresponding ray of the toric diagram and $a_P$ are 
constants.
\begin{figure}[H]
\begin{center}
\includegraphics[width=6cm]{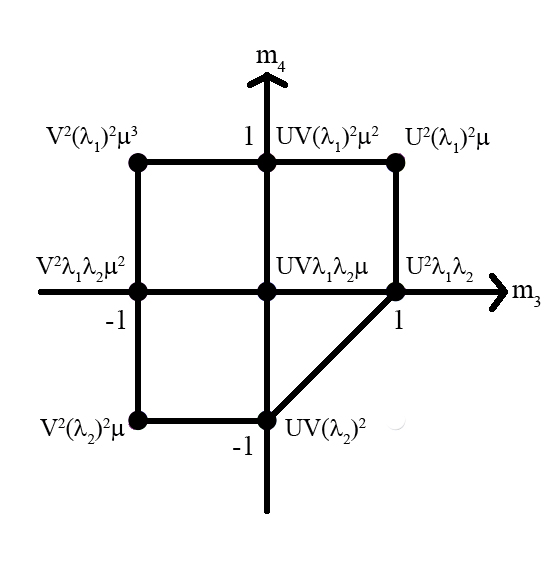}
\caption{The dual polytope $\Delta_{dP_2}$ and the associated monomials.}
 \label{figure2ddiagram}
\end{center}
\end{figure}
By the correspondence between cones of $\Delta^\circ_{dP_2}$ and 
vertices of $\Delta_{dP_2}$ the vertex corresponding to the monomial $V^2 \lambda_2^2 \mu$ 
is dual to the cone spanned by the rays $\rho_U$ and $\rho_{\lambda_1}$.
We denote the coordinates associated to the two rays of 
$\mathbb{P}^1_{[\lambda_1:\lambda_2]}$ inside $\Delta_{dP_2}^0$ appearing in \eqref{mapf2} by 
$\rho_{\lambda_1'}$ and  $\rho_{\lambda_2'}$. 
Note that  $f_2^{-1}(\rho_{\lambda_1'}) = \{\rho_{\lambda_1}, \rho_{\mu}\}$, while 
$f_2^{-1}(\rho_{\lambda_2'}) = \{\rho_{\lambda_2}\}$. 

\subsection{The toric stable degeneration limit} 
\label{subsection:stabledegenerationlimit}

In the following, we aim to show that the general fiber of the map $\pi$ gives rise to a smooth K3 surface while the pre-image of 
the point $[\lambda_1': \lambda_2']=[1:0]$ gives rise
to a degeneration into two half K3 surfaces $X_{2}^\pm$ that intersect in the elliptic fiber $Z_1$ 
over the point of intersection of the two $\mathbb{P}^1$ which are the respective bases of their elliptic fibrations.

Let us first consider the toric variety $f_2^{-1}(\lambda_2'=0)$ 
corresponding to the pre-image in $\Delta_{dP_2}^0$ of  
$\rho_{\lambda'_2}$. It is given by the star of $\rho_{\lambda_2}$ in $\Delta_{dP_2}^0$ which is 
just the generic fiber of $f_2$. Indeed, if $\lambda_2=0$, the coordinates $\mu$ and $\lambda_1$ are non-vanishing due to the Stanley-Reisner ideal. Two of the scaling relations \eqref{scalingrelationsdP2} can be used in order to eliminate the latter two variables while the remaining (linear combination) endows the coordinates $U$, $V$ with the well-known scaling relations of $\mathbb{P}^1_{[U,V]}$. In addition, the  monomials associated to the vertices of the dual facet of $\rho_{\lambda_2}$ give rise to the following sections
\be \label{explicitformsi}
s_{\lambda_2} := s_i^2U^2 + s_i^5 UV + s_i^8 V^2\, ,
\ee 
as follows from \eqref{coefficients4dCYconstraint} by setting $\lambda_2=0$.
These provide precisely the global sections of $O_{\mathbb{P}^1}(2)$ that are needed for the Veronese embedding, i.e.~the embedding of $\mathbb{P}^1_{[U,V]}$ into $\mathbb{P}^2$ as a conic 
\be
[U:V] \longmapsto [U^2: UV: V^2]\, . \label{embeddingP1}
\ee

In contrast, the preimage of $\rho_{\lambda_1'}$ consists of the two divisors $\lambda_1=0$ and 
$\mu=0$. In this case the Stanley-Reisner ideal forbids the vanishing of the coordinates $V$, $\lambda_2$ and $U$, 
$\lambda_2$ respectively. Taking again into account the scaling relations \eqref{scalingrelationsdP2}, one observes that the 
pre-image of the divisor $\lambda_1'=0$ consists of two $\mathbb{P}^1$'s that are given by
\be
D_{\lambda_1} = [U:0:\mu:1:1], \qquad D_{\mu} = [1:\lambda_1:0:V:1]\, .
\ee
These intersect in precisely one point given by $[1:0:0:1:1]$. One identifies the dual facets of $\rho_{\lambda_1}$ and $\rho_{\mu}$ as 
 $m_4=-1$ and $m_4=m_3-1$. In this case the global sections are given by
\beq
s_{\lambda_1} := s_i^3 U + s_i^6 \mu\,,\qquad 
s_{\mu} := s_i^1 \lambda_1 + s_i^3 V\,, \label{coefficientshalfK3}
\eeq
as follows again from \eqref{coefficients4dCYconstraint}.
This induces in this case only the trivial embedding via the identity map. 
Note that the union of the two divisors $D_{\lambda_1}$ and $D_{\mu}$ is given by a degenerate conic
\be
z_1 z_3 = z_2^2 \lambda_1 \mu\,, \qquad \text{with}\qquad  (V^2 \mu, UV, U^2 \lambda_1) \mapsto [z_1: z_2:z_3] \in \mathbb{P}^2\,,
\ee
which splits as just observed into the two lines $z_1=0$, $z_3=0$ at $\lambda_1=0$ and $\mu=0$.

A similar reasoning applies to the pre-image of $\rho_{\lambda_1'}$ under the composite map $f$.
As noted above, we have $f^{-1}(\rho_{\lambda_1'})=\{\rho_\mu,\rho_{\lambda_1}\}$, which implies that 
the pre-image is given by the two divisors  $\mathbb{P}_{\Sigma_3^+}=\{\mu=0\}$ and 
$\mathbb{P}_{\Sigma_3^-}=\{\lambda_1=0\}$ in $\mathbb{P}_{\Sigma_4}$.
They are obtained as the star of $\rho_{\mu}$ and $\rho_{\lambda_1}$ in $\Delta^0_4$, 
respectively, with their fans $\Sigma_{3}^{\pm}$ induced by $\Sigma_4$. The corresponding
respective dual facets are given by the three-dimensional facets $m_4=-1$ and $m_4= m_3-1$ in 
$\Delta_4$.
In addition, this gives rise to a splitting of $\Delta_3$ as
\begin{eqnarray}
{\Delta}_3^+ &=& \left\{\left(m_1,m_2,m_3 \right) \in \mathbb{Z}^3  \mid \left(m_1, m_2 \right) \in \Delta_2,\,\,\, m_3 \in \{0,1 \}  \right\} \, ,\\
{\Delta}_3^- &=& \left\{\left(m_1,m_2,m_3 \right) \in \mathbb{Z}^3  \mid \left(m_1, m_2 \right) \in \Delta_2,\,\,\, m_3 \in \{-1,0 \}  \right\}\, ,
\end{eqnarray}
which is also referred to as the top and bottom splitting \cite{Candelas:1996su}, c.f.~Figure \ref{fig:Delta3}. Thus, the section of the anti-canonical bundle $\mathcal{O}(-K_{\mathbb{P}_{\Sigma_3}})$ in \eqref{CYconstraint} in the limit becomes the sum of 
\begin{eqnarray}
\chi_{X_2^+} &:=& s^+_i \eta^i\,,\qquad \text{with}\qquad s^+_i = {s^+}_i^0 V +  {s^+}_i^1 \lambda_1\, ,  \nn \\
\chi_{X_2^-} &:=& s^-_i \eta^i\,,\qquad \text{with}\qquad s^-_i = {s^-}_i^0 U +  {s^-}_i^1 \mu\, ,
\end{eqnarray}
so that we can define the two surfaces $X_2^\pm$  as
\be
X_2^+ = X_2\vert_{\mathbb{P}_{\Sigma_3^+}}=\Big\{\chi =\mu=0\Big\}\,,\qquad  X_2^- = X_2\vert_{\mathbb{P}_{\Sigma_3^-}}=\Big\{\chi =\lambda_1=0\Big\} .\label{definitionXpm}
\ee
As we will prove in the next subsection, $X_2^+$  and  $X_2^-$  are two rational elliptic surfaces (half K3 surfaces).

In contrast, the pre-image of $\rho_{\lambda_2'}$ is given by the whole three-dimensional fan 
$\Sigma_3$ as it is also for a generic point in 
$\mathbb{P}^1_{[\lambda_1:\lambda_2]}$. To justify the latter statement inspect the fiber above the 
origin $0$ of the fan $\Sigma_{\mathbb{P}^1}$ corresponding to a 
generic point in $\mathbb{P}^1_{[\lambda_1:\lambda_2]}$.

Finally, we remark that the two rational elliptic surfaces $X_2^+$, $X_2^-$ that arise at the loci $\{\mu=0\}$ and $\{\mu=0\}$, 
respectively, are independent of the K3 surface which appears over the locus $\{\lambda_2=0\}$. In the following, we explain how 
the half K3 surfaces can be obtained from the data of the K3 surface $X_2$ directly. As the notation used so far is rather heavy, 
which is unfortunately necessary, we introduce a slightly easier notation that will be used in the discussion of explicit examples in 
section \ref{section:examples}. We rewrite a general hypersurface constraint as
\be
\chi= {s}_i \eta^i, \qquad {s}_i = {s}_{i1}^0 U^2 +  {s}_{i2}^1 UV + {s}_{i3}^2 V^2\, , \label{newconvention}
\ee
which requires, depending on the situation at hand, the  following identifications between the coefficients of 
\eqref{newconvention} and of \eqref{coefficients4dCYconstraint}:
\begin{eqnarray}
s_{i1} \equiv s_i^2\, , \qquad s_{i2} \equiv  s_i^5 \, , \qquad s_{i3} \equiv s_i^8\, ,\nn\\
\text{or}\qquad s_{i1} \equiv s_i^1\,,\qquad
s_{i2} \equiv s_i^3\,,\qquad 
s_{i3}  \equiv s_i^6 \,.
\end{eqnarray}
However, it is crucial to note that the pair of coordinates $U,V$ is only suited to describe the base $\mathbb{P}^1$ of the K3 
surface $X_2$, 
while the base coordinates $\mathbb{P}^1$'s of $X_2^+$ and $X_2^-$ are given by $\lambda_1$, $V$ and $U$, $\mu$, 
respectively. 

\subsection{Computing the canonical classes of the half K3 surfaces $X_2^\pm$}
\label{sec:proof12K3}

In this subsection, we discuss how the half K3 surfaces $X_2^\pm$ can be re-discovered in the toric stable degeneration limit. Note that  the two components $\mathbb{P}_{\Sigma_3^+}$ and $\mathbb{P}_{\Sigma_3^-}$ of 
the degenerate fiber, as divisors in $\mathbb{P}_{\Sigma_4}$, should equal the generic fiber  
$\mathbb{P}_{\Sigma_3}$:
\be \label{eq:P++P-=P}
\mathbb{P}_{\Sigma_3^+} + \mathbb{P}_{\Sigma_3^-} \cong \mathbb{P}_{\Sigma_3}\,.
\ee
In addition, we have 
\be \label{eq:P+-P=0}
\mathbb{P}_{\Sigma_3} \cdot\mathbb{P}_{\Sigma_3^\pm} = 0
\ee
 as
the generic fiber can be moved away from the locus $\lambda_1'=0$, cf.~Figure 
\ref{fig:dP2}.
This allows us to compute the canonical bundle of $\mathbb{P}_{\Sigma_3^\pm}$ using adjunction in
$\mathbb{P}_{\Sigma_4}$ as
\be \label{eq:K_PSigma3pm}
K_{\mathbb{P}_{\Sigma_3}^\pm} = \left.\big(K_{\mathbb{P}_{\Sigma_4}} \otimes \mathcal{O}_{\mathbb{P}_{\Sigma_4}}(\mathbb{P}_{\Sigma_3^\pm})\big) \right\vert_{\mathbb{P}_{\Sigma_3}^\pm} = \left.K_{\mathbb{P}_{ \Sigma_4}}\right\vert_{\mathbb{P}_{\Sigma_3}^\pm} \otimes \mathcal{O}_{\mathbb{P}_{\Sigma_3}^\pm} \Big(-\mathbb{P}_{\Sigma_3^\pm} \cdot \mathbb{P}_{\Sigma_3^{\mp}} \Big) \, ,
\ee
where we used  \eqref{eq:P++P-=P} and \eqref{eq:P+-P=0}.
Note that the divisor corresponding to the last term equals the class of the ambient space 
$\mathbb{P}_{\Sigma_2}$ of elliptic fiber of $X_2$, 
i.e.~$\mathbb{P}_{\Sigma_3^+} \cdot\mathbb{P}_{\Sigma_3^-} = \mathbb{P}_{\Sigma_2}$.
Making one more time use of the adjunction formula, one finally arrives at
\ba
K_{X_2^\pm} &=& \left.\Big( K_{\mathbb{P}_{\Sigma_3^\pm}} \otimes \mathcal{O}_{\mathbb{P}_{\Sigma_3^\pm}} (X_2^\pm) \Big)\right\vert_{X_2^\pm} =\left.\Big(\left.K_{\mathbb{P}_{ \Sigma_4}}\right\vert_{\mathbb{P}_{\Sigma_3}^\pm} \otimes \mathcal{O}_{\mathbb{P}_{\Sigma_3}^\pm} \big(-\mathbb{P}_{\Sigma_2}  \big)\otimes \mathcal{O}_{\mathbb{P}_{\Sigma_3^\pm}} (X_2^\pm)\Big)\right\vert_{X_2^\pm}\\
&=& \left.\mathcal{O}_{\mathbb{P}_{\Sigma_3}^\pm} \Big(-\mathbb{P}_{\Sigma_2}\Big)\right\vert_{X^\pm}=\mathcal{O}_{X_2^\pm}\Big(-\mathcal{E}\Big)\,,
\ea
where we used \eqref{eq:K_PSigma3pm} in the second equality  and 
$K_{\mathbb{P}_{\Sigma_4}}\vert_{\mathbb{P}_{\Sigma_3}^\pm}=\mathcal{O}_{\mathbb{P}_{\Sigma_3^\pm}}(X_2^\pm)$.
Thus, the anti-canonical class of $X_2^\pm$ is given by that of the elliptic fiber $E$ which leads to the conclusion that 
$X_2^\pm$ is indeed a rational elliptic surface.

\section{Examples of Heterotic/F-Theory Duals with U(1)'s} \label{section:examples}

In the section, we use the tools of Section \ref{verticalmodel} to construct explicit elliptically
fibered Calabi-Yau two- and threefolds whose stable degeneration limit is well under control. Our
geometries have generically two sections, which give rise to a \text{U}(1)-factor in the
corresponding F-Theory compactification. Performing the toric symplectic cut
allows us to explicitly track these sections through the stable degeneration limit and to make
non-trivial statements about the vector bundle data on the heterotic side in which the
$\text{U}(1)$-factor in the effective theory is encoded.
Finally, after having performed the stable degeneration limit as discussed in section
\ref{subsection:stabledegenerationlimit}, we split the resulting half K3 surfaces into the spectral
cover polynomial and the constraint for the heterotic elliptic curve. Then, we determine the common solutions of the latter two constraints which encode the data of a (split) spectral cover. 
The general geometries we consider as well as the procedure we apply their analysis is discussed in Section \ref{sec:generalgeometry}.
Despite the fact that we do not determine the embedding of the structure group into E$_8$ directly, we are able to match the spectral cover with the resulting gauge group in all cases. 
In particular, we consider three different classes of examples.
In subsection \ref{sec:splitcover} we investigate a number of examples whose heterotic dual gives rise to a split spectral cover. This class of examples has generically one U(1) factor embedded into both E$_8$-bundles of which only a linear combination is massless.
The next class of examples considered in subsection \ref{sec:torsionalcover} displays torsional points in its spectral cover. There is one example with a U(1)-factor on the F-theory side which is found to be only embedded into one E$_8$-bundle while the other E$_8$-bundle is kept intact.
Finally, in the last subsection \ref{sec:nonAbelian} we consider an example where the structure group reads SU(2) $\times$ SU(3). However, we argue that it is embedded in such a way that its centralizer necessarily contains a U(1) factor.

\subsection{The geometrical set-up: toric hypersurfaces in $\mathbb{P}^1 \times\text{Bl}_1 \mathbb{P}^{(1,1,2)}$ }
\label{sec:generalgeometry}

For convenience, we recall the three-dimensional polyhedron $\Delta_3^\circ$ for the resolved 
toric ambient space $\mathbb{P}_{\Sigma_3}=\mathbb{P}^1 \times \text{Bl}_1 \mathbb{P}^{(1,1,2)}$. It is given by the 
points
\be
\left(\begin{array}{ccc|c}-1 & 1 & 0 & x_1 \\  -1 & -1 & 0 & x_2 \\ 1 & 0 & 0 & x_3\\ 0 & 1 & 0 & x_4 \\  -1 & 0 & 0 & x_5  \\ 0 & 0 & 1 & U  \\  0 & 0& -1 & V \end{array} \right)\,. \label{DeltaP112}
\ee
Here, $x_1, \dots x_5$ are homogeneous coordinates on the resolved variety $\text{Bl}_1 \mathbb{P}^{(1,1,2)}$, while $U,V$ denote the two homogeneous coordinates of $\mathbb{P}^1$. In particular, $x_5$ resolves the $A_1$-singularity of the space $\text{Bl}_1{\mathbb{P}}^{(1,1,2)}$.

A generic section of the anti-canonical bundle of the ambient space $\mathbb{P}_{\Sigma_3}$
takes the form
\begin{eqnarray}
\chi \!\!\!&\!\!: =\!\!&\!\!\! s_1 x_1^4x_4^3x_5^2+ s_{2}x_1^3x_2x_4^2x_5^2+s_{3}x_1^2x_2^2x_4x_5+s_{4}x_1x_2^3x_5^2+s_{5}x_1^2x_3x_4^2x_5 \nn \\ \!\!\!&\!\!\!\!&\!\!\! +s_{6}x_1x_2x_3x_4x_5  +s_{7}x_2^2x_3x_5  +s_{8}x_3^2x_4  =0\,,  \label{genericsectionP112}
\end{eqnarray}
where the coefficients  $s_i$ are homogeneous quadratic polynomials in $U,V$.
An elliptically fibered K3 surface is defined by $X_2=\{\chi=0\}$. As can be seen for example its 
Weierstrass form, the K3 surface generically has a  Kodaira fiber 
of type $I_2$ at the locus $s_8=0$. It is resolved by the divisor  
$\left\{x_5=0 \right\} \cap X_2$ as mentioned above.\footnote{As the details of the 
resolution are not important, we can set $x_5=1$ in most computations performed here.} 
In addition, $X_2$ generically has a Mordell-Weil group of rank one. A choice of zero section $S_0$ and generator of the Mordell-Weil 
group $S_1$ are given by
\begin{equation}\label{eq:Rk1GeneralTwoSections}
 S_0 = X_2\cap \{x_1=0 \}, \qquad  S_1 = X_2 \cap \{x_4=0 \}\, .
\end{equation}
Explicitly, their coordinates read
\be \label{sectionsP112}
S_0 = \begin{cases} [0:1:1:s_{7}:-s_{8}] \quad \text{generically}\,, \\ [0:1:1:0:1] \quad \,\,\,\,\,\,\,\,\,\, \text{if}\,\,\, s_{7}=0\,, \\ [0:1:1:1:0] \quad \,\,\,\,\,\,\,\,\,\, \text{if}\,\,\, s_{8}=0\,, \end{cases}
\quad
S_{1} = \begin{cases} [s_{7}:1:-s_{4}:0:1] \quad\text{generically}\,, \\ [0:1:1:0:1] \quad \,\,\,\,\,\,\,\,\,\,  \text{if}\,\,\, s_{7}=0\,, \\ [1:1:0:0:1] \quad \,\,\,\,\,\,\,\,\,\, \text{if}\,\,\, s_{4}=0\,. \end{cases}
\ee
Here we distinguished the special cases with $s_7=0$ and $s_8=0$, respectively, from
the generic situation. 
Using the fact that the generic K3 surface $X_2$ has $h^{(1,1)}=5$ 
\cite{2010arXiv1011.1003B}, we, 
hence, conclude that  the full F-theory gauge group $G_{X_2}$ is\footnote{We note that $s_8=0$ has two solutions on $\mathbb{P}^1$. If we consider higher dimensional bases of the elliptic fibration, we will just have one SU(2) factor as $s_8=0$ is in general an irreducible divisor.} 
\beq
G_{X_2}=\text{SU}(2)  \times \text{SU}(2)  \times \text{U}(1)\,.
\eeq
We  note that if $s_{7}=0$, one observes that the two sections coincide, as was also employed 
in \cite{Morrison:2014era}. That the converse is true is shown in Appendix \ref{lemma}. This is 
expected as the vanishing of $s_7$ can be interpreted as a change of the toric fibre ambient 
space from $\text{Bl}_1 \mathbb{P}^{(1,1,2)}$ to $\mathbb{P}^{(1,2,3)}$, which has a purely 
non-Abelian gauge group \cite{Klevers:2014bqa}.

\subsubsection{Engineering gauge symmetry: specialized sections of $-K_{\mathbb{P}^1 \times \text{Bl}_1 \mathbb{P}^{(1,1,2)}}$}

In order to construct examples with higher rank gauge groups, we tune the coefficients of $\chi^{\text{sing}}$ further. To be concrete, every $s_i$ in \eqref{genericsectionP112} takes the form
\be
s_i = s_{i1} U^2 + s_{i2} U V + s_{i3} V^2 \label{definitionsi}
\ee
and a specialization corresponds to the identical vanishing of some $s_{ij}$. This specialization of coefficients implies that 
$\Delta$, the dual polyhedron of $\mathbb{P}^1 \times \text{Bl}_1 \mathbb{P}^{(1,1,2)}$, gets replaced by the Newton polytope 
$\Delta_{\text{spec.}}$ of the specialized constraint,
compare also Figure \ref{fig:tuning}. As a technical side-remark we note that we strictly speaking refer with $\Delta_{\text{spec.}}$ 
to the convex hull of the points defined by the non-vanishing monomials in \eqref{definitionsi} respectively 
\eqref{genericsectionP112}. As a consequence, also $\Delta^\circ$ changes to the dual of $\Delta_{\text{spec.}}^\circ$. Thus, we 
have secretly changed the toric ambient space by this specialization of coefficients. It is crucial to note that only those polyhedra 
$\Delta_{\text{spec.}}$ give rise to consistent geometries which are reflexive.

\begin{figure}[h!]
\begin{center}
\includegraphics[scale=0.22]{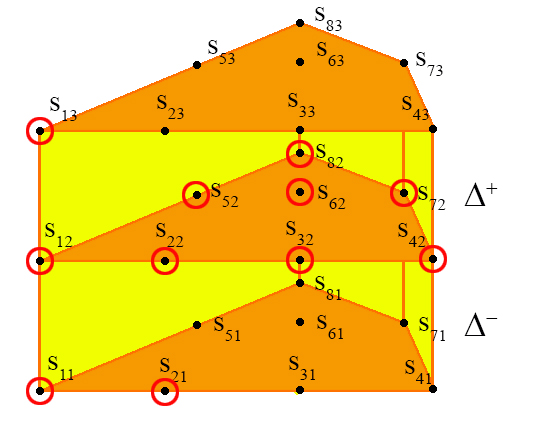}\hspace{1cm}
\includegraphics[scale=0.3]{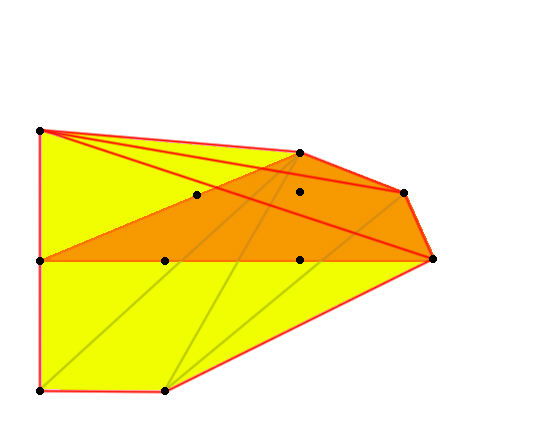}
\end{center}
\caption{This figure illustrates a specialization of the coefficients of the hypersurface $\chi=0$ such that the resulting gauge group is enhanced to E$_7$ $\times$ E$_6$ $\times$ \text{U}(1), see also the discussion in Section \ref{E7E6U1}. In the left picture, the non-vanishing coefficients are marked by a circle in the polytope $\Delta_3$. In the right figure the new polytope, i.e.~the Newton polytope of the specialized constraint $\chi=0$, is shown.}
\label{fig:tuning}
\end{figure}

In order to determine the gauge group of this specialized hypersuface, we need to transform 
$\chi=0$ into its corresponding model into Tate or Weierstrass normal 
form. For convenience, we provide the Weierstrass as well as the Tate form of the most general 
hypersurface in Appendix \ref{appendix:TateandWeierstrass}.

\subsubsection{Stable degeneration and the spectral cover polynomial} \label{section:splitintospectralcover}

As a next step, we show how the K3 surface $X_2$ defined via \eqref{genericsectionP112} can 
be decomposed into the two half K3 surfaces $X_2^\pm$ and the heterotic elliptic curve as well 
as the two spectral cover polynomials, respectively.
First, we write the Calabi-Yau hypersurface equation 
 \eqref{eq:Rk1GeneralTwoSections} for $X_2$ as
 \be
 \label{definitionspectralcover}
\chi= p^+(x_i, s_{j1}) U^2 + p_0 (x_i, s_{j2}) UV + p^-(x_i, s_{j3}) V^2\, ,
\ee
for appropriate polynomials $p^+$, $p_0$ and $p^-$ depending on the fiber coordinates. 
By the results of the previous section, the K3 surface $X_2$  in the semistable degeneration limit can be described  by the half K3 surfaces $X_2^\pm$ with defining equations
\be
X_2^+:\quad p^+(x_i, s_{j1}) U + p_0 (x_i, s_{j2}) V=0, \qquad X_2^-:\quad  p^-(x_i, s_{j3}) V + p_0 (x_i, s_{j2}) U=0\, .
\ee
It follows that generically  the two linearly independent sections \eqref{sectionsP112} of the 
K3  become independent sections in the half K3s, which we denote, by abuse of notation, by the
same symbols. They intersects along the common (heterotic) elliptic curve. 
This is a novel property of our toric degeneration.

In addition, the heterotic elliptic curve is given as $p_0 (x_i, s_j)=0$ while the data of the two 
background bundles are given by the spectral cover polynomials $p^+(x_i, s_{j1})=0$ and 
$p^-(x_i, s_{j3})=0$. The structure group of the two heterotic bundles is then determined by 
the common solutions of $p_0 (x_i, s_j)=0$ with $p^\pm(x_i, s_{j1/3})=0$ using the results 
and techniques from Section \ref{SUNellipticcurve}.

\subsubsection{Promotion to elliptically fibered threefolds} \label{promotionto6d}

Eventually, we are interested in examples of six-dimensional heterotic/F-Theory duality. In order 
to promote the  K3 surfaces $X_2$ constructed above to elliptically fibered threefolds we 
promote the coefficients $s_{ij}$, defined in \eqref{definitionsi}, to sections of a line bundle of 
another $\mathbb{P}^1$ with homogeneous coordinates $R,T$. The base of the previously 
considered K3 surface and the new $\mathbb{P}^1$ form a Hirzebruch surface $\mathbb{F}_n$. 
At this point, we only consider base geometries which are Fano and restrict our discussion to 
$\mathbb{F}_0$ and $\mathbb{F}_1$ for simplicity, avoiding additional singularities in the heterotic elliptic fibration.  
For these two geometries, the explicit form of the $s_{ij}$ reads
\be
s_{ij} = s_{ij1} R^2 + s_{ij2} RT + s_{ij3} T^2\, ,
\ee
for $\mathbb{F}_0$ and
\be
s_i = s_{i11}R + s_{i12} T + s_{i21} R^2 + s_{i22} RT + s_{i23} T^2 + s_{i31}R^3  + s_{i32} R^2 T + s_{i33}R T^2 + s_{i34} T^3\, ,
\ee
for the geometry $\mathbb{F}_1$. 

Next, we observe that the explicit expression of the discriminant of the heterotic Calabi-Yau manifold $Z_n$, which is given by $p_0=0$, contains a factor of 
$s_{82}^2$. While this is certainly not a problem in eight dimensions, as $s_{82}$ is just a 
constant there, it gives rise to an SU(2)-singularity at co-dimension one in the heterotic K3 
surface $Z_2$. This can be cured by a resolution of this singularity through an exceptional divisor 
$E$, which is the analog of $x_5$ in \eqref{DeltaP112}. In particular, the solutions to the 
spectral cover constraint will pass through the singular point in the fiber. 
Thus, one expects that the 
spectral cover curve will pick up contributions from the class $E$ in general. A similar situation 
has been analyzed in  \cite{Donagi:1999ez} where it has been argued that the introduction of 
this exceptional divisor will not change the structure of the spectral cover as an $N$-sheeted 
branched cover of the base except for a finite number of points where it wraps a whole new fiber 
component over the base. As discussed in \cite{Donagi:1999ez}, this introduces more freedom
in the construction of the heterotic vector bundle $V$.
As this work focuses on the mapping of $\text{U}(1)$-factors under the 
heterotic/F-theory duality, we only concentrate on the generic structure of the spectral cover 
and leave the resolution of this singularity as well as an exploration of the freedom in the 
construction of $V$ to future works. 

\subsection{$\text{U}(1)$'s arising from $\text{U}(1)$ factors in the heterotic structure group}
\label{sec:splitcover}

In this section, we consider examples that have an additional rational section in the dual heterotic geometry. We consider
K3 surfaces in F-theory, which are given as hypersurfaces in $\text{Bl}_1 \mathbb{P}^{(1,1,2)} \times \mathbb{P}^1$ with 
appropriately specialized coefficients generating a corresponding gauge symmetry.  
Elliptic K3 fibered Calabi-Yau threefolds are constructed straightforwardly as described in section \ref{promotionto6d}. 
Thus, our following discussion will be equally valid in six dimensions, although, in order to avoid confusion, we present our 
geometric discussions in eight dimensions. Having this in mind we, therefore, drop here in the rest of this work 
the subscripts on all considered manifolds $X_{n+1}$, $X_{n+1}^\pm$ and $Z_n$, respectively, .
In the following, we discuss the main geometric properties of the Calabi-Yau manifold $X$, demonstrate heterotic/F-theory duality 
and relations among different examples by a chain of Higgsings. 

We begin by a summary of key results and by setting some notation.
As we will see, all considered examples have the same heterotic Calabi-Yau manifold $Z$ in common. It is given by the most generic section of the anti-canonical bundle in $\text{Bl}_1 \mathbb{P}^{(1,1,2)}$ reading
\be
Z:\,\,\, s_{12} x_1^4+s_{22}x_1^3 x_2+s_{32}x_1^2x_2^2 +s_{42} x_1 x_2^3 +s_{52} x_1^2 x_3  +s_{62}x_1x_2x_3  +s_{72}x_2^2x_3  +s_{82}x_3^2  =0\, . \label{HeteroticcurveU1}
\ee
The examples considered here only differ among each other by the spectral covers, i.e.~by 
the choice of the coefficients $s_{i1}$ and $s_{i3}$ in \eqref{newconvention}, which will be different in each case.
Generically, all examples will have a \text{U}(1)-factor embedded into the structure groups of both heterotic vector bundles $V_1$, $V_2$. Thus, the maximal non-Abelian gauge group determining any chain of Higgsings is given by E$_7$ $\times$ E$_7$. For later reference we also note the Weierstrass normal form of \eqref{HeteroticcurveU1} is given by
\begin{eqnarray}
W_{Z}:\,\,\, y^2 &=& x^3 + \left( -\frac{1}{48}s_{62}^4 + \frac{1}{6} s_{52} s_{62}^2 s_{72} - \frac{1}{3} s_{52}^2 s_{72}^2 - 
 \frac{1}{2} s_{42} s_{52} s_{62} s_{82} + \frac{1}{6} s_{32} s_{62}^2 s_{82} \right. \nn \\ && \left.+ \frac{1}{3} s_{32} s_{52} s_{72} s_{82} - 
 \frac{1}{2} s_{22} s_{62} s_{72} s_{82} + s_{21} s_{72}^2 s_{82} - \frac{1}{3} s_{32}^2 s_{82}^2 + s_{22} s_{42} s_{82}^2 \right)x \nn \\ && 
 + \left(  \frac{1}{864}s_{62}^6 - \frac{1}{72} s_{52} s_{62}^4 s_{72} + \frac{1}{18} s_{52}^2 s_{62}^2 s_{72}^2 - 
 \frac{2}{27} s_{52}^3 s_{72}^3 + \frac{1}{24} s_{42} s_{52} s_{62}^3 s_{82} \right. \nn \\ && \left. - \frac{1}{72} s_{32} s_{62}^4 s_{82} - 
 \frac{1}{6} s_{42} s_{52}^2 s_{62} s_{72} s_{82} + \frac{1}{36} s_{32} s_{52} s_{62}^2 s_{72} s_{82} + 
 \frac{1}{24} s_{22} s_{62}^3 s_{72} s_{82} \right. \nn \\ && \left. + \frac{1}{9} s_{32} s_{52}^2 s_{72}^2 s_{82} - 
 \frac{1}{6} s_{22} s_{52} s_{62} s_{72}^2 s_{82} - \frac{1}{12} s_{21} s_{62}^2 s_{72}^2 s_{82} + 
 \frac{1}{3} s_{21} s_{52} s_{72}^3 s_{82} \right. \nn \\ && \left. + \frac{1}{4} s_{42}^2 s_{52}^2 s_{82}^2 - 
 \frac{1}{6} s_{32} s_{42} s_{52} s_{62} s_{82}^2 + \frac{1}{18} s_{32}^2 s_{62}^2 s_{82}^2 - 
 \frac{1}{12} s_{22} s_{42} s_{62}^2 s_{82}^2 \right. \nn \\ && \left. + \frac{1}{9} s_{32}^2 s_{52} s_{72} s_{82}^2 - 
 \frac{1}{6} s_{22} s_{42} s_{52} s_{72} s_{82}^2 - \frac{1}{6} s_{22} s_{32} s_{62} s_{72} s_{82}^2 + 
 s_{21} s_{42} s_{62} s_{72} s_{82}^2 \right. \nn \\ && \left.+ \frac{1}{4} s_{22}^2 s_{72}^2 s_{82}^2 - 
 \frac{2}{3} s_{21} s_{32} s_{72}^2 s_{82}^2 - \frac{2}{27} s_{32}^3 s_{82}^3 + 
 \frac{1}{3} s_{22} s_{32} s_{42} s_{82}^3 - s_{21} s_{42}^2 s_{82}^3 \right)\, . \nn \\ &&
\end{eqnarray}
Also, the two generic sections of the heterotic geometry $Z$, denoted by $S_0^Z$ and $S_1^Z$ read in Weierstrass normal form as
\begin{eqnarray}
S_0^Z &=& \left[1:1:0 \right]\,, \label{section1W} \\
S_1^Z &=& \left[ \frac{1}{12} \left(s_{62}^2 s_{72}^2 - 4 s_{52} s_{72}^3 - 12 s_{42} s_{62} s_{72} s_{82} + 
   8 s_{32} s_{72}^2 s_{82} + 12 s_{42}^2 s_{82}^2\right) \right. : \frac{1}{2} s_{82} \left(s_{42} s_{62}^2 s_{72}^2 \right. \nn \\ &&  \left. \left.  - s_{42} s_{52} s_{72}^3 - s_{32} s_{62} s_{72}^3 + 
     s_{22} s_{72}^4 - 3 s_{42}^2 s_{62} s_{72} s_{82} + 2 s_{32} s_{42} s_{72}^2 s_{82} + 
     2 s_{42}^3 s_{82}^2\right) : s_7 \right] \,. \nn \\ && \label{section2W}
\end{eqnarray}
Here, the first section $S_0^Z$ is the point at infinity, while the second section $S_1^Z$ can be seen also in the affine chart. We note that $S_0^Z$ can be obtained by a simple coordinate transformation\footnote{To be more precise, we refer in this case to \eqref{eq:Rk1GeneralTwoSections} as a section of the heterotic geometry.} from $S_0$ defined in \eqref{eq:Rk1GeneralTwoSections}, while $S_1^Z$ needs to be constructed using the procedure of Deligne applied in \cite{Morrison:2012ei}.

\subsubsection{Structure group $ \text{U}(1) \times \text{U}(1) $: $\text{E}_7$ $\times$ $\text{E}_7 \times \text{U}(1) $ gauge symmetry} \label{E7E7U1}

We start with a model which has a heterotic vector bundle of structure group $\text{U}(1) \times \text{U}(1)$. Upon commutation within the group $\text{E}_8 \times \text{E}_8$, the centralizer is given as $\text{E}_7 \times \text{U}(1) \times \text{E}_7 \times \text{U}(1)$. On the heterotic side, the two \text{U}(1) factors acquire a mass term so that only a linear combination of them is massless. This matches the F-theory gauge group given by $\text{E}_7 \times \text{E}_7 \times \text{U}(1)$.  

Our example is specified by the following non-vanishing coefficients:
\begin{center}
\begin{tabular}{|c|c|c|c|}\hline
Coefficient & $X$ & $X^-$ & $X^+$ \\ \hline
$s_1$ & $s_{11} U^2 + s_{12} U V +s_{13} V^2$ & $s_{12}U + s_{13}\mu$ & $s_{11}\lambda_1 + s_{12}V$ \\
$s_{2}$ & $s_{22} U V$ & $s_{22} U$ & $s_{22} V$ \\
$s_{3}$ & $s_{32} U V$ & $s_{32} U$ & $s_{32} V$ \\
$s_{4}$ & $s_{42} U V$ & $s_{42} U$ & $s_{42} V$ \\
$s_{5}$ &$s_{52} U V$ & $s_{52} U$ &  $s_{52} V$ \\
$s_{6}$ & $s_{62} U V$ & $s_{62} U$ & $s_{62} V$ \\
$s_{7}$ & $s_{72} U V$ & $s_{72} U$ & $s_{72} V$ \\
$s_{8}$ & $s_{82} U V$ & $s_{82} U$ & $s_{82} V$ \\ \hline
\end{tabular}
\end{center}
Here, the first columns denotes the coefficient in the Calabi-Yau constraint \eqref{CYconstraint}, the second column indices the 
chosen specialization and the third as well as fourth column contain the resulting coefficient in the half K3 fibrations $X^\pm$, respectively.

 Using the identities \eqref{coefficientshalfK3} and \eqref{definitionXpm}, we readily write down the defining equations for the half K3 surfaces $X_2^\pm$ obtained via stable degeneration explicitly. They read
\begin{eqnarray}
X^+: \quad \left(s_{11} \lambda_1 + s_{12} V \right) x_1^4+s_{22} V x_1^3 x_2+s_{32} V x_1^2x_2^2 +s_{42} V x_1 x_2^3 && \nn \\ +s_{52} V x_1^2 x_3 +s_{62} V x_1x_2x_3  +s_{72} V x_2^2x_3  +s_{82} V x_3^2  &=& 0\,, \nn \\
X^-: \,\,\,\, \quad \left( s_{12} U + s_{13} \mu \right) x_1^4+s_{22} U x_1^3 x_2+s_{32} U x_1^2x_2^2 +s_{42} U x_1 x_2^3 && \nn \\ +s_{52} U x_1^2 x_3 +s_{62} U x_1x_2x_3  +s_{72} U x_2^2x_3  +s_{82} U x_3^2  &=& 0\,.
\end{eqnarray}

\begin{figure}[t!]
\begin{center}
\includegraphics[width=5cm]{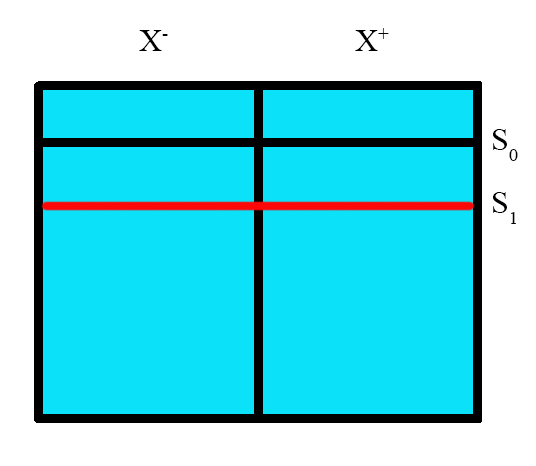}
\end{center}
\caption{This figure shows the stable degeneration limit of a K3 surface which has E$_7\times\text{E}_7\times\text{U}(1)$ gauge symmetry. There are the two half K3 surfaces, $X^+$ and $X^-$ which have both an E$_7$ singularity and intersect in a common elliptic curve. Both have two sections, $S_0$ and $S_1$ which meet in the common elliptic curve. Thus, there are two global sections in the full K3 surface and therefore a $\text{U}(1)$ factor.}
\label{figure:E7E7U1model}
\end{figure}
By explicitly evaluating the Tate coefficients \eqref{Tatecoefficients1}, one obtains the following orders of vanishing for the Tate 
vector at the loci $U=0$ and $V=0$ for the full K3 surface,
\be 
\vec{t}_U = \vec{t}_V = \left( 1 , 2 , 3 , 3 , 5 , 9 \right)\, ,
\ee
which reveal two $\text{E}_7$ singularities. Also, the two half K3 surfaces inherit an $\text{E}_7$ singularity each, which are located at $U=0$, $V=0$, respectively. Thus, the non-Abelian part of the gauge group is given by  $\text{E}_7\times \text{E}_7$. 
Both half K3 surfaces have two rational sections given by $S_0^{X^\pm}=[0:1:0]$ and $S_1^{X^+}=[0:s_{82}V:-s_{72} s_{82}V^2]$ and $S_1^{X^-}=[0:s_{82}U:-s_{72} s_{82}U^2]$, respectively.  In the intersection point of the two half K3's given by $[U:\lambda_1:\mu:V:\lambda_2]=[1:0:0:1:1]$, the sections $S_0^{X^\pm}$ from both half K3's intersect and meet each other, and similarly the sections $S_1^{X^\pm}$ from both half K3's intersect and meet each other, cf.~Figure \ref{figure:E7E7U1model}. Thus,  the six-dimensional gauge group contains a $\text{U}(1)$ factor.

However, if one evaluates the spectral cover, as described in section \ref{section:splitintospectralcover}, one obtains\footnote{Here, and in the following we set  $x_4\rightarrow 1$, $x_5 \rightarrow 1$ for convenience.}
\be
p_+ = s_{11}x_1^4, \qquad p_- = s_{13} x_1^4 \, .
\ee
which is mapped by use of the transformations \eqref{transfP112Weier} onto
\be
p_+^{W} = s_{11}z^4, \qquad p_-^W = s_{13} z^4 \, .
\ee
These expressions gives rise to a constant spectral cover in affine Weierstrass coordinates $x,y$ defined by $z=1$. However, on 
an elliptic curve there does not exist any function which has exactly one zero at a single point, in this case $S_1$.\footnote{However, note that the homogeneous expression $x_1^4$ vanishes indeed at the loci of $S_0^Z$ and of $S_1^Z$.} 
Nevertheless, one can use the two points $S_0^Z$ and $S_1^Z$ on the heterotic elliptic curve in order to construct the 
bundle $\mathcal{L}=\mathcal{O}(S_1^Z-S_0^Z)$ fiberwise, which is symmetrically embedded into both E$_8$-bundles. As 
argued in \cite{Aspinwall:2005qw}, this bundle promotes to a bundle $\mathcal{L}^{6d}$ in six dimensions whose first Chern 
class is given by the difference of the two sections $c_1(\mathcal{L}^{6d})= \sigma_{S_1^Z} - \sigma_{S_0^Z}$, up to fiber 
contributions.  Thus, the heterotic gauge group is given by  E$_7$ $\times$ E$_7$ $\times$ $\text{U}(1)$ $\times$ $\text{U}(1)$. 
Due to the background bundle $\mathcal{L}^{6d}$, these two $\text{U}(1)$'s seem both massive according to the St\"uckelberg
mechanism discussed in Section \ref{sec:Stueckelberg}. However, due to the symmetric embedding into both $\text{E}_8$'s their 
sum remains massless. Thus, one obtains a perfect match with the F-theory gauge group. 

We conclude with the remark that one can interpret this model also as a Higgsing of a model with $\text{E}_8\times\text{E}_8$
gauge symmetry as presented in the Appendix \ref{E8E8model}. Here, the Higgsing corresponds to a geometrical transition from 
the ambient space geometry $\mathbb{P}^{(1:2:3)}$ to $\text{Bl}_1 \mathbb{P}^{(1,1,2)}$ where the vacuum expectation value of
the Higgs corresponds to the non-vanishing coefficient $s_{72}$. 

\subsubsection{Structure group $\text{S}(\text{U}(2) \times \text{U}(1))$: $\text{E}_7$ $\times$ $\text{E}_6 \times \text{U}(1)$  gauge symmetry} \label{E7E6U1}

As a next step, we investigate an example which has $\text{E}_7 \times \text{E}_6 \times \text{U}(1)$ gauge symmetry. On the 
heterotic side we find an $\text{U}(1) \times \text{SU}((2) \times \text{U}(1))$ structure group which directly 
matches the non-Abelian gauge group and gives rise to one massless as well as one massive $\text{U}(1)$. 
The model is specified by the following non-vanishing coefficients:
\begin{center}
\begin{tabular}{|c|c|c|c|}\hline
Coefficient & $X$ & $X^-$ & $X^+$ \\ \hline
$s_1$ & $s_{11} U^2 + s_{12} U V +s_{13} V^2$ & $s_{12}U + s_{13}\mu$ & $s_{11}\lambda_1 + s_{12}V$ \\
$s_{2}$ & $s_{21}U^2 + s_{22} U V$ & $s_{22} U$ & $s_{21} \lambda_1 +s_{22} V$ \\
$s_{3}$ & $s_{32} U V$ & $s_{32} U$ & $s_{32} V$ \\
$s_{4}$ & $s_{42} U V$ & $s_{42} U$ & $s_{42} V$ \\
$s_{5}$ &$s_{52} U V$ & $s_{52} U$ &  $s_{52} V$ \\
$s_{6}$ & $s_{62} U V$ & $s_{62} U$ & $s_{62} V$ \\
$s_{7}$ & $s_{72} U V$ & $s_{72} U$ & $s_{72} V$ \\
$s_{8}$ & $s_{82} U V$ & $s_{82} U$ & $s_{82} V$ \\ \hline
\end{tabular}
\end{center}
\begin{figure}[h!]
\begin{center}
\includegraphics[width=5cm]{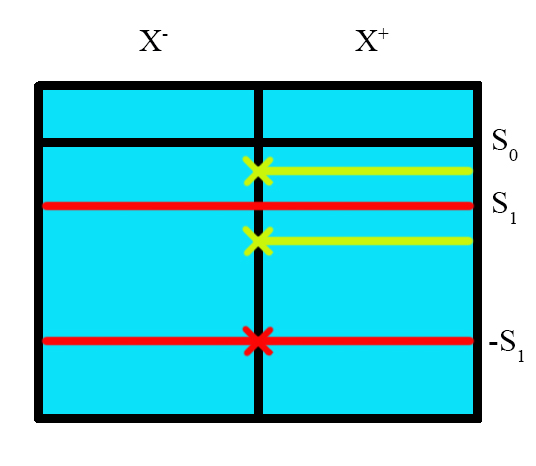}
\end{center}
\caption{The interpretation of this figure is similar to Figure \ref{figure:E7E7U1model}. The additional structure arises from two sections shown in yellow which form together with $\boxminus S_1$ the zeros of the spectral cover.}
\label{E7xE6xU(1)} 
\end{figure}
The evaluation of the order of vanishing of the Tate coefficients is summarized in the two Tate vectors
\be 
\vec{t}_V =  (1,2,2 , 3 , 5 , 8) \,  \qquad \vec{t}_U = ( 1,  2 , 3, 3 , 5 , 9) \, .
\ee
It signal one $\text{E}_6$ singularity at $V=0$ and one $\text{E}_7$ singularity at $U=0$. The E$_7$ singularity is inherited by the half K3 surface $X^-$ while the E$_6$ singularity is contained in the half K3 surface $X^+$ after stable degeneration. 

Next, we turn to the heterotic side. Here, the analysis of sections and $\text{U}(1)$ symmetries from the perspective of the gluing 
condition is completely analogous to the geometry with $\text{E}_7$ $\times$ $\text{E}_7$ $\times \text{U}(1)$ gauge symmetry 
discussed in the previous Section \ref{E7E7U1}. The situation at hand is summarized in Figure \ref{E7xE6xU(1)}.  However, there 
is a crucial difference in the evaluation of the spectral cover which we discuss next.

The corresponding split of the two half K3 surfaces into a spectral cover polynomial and the heterotic elliptic curve results in
\be
p^+ = s_{11} x_1^4 + s_{21} x_1^3 x_2 \, , \qquad    p^- = s_{13} x_1^4\, .
\ee
Again, in order to evaluate the spectral cover information, one needs to transform both constraints
into Weierstrass normal form. $p^-$ is again just a constant and its interpretation is along the lines of the previous example in 
Section \ref{E7E7U1}. However, in the case of $p^+$ something non-trivial happens. 
Its transform into Weierstrass coordinates reads explicitly
\begin{eqnarray}
p^+_W &=& \left(s_{21} s_{62}^3 s_{72} - 4 s_{21} s_{52} s_{62} s_{72}^2 - 2 s_{11} s_{62}^2 s_{72}^2 + 
   8 s_{11} s_{52} s_{72}^3 - 2 s_{21} s_{42} s_{62}^2 s_{82} \right. \nn \\ && \left. - 4 s_{21} s_{42} s_{52} s_{72} s_{82} - 
   4 s_{21} s_{32} s_{62} s_{72} s_{82} + 24 s_{11} s_{42} s_{62} s_{72} s_{82} + 
   12 s_{21} s_{22} s_{72}^2 s_{82} \right. \nn \\ && \left. - 16 s_{11} s_{32} s_{72}^2 s_{82} + 
   8 s_{21} s_{32} s_{42} s_{82}^2 - 24 s_{11} s_{42}^2 s_{82}^2 - 12 s_{21} s_{62} s_{72} x \right. \nn \\ && \left. + 
   24 s_{11} s_{72}^2 x + 24 s_{21} s_{42} s_{82} x + 
   24 s_{21} s_{72} y\right)/\left(2 \left(-s_{62}^2 s_{72}^2 + 4 s_{52} s_{72}^3 \right. \right. \nn \\ && \left. \left.+ 
     12 s_{42} s_{62} s_{72} s_{82} - 8 s_{32} s_{72}^2 s_{82} - 12 s_{42}^2 s_{82}^2 + 
     12 s_{72}^2 x\right)\right)\, . \label{spectralcoverE7E6U1}
\end{eqnarray}

In contrast to the well-known case of the spectral cover in the $\mathbb{P}^{(1,2,3)}$-model which takes only poles at infinity, 
one observes that the denominator of \eqref{spectralcoverE7E6U1} has two zeros at $S_1^Z$ and at $\boxminus S_1^Z$, the 
negative of $S_1^Z$ in the Mordell-Weil group of $Z$. 
In addition, the numerator has zeros at two irrational points $Q_1$, $Q_2$ and at $\boxminus S_1$.  Finally, there is a pole of 
order one at $S^Z_0$. Here, $S^Z_0$ and $S^Z_1$ refer to the two sections \eqref{section1W}, \eqref{section2W}.
Thus, the divisor of $p^-_W$ is given by
\be
\text{div}\left(p^+_W \right) = Q_1 + Q_2 - S_1 - S_0\, . \label{divisorP+}
\ee
Clearly, in order to promote the points defined by the spectral cover polynomial in eight dimensions to a curve in six dimensions, 
the current form of $p^-_W$ is not suitable due to its non-trivial denominator. However, one observes that the polynomial given 
by the numerator of $p^+_W$ gives rise to the divisor
\be
\text{div}\left( \text{Numerator}\left( p^+_W \right)\right) = Q_1 + Q_2 + \boxminus S_1 - 3 S_0
\ee
which is, however, linearly equivalent\footnote{To see this, one notices that the element $-S_1 + S_0$ in Pic$^0(E)$ is equivalent 
to $-S_1 + S_0 + f$ where $f$ is defined as $x-x_{S_1}$ on $E$ with $x_{S_1}$ denoting the $x$-coordinate of $S_1$. It holds 
that div$(f)=S_1 + \boxminus S_1 -2 S_0$. Thus, $-S_1+S_0$ maps to $\boxminus S_1$ on $E$ under the map 
\eqref{mapEtoPicE}.} to the divisor \eqref{divisorP+}.  Consequently, a spectral cover, valid also for the construction of
lower-dimensional compactifications, is defined by the numerator of
\eqref{spectralcoverE7E6U1}.

Thus, the three zeros $Q_1, Q_2$ and $\boxminus S_1$ form, following Section \ref{SUNellipticcurve}, a split SU(3) spectral 
cover, i.e.~an S(U(2)$\times\text{U}(1)$) spectral cover. All three points extend as sections into the half K3 surface $X^+$, cf.~Figure \ref{E7xE6xU(1)}. Two of these sections are linearly independent and are in eight dimensions the generators of the rank two Mordell-Weil group corresponding to a rational elliptic surface with an E$_6$ singularity. However, due to monodromies of $Q_1$ and $Q_2$ only $\boxminus S_1$ survives in six dimensions as a rational section. 

In conclusion, this spectral cover gives rise to an $\text{S}(\text{U}(2) \times \text{U}(1))$ background bundle which is embedded 
into the $\text{E}_8$ factor corresponding to $X^+$. The centralizer of this is given by $\text{E}_6 \times$ $\text{U}(1)$. The latter 
factor seems again massive due to the $\text{U}(1)$ background bundle. However, this  $\text{U}(1)$ forms together with the 
seemingly massive $\text{U}(1)$ of the  half K3 surface $X^-$ a massless linear combination. In conclusion, there is a perfect 
match with the F-theory analysis of the low energy gauge group. 
Analogously to the previous case in Section  \ref{E7E7U1}, this model can be understood as arising by Higgsing the non-Abelian 
model \ref{E8E7model} with gauge symmetry $\text{E}_8$ $\times$ $\text{E}_7$. Here, a (massive) $\text{U}(1)$ factor is 
embedded minimally into both factors. Again, the vacuum expectation value of the Higgs corresponds to the coefficient $s_{72}$. 
In addition, we can view this model also as arising by a Higgsing process from a compactification with $\text{E}_7\times\text{E}_7\times\text{U}(1)$ gauge group where a vacuum expectation value of the Higgs corresponds to $s_{21}$.

\subsubsection{Structure group $\text{SU}(2) \times \text{SU}(2) \times \text{U}(1)$: $\text{E}_7 \times \text{SO}(9) \times \text{U}(1)$ gauge symmetry} \label{E7SO9U1}

The final example in this chain of Higgsings is given by a model with $\text{E}_7 \times \text{SO}(9) \times \text{U}(1)$ gauge symmetry. On the heterotic side we find an $\text{U}(1) \times \left(\text{SU}(2) \times \text{SU}(2)\times \text{U}(1)\right)$ structure group which matches the non-Abelian gauge content. Also in this case we find one massless as well as one massive $\text{U}(1)$ on the heterotic side. 

As before, we define the model  by the following choice of coefficients in $X$:
\begin{center}
\begin{tabular}{|c|c|c|c|}\hline
Coefficient & $X$ & $X^-$ & $X^+$ \\ \hline
$s_1$ & $s_{11} U^2 + s_{12} U V +s_{13} V^2$ & $s_{12}U + s_{13}\mu$ & $s_{11}\lambda_1 + s_{12}V$ \\
$s_{2}$ & $s_{21}U^2 + s_{22} U V$ & $s_{22} U$ & $s_{21} \lambda_1 +s_{22} V$ \\
$s_{3}$ & $s_{31}U^2 + s_{32} U V$ & $s_{32} U$ & $s_{31} \lambda_1 + s_{32} V$ \\
$s_{4}$ & $s_{42} U V$ & $s_{42} U$ & $s_{42} V$ \\
$s_{5}$ &$s_{52} U V$ & $s_{52} U$ &  $s_{52} V$ \\
$s_{6}$ & $s_{62} U V$ & $s_{62} U$ & $s_{62} V$ \\
$s_{7}$ & $s_{72} U V$ & $s_{72} U$ & $s_{72} V$ \\
$s_{8}$ & $s_{82} U V$ & $s_{82} U$ & $s_{82} V$ \\ \hline
\end{tabular}
\end{center}
\begin{figure}[t]
\begin{center}
\includegraphics[width=5cm]{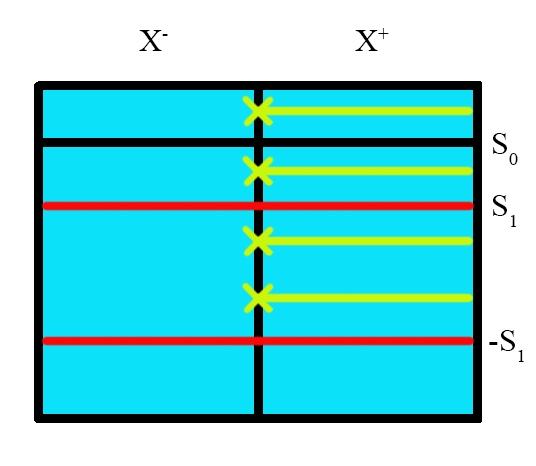}
\end{center}
\caption{The half K3 surface $X^-$ only exhibits the section $S^1$ in addition to the zero section. In contrast, $X^+$ gives rise to a spectral cover polynomial that has two pairs of irrational solutions $Q_1, Q_2$, $R_1, R_2$ that sum up to $S_1^Z$ each.}
\label{E7xSO9xU(1)} 
\end{figure}
Once again we begin the analysis on the F-theory side with the evaluation of the order of vanishing of the Tate coefficients. We obtain the Tate vectors
\be
\vec{t}_V =( 1,  1,  2,  3,  4,  7) \,,  \qquad \vec{t}_U = ( 1 , 2 , 3 , 3 , 5 , 9) \, ,
\ee
which signal one SO(9) singularity at $V=0$ and one $\text{E}_7$ singularity at $U=0$, each of which being inherited by one half 
K3 surface. 

For the analysis of  the heterotic side, we  split the two half K3 surfaces into a spectral cover polynomial and the heterotic elliptic 
curve. We obtain
\be
p^+ = s_{11} x_1^4 + s_{21} x_1^3 x_2 + s_{31} x_1^2 x_2^2\,, \qquad p^- = s_{11} x_1^4\,,
\ee
from which we see that $p^-$ is again a trivial spectral cover.
Again, in order to evaluate the non-trivial spectral cover $p^+$, one needs to transform both constraints into Weierstrass normal 
form. The interpretation of $p^+$ is as in the previous cases. We again obtain a Weierstrass form $p^+_W$ with
a denominator. The explicit expression is rather lengthy and can be provided upon request. Its  divisor is given by
\be
\text{div}\left(p^+_W \right) = Q_1 + Q_2 + R_1 + R_2 - 2  S_1 - 2 S_0\, , 
\ee
Here $Q_1$, $Q_2$ and $R_1$, $R_2$ are two pairs of irrational points which obey $Q_1 \boxplus Q_2 \boxminus S_1 =0$ and $R_1 \boxplus R_2 \boxminus S_1 =0$. 
The divisor of $p_W^+$ is again equivalent to the divisor of its numerator reading
\be
\text{div}\left( \text{Numerator}\left( p^+_W \right)\right) = Q_1 + Q_2 + R_1 + R_2 + 2 \boxminus S_1 - 6 S_0\,.
\ee
By a similar token as before, we thus drop the denominator and just work with the numerator of $p_W^+$.

All the points appearing here extend to sections of the half K3 surface $X^+$. However, while $Q_1, Q_2, R_1, R_2$ extend to rational sections of the half K3 surface they do not lift to rational sections of the fibration of the rational 
elliptic surface over $\mathbb{P}^1$. Altogether, we obtain as in the previous examples two rational sections in both half K3
surfaces which glue to global sections and therefore give rise to a $\text{U}(1)$ factor. Besides that the spectral cover is split and 
describes a vector bundle with structure group $\text{S}(\text{U}(2) \times \text{U}(1))\times\text{S}(\text{U}(2) \times \text{U}(1))$, 
where the $\text{U}(1)$ part in both factors needs to be identified. This is due to the fact that in both cases the same point, 
$\boxminus S_1^Z$, splits off. Thus, the spectral cover is isomorphic to $\text{SU}(2) \times \text{SU}(2) \times \text{U}(1)$ 
whose centralizer\footnote{We employ here  the breaking $\text{E}_8$ 
$\longrightarrow \text{SO}(9) \times \text{SU}(2) \times \text{SU}(2) \times \text{SU}(2)$.} within $\text{E}_8$ is given by 
$\text{SO}(9) \times \text{U}(1)$. Thus we obtain again two seemingly massive $\text{U}(1)$'s which give rise to one massless 
linear combination. 

This model can be understood by a Higgsing mechanism. Either it can be viewed as arising from the non-Abelian model in 
Section \ref{E8SO11model} with $\text{E}_8 \times \text{SO}(11)$ gauge symmetry, by giving a vacuum expectation value to a 
Higgs corresponding  to $s_{72}$, or from the previous example in Section \ref{E7E6U1}, by giving a vacuum expectation value 
to a Higgs associated to $s_{31}$. 

\subsubsection{Example with only one massive $\text{U}(1)$: $\text{S}(\text{U}(1) \times \text{U}(1))$ structure group}

Finally, we conclude the list of examples with a model which has only one $\text{U}(1)$-bundle embedded into one of its E$_8$ factors while the other E$_8$ stays untouched.  Accordingly there is only one massive $\text{U}(1)$ symmetry. On the F-theory side we obtain an $\text{E}_8 \times \text{E}_6 \times \text{SU}(2)$ gauge symmetry which matches the findings on the heterotic side.

The model is defined by the following specialization of the coefficients in the constraint \eqref{genericsectionP112}:
\begin{center}
\begin{tabular}{|c|c|c|c|}\hline
Coefficient & $X$ & $X^-$ & $X^+$ \\ \hline
$s_{1}$ & $s_{13}V^2$ & $s_{13} \mu$ & $0$ \\
$s_{2}$ & $s_{22} UV$ & $s_{22} U$ & $s_{22} V$ \\
$s_{3}$ & $s_{32} UV$ & $s_{32} U$ & $s_{32} V$ \\
$s_{4}$ & $s_{41} U^2+s_{42} UV$ & $s_{42} U$ & $s_{42} V+s_{41} \lambda_1$ \\
$s_{5}$ & $s_{52} UV$ & $s_{52} U$ & $s_{52} V$ \\
$s_{6}$ & $s_{62} UV$ & $s_{62} U$ & $s_{62} V$ \\
$s_{7}$ & $0$ & $0$ & $0$ \\
$s_{8}$ & $s_{82} UV$ & $s_{82} U$ & $s_{82} V$ \\ \hline
\end{tabular}
\end{center}
\begin{figure}[h!]
\begin{center}
\includegraphics[width=5cm]{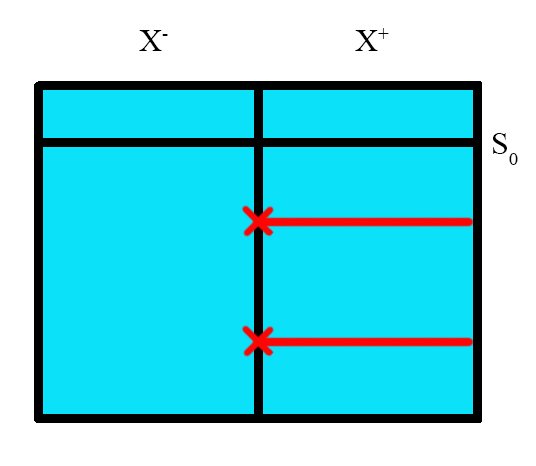}
\end{center}
\caption{The half K3 surface $X^-$ has only one section $S^{X^-}_0$ which merges with the section $S^{X^+}_0$ from the other half K3 surface $X^+$. $X^+$ has in addition also the section $S^{X^+}$ which does not merge with a section of $X^-$. Thus, there is no U(1)-factor on the F-theory side.}
\label{E7xSO9xSU(2)} 
\end{figure}
First of all, we note that the coefficient $s_7$ vanishes identically. Thus, we have changed the ambient space of the fiber from $\text{Bl}_1 \mathbb{P}^{(1,1,2)}$ to $\mathbb{P}^{(1,2,3)}$. Therefore, we do not expect to see another section besides the zero section on the F-theory side and therefore no $\text{U}(1)$, cf.~Appendix \ref{lemma}. 

First, we determine the gauge group on the F-theory. As before, we evaluate the Tate coefficients along the singular fibers which 
are in the case at hand located at $U=0$, $V=0$ and $s_{41} U+s_{42} V=0$. One obtains the Tate vectors
\begin{equation}
\vec{t}_{U} = \left( 1 , 2 ,3 , 4 , 5 , 10 \right)\, , \qquad \vec{t}_V = \left(1 , 2 , 2 , 3 , 5 , 8 \right)\, , \qquad   \vec{t}_{s_{41}U+s_{42}V} = \left( 0 , 0 , 1 , 1 , 2, 2 \right) \, .
\end{equation}
Clearly, these signal an $\text{E}_8 \times \text{E}_6 \times \text{SU}(2)$ gauge group in  F-Theory. Also, after the stable 
degeneration limit, one obtains one half K3 surface $X^-$ with an $\text{E}_8$ singularity and one, $X^+$, with an $\text{E}_6\times \text{SU}(2)$ singularity.  

For the further analysis we remark that there is the zero section $S_0=[0:1:0]$ in the K3 surface only. Here and in the following, 
we refer to the $\mathbb{P}^{(1,1,2)}$ coordinates $[x_1:x_2:x_3]$ only, i.e.~we work in the limit $x_4 \rightarrow 1$, $x_5 \rightarrow 1$. For the two half K3s one finds that $X^-$ has only a zero section. In contrast, one observes the sections\footnote{Clearly, as the rank of the Mordell Weil group of $X^+$ is positive, there are in fact infinitely many sections.} $S^{X^+}=[1:0:0]$ and $\boxminus S^{X^+}=[s_{82}V:0:-s_{52} s_{82}V^2]$ in the other half K3 surface $X^+$.
However, these sections do not glue with another section of $X^-$ and therefore do not give rise to a $\text{U}(1)$ symmetry from the F-theory perspective. However, from the heterotic perspective they should give rise to a \textit{massive} $\text{U}(1)$ which upon commutation within $\text{E}_8$ leaves an $\text{E}_6 \times \text{SU}(2)$ gauge symmetry. 

This result is in agreement with the spectral cover analysis. One evaluates the spectral cover polynomials as
\be
p^-= s_{13}x_1^4  \qquad p^+ = s_{41} x_1x_2^3 \,.
\ee
As observed already before, the Weierstrass transform $p_W^-$ of $p_W^-$ does not have any common solution with the heterotic elliptic curve and therefore the $\text{E}_8$-symmetry does not get broken. For the half K3 surface $X^+$, the common solutions to $p^+_W$ and the heterotic elliptic curve are given in Weierstrass coordinates $[x:y:z]$ as
\begin{eqnarray}
S^Z_W  &=&  \left[\frac{1}{12}(s_{62}^2-4s_{32}s_{82}): -\frac{1}{2}s_{42}s_{52}s_{82} :1 \right]\,, \nonumber\\ 
\boxminus S^Z_W  &=&  \left[\frac{1}{12}(s_{62}^2-4s_{32}s_{82}): \frac{1}{2}s_{42}s_{52}s_{82} :1 \right]\, .
\end{eqnarray}
Here, $S^Z_W$ and $\boxminus S^Z_W$ denote the intersections of $S^{X^+}$ and $\boxminus S^{X^+}$ with the heterotic 
geometry $Z$ respectively, in Weierstrass coordinates. Thus, we observe a split spectral cover pointing towards the structure 
group $S(\text{U}(1) \times \text{U}(1))$. Using the breaking $\text{E}_8 \longrightarrow \text{E}_6 \times \text{SU}(2) \times \text{U}(1)$, this spectral cover matches with the observed gauge group. The U(1) is decoupled from the massless spectrum via the St\"uckelberg effect of Section \ref{sec:Stueckelberg}.

\subsection{Split spectral covers with torsional points}
\label{sec:torsionalcover}

In the following, we discuss examples which exhibit a torsional  section in their spectral covers. As mentioned before, 
heterotic/F-theory duality suggests that the structure group of the heterotic vector bundle should contain a discrete part.

\subsubsection{Structure group $\mathbb{Z}_2$: $\text{E}_8 \times \text{E}_7 \times \text{SU}(2)$ gauge symmetry} \label{structuregroupZ2}

We consider a model which arises by the following specialization of coefficients in \eqref{genericsectionP112}:
\begin{center}
\begin{tabular}{|c|c|c|c|}\hline
Coefficient & $X$ & $X^-$ & $X^+$ \\ \hline
$s_1$ & $s_{13} V^2$ & $s_{13}\mu$ & $0$ \\
$s_{2}$ & $s_{22} U V$ & $s_{22} U$ & $s_{22} V$ \\
$s_{3}$ & $s_{32} U V$ & $s_{32} U$ & $s_{32} V$ \\
$s_{4}$ & $s_{41} U^2+ s_{42} U V$ & $s_{42} U$ & $s_{42} V+s_{41}\lambda_1$ \\
$s_{5}$ & $0$ & $0$ & $0$ \\
$s_{6}$ & $s_{62} U V$ & $s_{62} U$ & $s_{62} V$ \\
$s_{7}$ & $0$ & $0$ & $0$ \\
$s_{8}$ & $s_{82} U V$ & $s_{82} U$ & $s_{82} V$ \\ \hline
\end{tabular}
\end{center}

\begin{figure}[t]
\begin{center}
\includegraphics[width=5cm]{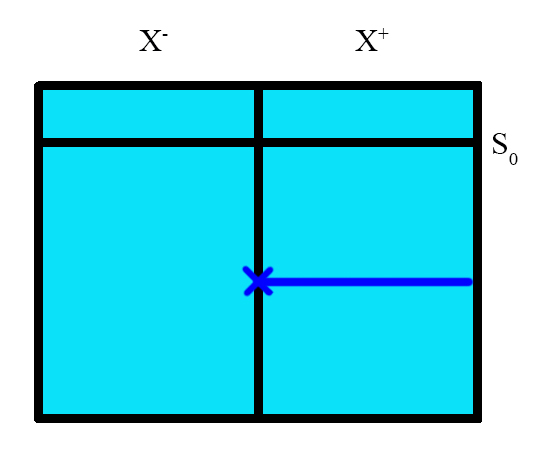}
\end{center}
\caption{The stable degeneration limit of a K3 surface with $\text{E}_8 \times (\text{E}_7 \times \text{SU}(2))/ \mathbb{Z}_2$. The half K3 surface $X^-$ has trivial Mordell-Weil group, while the half K3 surface $X^+$ has a torsional 
Mordell-Weil group $\mathbb{Z}_2$.}
\label{E8xE7xSU(2)} 
\end{figure}

We start the analysis with the gauge group on the F-theory side first. There are three singular loci of the fibration at $U=0$, $V=0$ and $s_{41}U+s_{42}V$. The evaluation of the Tate coefficients reveals the Tate vectors
\begin{equation}
\vec{t}_{U} = \left( 1 , 2 , \infty , 4 , 5 , 10\right)\, , \qquad t_V = \left( 1 , 2 , \infty , 3 , 5 , 9 \right) \, , \qquad \vec{t}_{s_{41}U+s_{42}V} = \left( 0 , 0 , \infty , 1 , 2, 2 \right) \, .
\end{equation}
Thus, there are an $\text{E}_8$ singularity as well as an $\text{E}_7$ and an SU(2) singularity. The E$_8$ singularity is inherited by the half K3 surface $X^-$ while $X^+$ gets endowed with an E$_7$ and an SU(2) singularity.

\begin{figure}[t]
\begin{center}
\includegraphics[width=5.2cm]{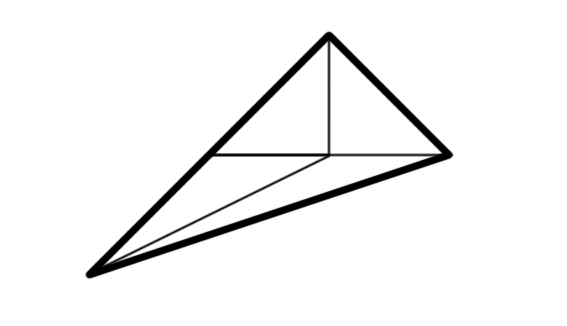}\hspace{1cm}
\includegraphics[width=2.8cm]{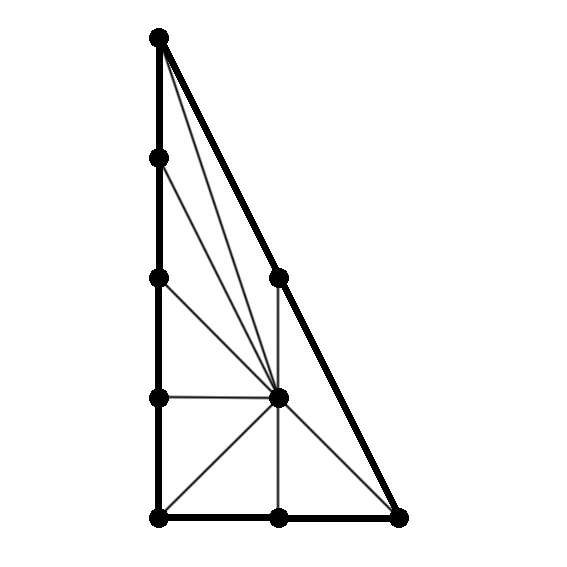}
\end{center}
\caption{The left picture shows the specialized two-dimensional polytope $\Delta_2$ corresponding to the half K3 surface $X^+$. The right figure shows its dual, $\Delta_2^\circ$, which specifies the ambient space of the elliptic fiber of $X^+$.}
\label{specialized2dambientspace}
\end{figure}

As a next step, we observe that there is only one section given by $[1:0:0]$ in the half K3 surface $X^-$ and two sections given 
by $[x_1:x_2:x_3]=[0:1:0]$ and $[x_1:x_2:x_3]=[1:0:0]$ in the half K3 surface $X^+$. Here, we work again in the limit 
$x_4=x_5=1$. In contrast, the full K3 surface has only one section namely the point at infinity. Moreover, a transformation into 
Weierstrass coordinates shows that the generic section $S_1^{X^+}$ has specialized into a torsional section of order two
as can be checked using the results of \cite{Aspinwall:1998xj}. This is 
expected, as the centralizer of the gauge algebra\footnote{To be precise, $\text{E}_8$ only contains the group 
$(\text{E}_7 \times \text{SU}(2))/ \mathbb{Z}_2$ as a subgroup. } $\text{E}_7 \times \text{SU}(2)$ within $\text{E}_8$ is given by 
$\mathbb{Z}_2$, which is also expected from the general analysis of \cite{OguisoShioda}. In contrast, the full K3 surface $X$ 
does not seem to exhibit a torsional section of order two.

Finally, we turn towards the analysis of the gauge group from the heterotic side. Here, the spectral cover is given by
\be
p^- = s_{13} x_1^4, \qquad p^+ = s_{41} x_1 x_2^3 \, .
\ee
After transformation to Weierstrass normal coordinates $p^-_W$ is given by a constant which has no common solution with the elliptic curve. In contrast, the transformed quantity $p_W^+$ gives rise to the point 
\begin{equation}
[x:y:z] = \left[\frac{1}{3} \left(\frac{s_{62}^2}{4}-s_{32}s_{82} \right):0:1 \right]\,.
\end{equation}
which is a torsion point of order two. In other words we see that the spectral cover is just given by a torsional point. 

In 
\cite{Aspinwall:1998xj} it has been suggested that an F-theory compactification with a torsional section in an 
elliptically fibered Calabi-Yau manifold and its stable degeneration limit should be dual to pointlike instantons with discrete holonomy on the heterotic side. Due to the similarity to the considered example, we propose that the spectral cover $p_W^+$
is to be interpreted as describing such a pointlike instanton with discrete holonomy.  In addition, as pointed out above, 
the matching of gauge symmetry on both sides of the duality only works if the spectral cover $p^+_W$ is interpreted in this way. It would be important to confirm this proposal further by a more detailed analysis of the spectral cover, computation
of the heterotic tadpole, or an analysis of codimension two singularities in F-theory.

\subsubsection{Structure group S(U(2) $\times \mathbb{Z}_2$): E$_8$ $\times$ E$_6$ $\times$ $\text{U}(1)$ gauge symmetry} \label{E8E6U1model}

In this section we present another example whose spectral cover polynomial containing a torsional point and leading to an E$_8\times\text{E}_6\times\text{U}(1)$ gauge symmetry. As one E$_8$ factor is left intact, the $\text{U}(1)$ factor needs to be embedded solely into one E$_8$ bundle.  

The starting point of our analysis  is the following specialization of  coefficients in \eqref{genericsectionP112}:
\begin{center}

\begin{tabular}{|c|c|c|c|}\hline
Coefficient & $X$ & $X^-$ & $X^+$ \\ \hline
$s_1$ & $s_{13} V^2$ & $s_{13}\mu$ & $0$ \\
$s_{2}$ & $s_{22} U V$ & $s_{22} U$ & $s_{22} V$ \\
$s_{3}$ & $s_{32} U V$ & $s_{32} U$ & $s_{32} V$ \\
$s_{4}$ & $s_{41} U^2+ s_{42} U V$ & $s_{42} U$ & $s_{42} V + s_{41} \lambda_1$ \\
$s_{5}$ & $0$ & $0$ & $0$ \\
$s_{6}$ & $s_{62} U V$ & $s_{62} U$ & $s_{62} V$ \\
$s_{7}$ & $s_{71} U^2$ & $0$ & $s_{71} \lambda_1$ \\
$s_{8}$ & $s_{82} U V$ & $s_{82} U$ & $s_{82} V$ \\ \hline
\end{tabular}
\end{center}

As in the previous cases, we compute the orders of vanishing of the Tate coefficients in order to determine the gauge group on 
the F-theory side. The computed Tate vectors signal an E$_8$ symmetry at $U=0$ and an E$_6$ symmetry at $V=0$. 
As a next step, we investigate the rational sections of $X$. As the coefficient $s_7$ does not vanish for the full K3 surface, there 
are the two generic sections $S_0, S_1$ realized in this model. However, the half K3  surface $X^-$ only has the zero section $S_0$. In contrast, the half K3 surface $X^+$  has two sections given by $S_0$, $S_1$, which unify in the heterotic elliptic curve and continue as one section into the other half K3 surface, see Figure \ref{E8xE6xU(1)}. This behavior of rational sections explains the origin of the $\text{U}(1)$-factor from the gluing condition discussed in Section \ref{gluing}.

\begin{figure}[t]
\begin{center}
\includegraphics[width=5cm]{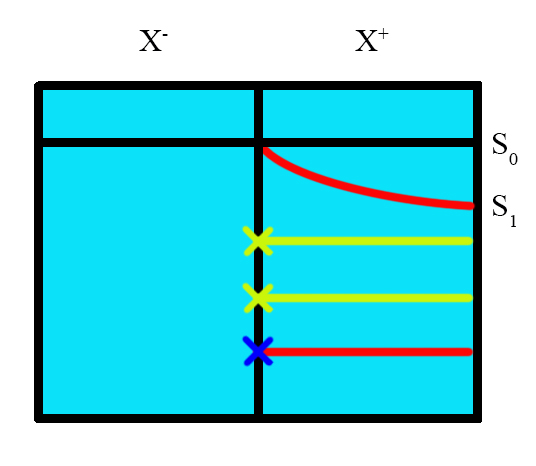}
\end{center}
\caption{The half K3 surface $X^-$ exhibits only the zero section, while  the half K3 surface $X^+$ has also the section $S_1^{X^+}$ which merges with the section $S_0^{X^-}$ along the heterotic geometry. Thus there are two independent sections in the full K3 surface giving rise to a U(1) gauge group factor. In addition, the inverse of $S_1^{X^+}$ becomes a torsion point of order two when hitting the heterotic geometry.}
\label{E8xE6xU(1)} 
\end{figure}

As a further step, we investigate how this $\text{U}(1)$ factor is reflected in the spectral cover on the heterotic side. 
The spectral cover polynomials computed by stable degeneration read
\be
p^- = s_{13} x_1^4, \qquad p^+ = s_{41} x_1 x_2^3 + s_{71} x_2^2 x_3\,.
\ee
The interpretation of $p^-$ is as in all the other cases just a trivial spectral cover. The common solution to $p^+$ and the 
heterotic Calabi-Yau manifold $Z$ is given by a pair of irrational points $R_1, R_2$ as well as a further point $T_t$ which has in 
Weierstrass normal form coordinates
\be
T_t =  \left[\frac{1}{3} \left(\frac{1}{4}s_{62}^2 - s_{32} s_{82} \right) :0 :1 \right]\,.
\ee
Thus, it is a torsion point of order two. However, it does not extend as a full torsional section into 
the half K3 surface $X^+$. The corresponding section is rather the inverse of $S_1$. 

Again we see that the split spectral cover $p^+$ contains a torsional section.  Let us comment on the interpretation of this for 
the structure group of the heterotic vector bundle.
Heterotic/F-theory duality implies that the low-energy effective theory contains a massless $\text{U}(1)$-symmetry.
However, as we have seen in Section \ref{sec:Stueckelberg}, a $\text{U}(1)$ background bundle in the heterotic theory
has a non-trivial field strength and thus a 
non-vanishing first Chern class, which would yield a massive U(1) in the effective field theory. Thus, we can not interpret
the torsional component $T_t$ to the spectral cover as a U(1) background bundle. 
By the arguments of Section \ref{SUNellipticcurve} and the similarity to the setups considered in  \cite{Aspinwall:1998xj},  it is 
tempting to identify this torsional component $T_t$ as a pointlike heterotic instanton with discrete holonomy.
In order to justify this statement, it would be necessary to compute the first Chern class of a heterotic line bundle that
is defined in terms of components to the cameral cover given by rational sections of the half K3 fibrations arising in stable degeneration. In  \cite{Aspinwall:2005qw}, it has been argued that the first Chern class is given, up to vertical components,
by the difference of the rational section and the zero section. If the first Chern class were completed into
the Shioda map of the rational section, which we conjecture to be the case, it would be zero precisely for a torsional section \cite{Mayrhofer:2014opa}.
Consequently, the U(1) in the commutant of $\text{E}_8$ would remain massless as the gauging in \eqref{eq:LStueck}  
would be absent. It would be important  to
confirm this conjecture by working out the missing vertical part in the formula for the first Chern class
of a U(1) vector bundle.

\subsection{$\text{U}(1)$ factors arising from purely non-Abelian structure groups}
\label{sec:nonAbelian}

In this final section, we present an example in which the heterotic vector bundle  has only purely  
non-Abelian structure group, while the F-Theory gauge group analysis clearly signals a $\text{U}(1)$ factor. 

As in the previous cases, we start by specifying the specialization of the coefficients in the defining hypersurface equation for
$X$:
\begin{center}
\begin{tabular}{|c|c|c|c|}\hline
Coefficient & $X$ & $X^-$ & $X^+$ \\ \hline
$s_{1}$ & $s_{12} UV + s_{13} V^2$ & $s_{12} U + s_{13} \mu$ & $s_{12} V$ \\
$s_{2}$ & $s_{22} UV$ & $s_{22} U$ & $s_{22} V$ \\
$s_{3}$ & $s_{32} UV$ & $s_{32} U$ & $s_{32} V$ \\
$s_{4}$ & $s_{42} UV$ & $s_{42} U$ & $s_{42} V$ \\
$s_{5}$ & $s_{52} UV$ & $s_{52} U$ & $s_{52} V$ \\
$s_{6}$ & $s_{62} UV$ & $s_{62} U$ & $s_{62} V$ \\
$s_{7}$ & $s_{71} U^2$ & $0$ & $s_{71} \lambda_1$ \\
$s_{8}$ & $s_{82} UV$ & $s_{82} U$ & $s_{82} V$ \\ \hline
\end{tabular}
\end{center}
We determine the gauge symmetry of the F-theory side by analysis of the Tate coefficients. We obtain the Tate vectors
\be
\vec{t}_U = \left( 1 , 2 , 3 , 4 , 5 , 10 \right)\,  \qquad \vec{t}_V = \left(1 , 1 , 2 , 2 , 4 , 6 \right)\, ,
\ee
which reveals an E$_8$ singularity at $U=0$ and an SO(7) singularity at $V=0$. We note that it is not directly possible to 
distinguish an SO(7) singularity from an SO(8) singularity using the Tate table \ref{Tatetable} only. To confirm that the type of 
singularity is indeed SO(7) we have to investigate the monodromy cover \cite{Grassi:2011hq} which is for an $I_0^*$ fiber given 
by
\be
A:\qquad \psi^3 +  \left(\left. \frac{f}{v^2} \right|_{v=0} \right) \psi +\left(\left. \frac{g}{v^3} \right|_{v=0} \right)\, .
\ee
Here, $v$ is the affine coordinate $V/U$ and $f,g$ are the Weierstrass coefficients. An $I_0^*$ fiber is SO(7) if the monodromy 
cover $A$ factors into a quadratic and a linear constraint,  which is indeed the case for the example at hand.
\begin{figure}[t]
\begin{center}
\includegraphics[width=5cm]{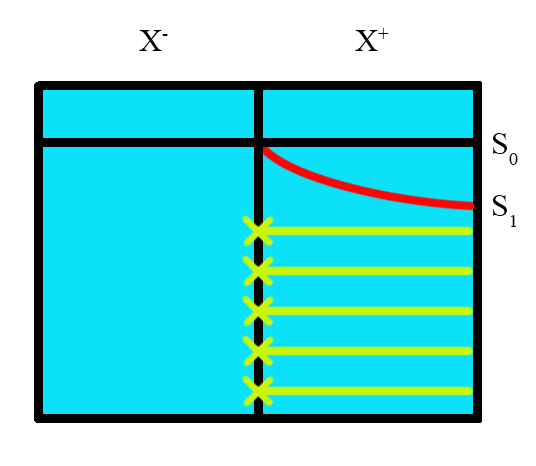}
\end{center}
\caption{The half K3 surface $X^-$ exhibits only the zero section, while  the half K3 surface $X^+$ has also the section $S_1^{X^+}$ which merges with the section $S_0^{X^-}$ in the heterotic geometry. Thus, there are two independent sections in the full K3 surface giving rise to a U(1) gauge group factor.}
\label{E8xSO7xU(1)} 
\end{figure}

The stable degeneration limit yields two half K3 surfaces, $X^+$ and $X^-$, cf.~Figure \ref{E8xSO7xU(1)}. There only exists the 
zero section in $X^-$. In contrast, $X^+$ has two sections which are given by $S_0^{X^+}$ and $S_1^{X^+}$. As in the 
previously considered case in Section \ref{E8E6U1model}, $S_1^{X^+}$ unifies with $S_0^{X^-}$ on the heterotic elliptic curve. 
Thus, there are two global sections in the full K3 surface and therefore a $\text{U}(1)$ factor in the F-theory compactification. 

Turning towards the discussion of the heterotic gauge bundles, one finds that the $\text{U}(1)$ factor is encoded in the data of 
the spectral cover polynomial as follows.
We observe that the spectral covers following  $X^+$ and $X^-$, respectively, are given by
\be
p^- = s_{13} x_1^4\, , \qquad p^+ = s_{71} x_2^2 x_3\, .
\ee
The intersection of its Weierstrass transform $p^+_W$ with the heterotic elliptic curve gives five irrational points $R_1, R_2, T_1, T_2, T_3$ with $R_1 \oplus R_2 = 0$ and $T_1 \oplus T_2 \oplus T_3 =0$. Thus we have a heterotic vector bundle with SU(2) $\times$ SU(3). As the spectral cover $p^+$ has one free parameter only, namely $s_{71}$, this model does not seem to have any moduli.

As our understanding of the precise embedding of the structure group into E$_8$ is limited, we have checked all possible ways to embed the group SO(7) $\times$ SU(2) $\times$ SU(3) into E$_8$. Independently of the chosen embedding, there is always a $\text{U}(1)$ in all possible breakings. Thus, we are led to conclude that the centralizer of SU(2) $\times$ SU(3) necessarily produces a $\text{U}(1)$ factor which matches with the F-theoric analysis.

\section{Conclusions and Future Directions}
\label{sec:conclusions}

In this paper  we have presented a first explicit analysis of the origin of Abelian gauge symmetries  
for string theory compactifications within the duality between the 
$\text{E}_8\times \text{E}_8$ heterotic string and F-theory. Here we summarize  the framework of the analysis,  highlight the key advancements, and conclude with future directions.
~
\\

\noindent {\bf Framework}\\

{\noindent We have focused on  F-theory compactifcations  with a rank one  Mordell-Weil group 
of rational sections both for compactifications to D=8 and D=6. 
We have systematically  studied a broad class  F-theory compactifications on elliptically 
fibered Calabi-Yau $(n+1)$-folds (with $n=1,2$, respectively) with rational sections and 
rigorously performed the stable degeneration limit  to  dual heterotic compactifications on 
elliptically fibered Calabi-Yau $n$-folds.  All considered examples are toric hypersurfaces and 
the stable degeneration limit is performed as a toric symplectic cut.}

The key aspects of the analysis are the following:
\begin{itemize}

\item
We  have carefully investigated the solutions of the spectral cover polynomial and the 
hypersurface for the heterotic elliptically fibered Calabi-Yau manifold. 
We have used the group law of the elliptic 
curve in Weierstrass normal form in order to determine the structure group of the heterotic 
background bundle.

\item
We have analyzed the origin of the resulting gauge  group. In D=6 this involves incorporation 
of  the massive $\text{U}(1)$ gauge symmetries, due to the heterotic
St\"uckelberg mechanism, that are not visible in F-theory. 
 
 \end{itemize}

\noindent {\bf  Key Results}\\

{\noindent
While the F-theory side provides a unifying treatment of Abelian gauge symmetries, as encoded in 
the Mordell-Weil group of elliptically fibered Calabi-Yau $(n+1)$-folds, a detailed analysis of a 
broad classes of toric F-theory compactifications has resulted in the proposal of three different 
classes of heterotic duals that give rise to $\text{U}(1)$ gauge group factors: }
\begin{itemize}
\item  Split spectral covers describing bundles with $\text{S(U({\it m})}\times \text{U}(1))$ structure group. Examples of this type have been discussed in Section \ref{sec:splitcover}.
\item Spectral covers containing torsional sections  giving  rise to bundles with 
$\text{SU({\it m})}\times\mathbb{Z}_n$ structure group. Classes of examples with this structure group
have been presented in Section \ref{sec:torsionalcover}.
\item  The appearance of bundles with structure groups of the type $\text{SU({\it m})}\times \text{SU({\it n})}$ whose commutants inside $\text{E}_8$  contain a $\text{U}(1)$-factor.  Explicit examples of
this form can be found in Section \ref{sec:nonAbelian}.

 \end{itemize}
 \newpage
 \noindent {\bf  Future Directions}\\
 
 {\noindent 
While the work presents a pioneering effort, addressing comprehensively the origin of Abelian 
gauge group factors in heterotic/F-theory duality for a class of compactifications,  the analysis provides a stage for further studies, 
both by extending the systematics of the analysis and by further detailed studies of the dual 
heterotic geometry and vector bundle data.}
 
 \begin{itemize}

 \item  It would be  important to extend  the studies to examples  within  larger classes of pairs 
 of dual  toric varieties as well as of more general  elliptically fibered Calabi-Yau manifolds, 
 respectively. In particular, this would allow to account for studies of dual geometries with 
broader classes of complex structure moduli spaces, and thus for an analysis of more general  
spectral covers of  dual heterotic vector bundles. In D=6 our analysis has been limited to a 
specific elliptically fibered Calabi-Yau $(n+1)$-folds, which has resulted in  constrained 
appearances of non-Abelian gauge symmetries and additional $\text{U}(1)$'s.   In particular, it 
would be illuminating to elaborate on the stable degeneration limit for general toric
fibrations of two-dimensional polyhedra over $\mathbb{P}^1$ in eight dimensions and, in 
addition, over Hirzebruch surfaces in six dimensions. 

\item It would be interesting to have the tools to study the spectral cover directly in the 
$\text{Bl}_1 \mathbb{P}^{(1,1,2)}$ model or more generally for fiber geometries which are 
given by the sixteen two-dimensioal reflexive polyhedra. This would require in particular a notion 
of the group law for these representations of elliptic curves. 

 \item The study of the properties of the spectral cover was  primarily confined to the 
derivation of the resulting gauge symmetries and the structure groups of the heterotic vector 
bundles. Further analysis 
of the spectral cover in compactifications to D=6  (and extensions to D=4) is needed; it should 
shed light on the further spectral cover data, which enter Chern classes, anomaly cancellation 
and matter spectrum calculations. This study is complicated by the resolution of singularities of 
the heterotic geometry that may have to be performed, resulting in spectral covers, which are
not finite \cite{Donagi:1999ez}.

 \item  Our analysis has been primarily constrained to studies of Abelian gauge symmetries
in the language of a perturbative heterotic dual. Although we have encountered spectral covers 
which seem to describe small instantons, i.e.~non-perturbative M5-branes, with discrete 
holonomy, we have not systematically analyzed their effect. In F-theory, M5-branes are visible as 
non-minimal singularities which occur at co-dimension two loci that have to be blown up.  
It would be interesting to thoroughly perform this geometric analysis.
We expect in addition rich structures of Abelian gauge symmetry factors in  F-theory whose 
heterotic duals are due to other types of non-perturbative M5-branes. 
Furthermore, it would be interesting to study the geometric transitions between F-theory 
geometries with different numbers of tensor multiplets, 
whose discussion is again related to this resolution process.

  \end{itemize}

\vskip 0.5cm
\newpage

\noindent  {\bf Acknowledgments}\\
\vskip 0.2cm
{\noindent 
We thank L. Anderson, R. Donagi, J. Gray, T. Grimm, A. Klemm, D. Morrison,  R. Pardini, S. Sch\"afer-Nameki, V. Perduca  and W. Taylor  for useful discussions.
We are grateful to the Theory Division of CERN  (M.C. and M.P.), the 2015 Summer  Program on
``F-Theory at the interface of particle physics and mathematics''  at the Aspen Center for Physics (M.C., A.G., D.K. and M. P.)
and the UPenn (D.K.)  for  hospitality during the course of the project.  
This research is supported in part by the DOE  Grant Award DE-SC0013528, (M.C., M.P., P.S.), UPenn School of Arts and Sciences Funds  for Faculty Working Group (A.G. and M.C.), the Fay R. and Eugene L.~Langberg Endowed Chair (M.C.) and the Slovenian Research Agency (ARRS) (M.C.).}

\appendix

\section{Weierstrass and Tate form of the hypersurface $\chi^{\text{sing}}$} \label{appendix:TateandWeierstrass}

In this appendix, we summarize the Weierstrass normal form as well as the Tate coefficients of the $\chi^{\text{sing}}$ model. For convenience, we recall the most general form of the hypersurface $\chi^{\text{sing}}$ which reads
\begin{equation}
\chi^{\text{sing}} : = s_1x_1^4+s_{2}x_1^3x_2+s_{3}x_1^2x_2^2+s_{4}x_1x_2^3+s_{5}x_1^2x_3  +s_{6}x_1x_2x_3  +s_{7}x_2^2x_3  +s_{8}x_3^2  =0, \quad s_i \in \mathcal{O}_{\mathbb{P}^1}(2) \, . \label{chisingappendix}
\end{equation}
This can be brought in the so-called Tate form
\begin{equation}
y^2 + a_1 x y + a_3 y = x^3 + a_2 + a_4 x + a_6\, .
\end{equation}
The Tate coefficients are explicitly given as \cite{Klevers:2014bqa}
\begin{eqnarray}
a_1 &=& s_6\, , \nn \\
a_2 &=& -s_5 s_7 -s_3 s_8 \, , \nn \\
a_3 &=& -s_4 s_5 s_8 -s_2 s_7 s_8 \, , \nn \\  
a_4 &=& s_3 s_5 s_7 s_8 + s_1 s_7^2 s_8 + s_2 s_4 s_8^2 \, , \nn \\ 
a_6 &=& -s_1 s_3 s_7^2 s_8^2 -s_1 s_4^2 s_8^3 + s_4 s_7 \left(-s_2 s_5 s_8^2 + s_1 s_6 s_8^2 \right) \, . \label{Tatecoefficients1}
\end{eqnarray}
In addition, it is useful, to introduce the quantities
\begin{eqnarray}
b_2 &=& a_1^2 + 4 a_2 \, , \nn \\ 
b_4 &=& a_1 a_3 + 2 a_4\, , \nn \\ 
b_6 &=& a_3^2 + 4 a_6 \, .
\end{eqnarray}
The Weierstrass normal form of $\chi^{\text{sing}}$ reads
\begin{eqnarray} \label{g2g3}
f &=& \left( -\frac{1}{48}s_{62}^4 + \frac{1}{6} s_{52} s_{62}^2 s_{72} - \frac{1}{3} s_{52}^2 s_{72}^2 - 
 \frac{1}{2} s_{42} s_{52} s_{62} s_{82} + \frac{1}{6} s_{32} s_{62}^2 s_{82} \right. \nn \\ && \left.+ \frac{1}{3} s_{32} s_{52} s_{72} s_{82} - 
 \frac{1}{2} s_{22} s_{62} s_{72} s_{82} + s_{21} s_{72}^2 s_{82} - \frac{1}{3} s_{32}^2 s_{82}^2 + s_{22} s_{42} s_{82}^2 \right). \nonumber\\
  \nn \\
g &=&  \left(  \frac{1}{864}s_{62}^6 - \frac{1}{72} s_{52} s_{62}^4 s_{72} + \frac{1}{18} s_{52}^2 s_{62}^2 s_{72}^2 - 
 \frac{2}{27} s_{52}^3 s_{72}^3 + \frac{1}{24} s_{42} s_{52} s_{62}^3 s_{82} \right. \nn \\ && \left. - \frac{1}{72} s_{32} s_{62}^4 s_{82} - 
 \frac{1}{6} s_{42} s_{52}^2 s_{62} s_{72} s_{82} + \frac{1}{36} s_{32} s_{52} s_{62}^2 s_{72} s_{82} + 
 \frac{1}{24} s_{22} s_{62}^3 s_{72} s_{82} \right. \nn \\ && \left. + \frac{1}{9} s_{32} s_{52}^2 s_{72}^2 s_{82} - 
 \frac{1}{6} s_{22} s_{52} s_{62} s_{72}^2 s_{82} - \frac{1}{12} s_{21} s_{62}^2 s_{72}^2 s_{82} + 
 \frac{1}{3} s_{21} s_{52} s_{72}^3 s_{82} \right. \nn \\ && \left. + \frac{1}{4} s_{42}^2 s_{52}^2 s_{82}^2 - 
 \frac{1}{6} s_{32} s_{42} s_{52} s_{62} s_{82}^2 + \frac{1}{18} s_{32}^2 s_{62}^2 s_{82}^2 - 
 \frac{1}{12} s_{22} s_{42} s_{62}^2 s_{82}^2 \right. \nn \\ && \left. + \frac{1}{9} s_{32}^2 s_{52} s_{72} s_{82}^2 - 
 \frac{1}{6} s_{22} s_{42} s_{52} s_{72} s_{82}^2 - \frac{1}{6} s_{22} s_{32} s_{62} s_{72} s_{82}^2 + 
 s_{21} s_{42} s_{62} s_{72} s_{82}^2 \right. \nn \\ && \left.+ \frac{1}{4} s_{22}^2 s_{72}^2 s_{82}^2 - 
 \frac{2}{3} s_{21} s_{32} s_{72}^2 s_{82}^2 - \frac{2}{27} s_{32}^3 s_{82}^3 + 
 \frac{1}{3} s_{22} s_{32} s_{42} s_{82}^3 - s_{21} s_{42}^2 s_{82}^3 \right)\,.
\end{eqnarray}
In particular, the discriminant reads
\begin{equation}\label{eq:Delta}
\Delta = 4f^3 + 27 g^2 = \frac{1}{48} s_{82}^2 \left(\dots \right)\,.
\end{equation}
where the expression in the bracket denotes a generic polynomial. 
\subsection{The map to Weierstrass normal form}

In this subsection we discuss the bi-rational map of \eqref{chisingappendix}  to Weierstrass normal form. As a first step, we transform \eqref{genericsectionP112} into the form
\be
\tilde{s}_1x_1^4+\tilde{s}_{2}x_1^3x_2+\tilde{s}_{3}x_1^2x_2^2+\tilde{s}_{4}x_1x_2^3 + {s}_{7}x_2^2x_3  + x_3^2  =0\, . \label{intermediatestepWeierstrass}
\ee
Here, we have introduced the new quantities
\be
\tilde{s}_1 = -\frac{1}{4}s_5^2 + S_0 s_8, \quad \tilde{s}_2 = -\frac{1}{2}s_5 s_6 + S_1 s_8, \quad \tilde{s}_3 = -\frac{1}{4}s_6^2 - \frac{1}{2}s_5 s_7 + s_3 s_8, \quad \tilde{s}_4 = -\frac{1}{2}s_6 s_7 + s_4 s_8
\ee
Next, one uses the transformations provided in \cite{Morrison:2012ei}
\begin{eqnarray}
x_1 &\longmapsto& z \nn \\
x_2 &\longmapsto& \frac{6 {s}_7 y + 6 \tilde{s}_4 x z + 2 \tilde{s}_3 \tilde{s}_4 z^3 + 
 3 \tilde{s}_2 s_7^2 z^3}{2 (3 s_7^2 x - 3 \tilde{s}_4^2 z^2 - 2 \tilde{s}_3 s_7^2 z^2)} \nn \\
x_3 &\longmapsto& \left(108 s_7^3 x^3 - 108 s_7^3 y^2 - 108 \tilde{s}_4 s_7^2 x y z - 
   216 \tilde{s}_4^2 s_7 x^2 z^2 - 108 \tilde{s}_3 s_7^3 x^2 z^2 - 108 \tilde{s}_4^3 y z^3 \right. \nn \\ && \left. - 
   144 \tilde{s}_3 \tilde{s}_4 s_7^2 y z^3 - 108 \tilde{s}_2 s_7^4 y z^3 - 36 \tilde{s}_3 \tilde{s}_4^2 s_7 x z^4 - 
   54 \tilde{s}_2 \tilde{s}_4 s_7^3 x z^4 + 12 \tilde{s}_3^2 \tilde{s}_4^2 s_7 z^6 \right. \nn \\ && \left. - 54 \tilde{s}_2 \tilde{s}_4^3 s_7 z^6 + 
   16 \tilde{s}_3^3 s_7^3 z^6 - 72 \tilde{s}_2 \tilde{s}_3 \tilde{s}_4 s_7^3 z^6 - 
   27 \tilde{s}_2^2 s_7^5 z^6\right) \Big/ \nn \\ &&  12 \left(3 s_7^2 x - 3 \tilde{s}_4^2 z^2 - 2 \tilde{s}_3 s_7^2 z^2\right)^2 \label{transfP112Weier}
\end{eqnarray}
in order to finally bring \eqref{intermediatestepWeierstrass} into Weierstrass normal form in $\mathbb{P}^{(1,2,3)}$.
We also note that the transformations \eqref{transfP112Weier} simplify in the case $s_7=0$, in particular their denominators loose their dependence on $x,y$. 

\section{Spectral Cover Examples with no $\text{U}(1)$}
\label{app:noU1s}

For convenience and to demonstrate how our formalism works in a well-understood situation, we analyze several examples with pure non-Abelian gauge content only. These are related to the examples \ref{E7E7U1}, \ref{E7E6U1} and \ref{E7SO9U1} by a Higgsing process which gives $s_{72}$ a vacuum expectation value. 
\Mnote{Check this statement about the vev.}

\subsection{Trivial structure group: E$_8$ $\times$ E$_8$ gauge symmetry} \label{E8E8model}

As described in the previous section, we can obtain examples with higher rank gauge symmetry by specializing the coefficients of $chi_{\text{sing}}$. Aiming for a model with E$_8$ $\times$ E$_8$ gauge symmetry, one obtains the following coefficients. 

\begin{center}
\begin{tabular}{|c|c|c|c|}\hline
Coefficient & K3 & $X^-$ & $X^+$ \\ \hline
$s_1$ & $s_{11} U^2 +s_{12} UV + s_{13}V^2$ & $s_{12}U+s_{13} \mu$ & $s_{12}V+s_{11} \lambda_1$ \\
$s_{2}$ & $s_{22} U V$ & $s_{22} U$ & $s_{22} V$ \\
$s_{3}$ & $s_{32} U V$ & $s_{32} U$ & $s_{32} V$ \\
$s_{4}$ & $s_{42}  U V$ & $s_{42} U$ & $s_{42} V$ \\
$s_{5}$ & $s_{52} UV$ & $s_{52} U$ & $s_{52} V$ \\
$s_{6}$ & $s_{62} U V$ & $s_{62} U$ & $s_{62} V$ \\
$s_{7}$ & $0$ & $0$ & $0$ \\
$s_{8}$ & $s_{82} U V$ & $s_{82} U$ & $s_{82} V$ \\ \hline
\end{tabular}
\end{center}

Here the second row displays the coefficients of the full K3 surface while the coefficients of the two half K3 surfaces are displayed in row three and four. In particular, one notices that the coefficient $s_7$ is missing which means that one is passing from the toric ambient space $Bl_1 \mathbb{P}^{(1,1,2)} \times \mathbb{P}^1$ to the ambient space $\mathbb{P}^{(1,2,3)} \times \mathbb{P}^1$. Clearly, a generic section of the anti-canonical bundle of $\mathbb{P}^{(1,2,3)}$ does not have a second section, so there is also no reason to expect any $\text{U}(1)$. 

We proceed by analyzing the F-Theory gauge group. The analysis of the Tate vectors reveals that 
\be
\vec{t}_U = \vec{t}_V = \left(1,2,3,4,5,10 \right)
\ee
and thus there is an $\text{E}_8 \times \text{E}_8$ gauge symmetry. After the stable degeneration limit, both half K3 surfaces $X^+$ and $X^-$ obtain one E$_8$ singularity each. 

Finally, we turn to the Heterotic side. The splitting of the two half K3's into the Heterotic elliptic curve and the spectral cover contributions reveals that
\be
p^+ = s_{11}x_1^4, \qquad p^- = s_{13} x_1^4 \, .
\ee
After transforming these expression into the affine Weierstrass coordinates $x,y$, one obtains
\be
p^+_W = s_{11}, \qquad p^-_W = s_{13}
\ee
In both cases,one obtains an S$\text{U}(1)$ spectral cover. However, the centralizer of the identity in E$_8$ is E$_8$ and one obtains a perfect match with the F-theory calculation.

\subsection{Structure group SU(1) $\times$ SU(2): E$_8$ $\times$ E$_7$ gauge symmetry} \label{E8E7model}
We consider the following model which is specified by the following coefficients in \eqref{genericsectionP112}.
\begin{center}

\begin{tabular}{|c|c|c|c|}\hline
Coefficient & K3 & $X^-$ & $X^+$ \\ \hline
$s_1$ & $s_{11} U^2 +s_{12} UV + s_{13}V^2$ & $s_{12}U+s_{13} \mu$ & $s_{12}V+s_{11} \lambda_1$ \\
$s_{2}$ & $s_{21}U^2 + s_{22} U V$ & $s_{22} U$ & $s_{22} V + s_{21} \lambda_1$ \\
$s_{3}$ & $s_{32} U V$ & $s_{32} U$ & $s_{32} V$ \\
$s_{4}$ & $s_{42}  U V$ & $s_{42} U$ & $s_{42} V$ \\
$s_{5}$ & $s_{52} UV$ & $s_{52} U$ & $s_{52} V$ \\
$s_{6}$ & $s_{62} U V$ & $s_{62} U$ & $s_{62} V$ \\
$s_{7}$ & $0$ & $0$ & $0$ \\
$s_{8}$ & $s_{82} U V$ & $s_{82} U$ & $s_{82} V$ \\ \hline
\end{tabular}
\end{center}
This time, we obtain the following Tate vectors
\be
\vec{t}_V = \left(1,2,3,3,5,9 \right), \qquad \vec{t}_U = \left(1,2,3,4,5,10 \right)
\ee
which signal an E$_7$ singularity at $V=0$ as well as an E$_8$ singularity at $U=0$. The former one is inherited by the half K3 surface $X^-$ while the latter one moves into $X^+$.

The spectral cover is in this case given by
\be
p^+ = x_1^3  \left(s_{11} x_1  + s_{21} x_2 \right) , \qquad p^- = s_{11}x_1^4 \, .
\ee
We only comment on the non-trivial spectral cover. After applying the transformation \eqref{transfP112Weier}, it reads
\be
p^+_W = c_0 + c_1 x
\ee
which defines an SU(2) spectral cover and is precisely what is expected. Explicitly, the $a_i$'s read
\be
 c_0 = s_{21} s_{62}^2 - 4 s_{21} s_{32} s_{82} + 12 s_{11} s_{42} s_{82}  \qquad c_1 = -s_{21}
\ee
Note that the $a_i$ are indeed proportional to $s_{11}$, $s_{21}$ which define the spectral cover. Thus, we obtain an SU(2) spectral cover in the case of $X^+$ and a trivial structure group for the case of $X^-$. In conclusion, there is a perfect match with the F-theory analysis.


\subsection{Example with gauge group E$_8$ $\times$ SO(11)} \label{E8SO11model}

We consider the following model which is specified by the following coefficients in \eqref{genericsectionP112}.
\begin{center}

\begin{tabular}{|c|c|c|c|}\hline
Coefficient & K3 & $X^-$ & $X^+$ \\ \hline
$s_1$ & $s_{11} U^2 +s_{12} UV + s_{13}V^2$ & $s_{12}U+s_{13} \mu$ & $s_{12}V+s_{11} \lambda_1$ \\
$s_{2}$ & $s_{21}U^2 + s_{22} U V$ & $s_{22} U$ & $s_{22} V + s_{21} \lambda_1$ \\
$s_{3}$ & $s_{31}U^2 + s_{32} U V$ & $s_{32} U$ & $s_{32} V + s_{31} \lambda_1$ \\
$s_{4}$ & $s_{42}  U V$ & $s_{42} U$ & $s_{42} V$ \\
$s_{5}$ & $s_{52} UV$ & $s_{52} U$ & $s_{52} V$ \\
$s_{6}$ & $s_{62} U V$ & $s_{62} U$ & $s_{62} V$ \\
$s_{7}$ & $0$ & $0$ & $0$ \\
$s_{8}$ & $s_{82} U V$ & $s_{82} U$ & $s_{82} V$ \\ \hline
\end{tabular}
\end{center}
This time, we obtain the following Tate vectors
\be
\vec{t}_V = \left(1,1,3,3,5,8 \right), \qquad \vec{t}_U = \left(1,2,3,4,5,10 \right)
\ee
which signal an SO(11) singularity at $V=0$ as well as an E$_8$ singularity at $U=0$. The former one is inherited by the half K3 surface $X^+$ while the latter one moves into $X^-$.

The spectral cover is in this case given by
\be
p^+ = x_1^2  \left(s_{11} x_1^2  + s_{21} x_1 x_2 + s_{31} x_2^2 \right) , \qquad p^- = s_{11}x_1^4 \, .
\ee
We only comment on the non-trivial spectral cover. After applying the transformation \eqref{transfP112Weier}, it reads
\be
p^+_W = c_0 + c_1 x + c_2 x^2
\ee
which defines an Sp(2) $\cong$ SO(5)  spectral cover\footnote{Sometimes, Sp($N$) is denoted by Sp($2N$).} \cite{Friedman:1997yq} and is precisely what is expected.  Thus, we obtain an Sp(2) spectral cover in the case of $X^+$ and a trivial structure group for the case of $X^-$. The commutant of SO(5) within E$_8$ is given by SO(11).

\section{Tuned models without rational sections} \label{lemma}

In this appendix we reproduce \cite{Morrison:2012ei, Morrison:2014era} the following

\begin{lemma}\label{Lemma:SectionsMerge}
The two sections denoted by $x_1=0$ and $x_4=0$ in \eqref{eq:Rk1GeneralTwoSections} merge into a single section if and only if $s_{7}=0$ in \eqref{genericsectionP112}. Furthermore, the single section is given by $[x_1:x_2:x_3:x_4:x_5]=[0:1:1:0:1]$.
\end{lemma}

\begin{proof}
Suppose the two sections $x_1=0$ and $x_4=0$ merge into a single section. Then this single section obeys both $x_1=0$ and $x_4=0$, everywhere. Thus the Stanley-Reisner ideal requires $x_2 \neq 0$, $x_3 \neq 0$ and $x_5 \neq 0$ everywhere. Making use of the skaling relations of the resolved space $\text{Bl}_1 \mathbb{P}^{(1,1,2)}$, one obtains that this section is indeed given by $[x_1:x_2:x_3:x_4:x_5]=[0:1:1:0:1]$. 

Suppose now that $s_7=0$. Setting $x_1$ in \eqref{genericsectionP112} to zero, results in the equation $s_8 x_3^2 x_4=0$. As $x_3 \neq 0$ due to the Stanley Reisner ideal, $x_4$ has to vanish as well resulting in the merging of the two sections. Similarly, $x_4=0$ requires that $s_4 x_1 x_2^3 x_5^2=0$. The Stanley Reisner ideal requires $x_2$ and $x_5$ to be non-vanishing. Thus, there is also in this case only one section given by $[x_1:x_2:x_3:x_4:x_5]=[0:1:1:0:1]$. 

\end{proof}

\section{Non-commutativity of the semi-stable degeneration limit and the map to Weierstrass form}
\label{app:noncomm}

We illustrate the non-commutativity of the diagram \eqref{noncommutativediagram} using the above example with gauge group  E$_7$ $\times$ SO(9) $\times$ $\text{U}(1)$. To be precise, on the top left corner of the diagram, the section $\chi$ of $-K_{\mathbb{P}^{(1,1,2)}\times \mathbb{P}^1} $ is given by

\begin{align}
\chi: \qquad & s_1x_1^4+s_2x_1^3x_2+s_3x_1^2x_2^2+s_4x_1x_2^3+s_5x_1^2x_3+s_6x_1x_2x_3+s_7x_2^2x_3+s_8x_3^2=0\,,
\nonumber\\
\text{where} \qquad & s_1= s_{11}U^2+s_{12}UV+s_{13}V^2\,, \quad 
s_i=s_{i1}U^2+s_{i2}UV \text{ for } 2\leq i \leq 3 \,, \nonumber\\
& s_i= s_{i2}UV \text{ for } 4\leq i \leq 8\,.
\end{align}
Under the stable degeneration limit, denoted by the left map in the diagram \eqref{noncommutativediagram}, $\chi$ is split into $\chi^{\pm}$, which are in turn defined by

\begin{align}
\chi^{\pm}: \qquad & s_1^{\pm}x_1^4 + s_2^{\pm}x_1^3x_2 + s_3^{\pm} x_1^2x_2^2 + s_4^{\pm}x_1x_2^3 + s_5^{\pm}x_1^2x_3 + s_6^{\pm}x_1x_2x_3 + s_7^{\pm}x_2^2x_3 + s_8^{\pm}x_3^2=0\,,
\nonumber\\
\text{where} \quad & s_1^+= s_{12}U+s_{13}\mu \,, \quad s_1^-= s_{11}\lambda_1+s_{12}V \,, \nonumber\\ 
& s_i^+=s_{i2}U \text{ and } s_i^- = s_{i1}\lambda_1+s_{i2}V \text{ for } 2\leq i \leq 3\,, \nonumber\\
& s_i^+=s_{i2}U \text{ and } s_i^-=s_{i2}V \text{ for } 4\leq i \leq 8\,. 
\end{align}
We further map $\chi^{\pm} $, under the bottom map of the diagram \eqref{noncommutativediagram}, into their respective Weierstrass forms 
\begin{equation}
W_{\chi}^{\pm}: y^2 =x^3 + f_\chi^\pm x z^4 +g_\chi^\pm z^6 \,. 
\end{equation}
We can show that $W_{\chi}^{\pm}$ obtained in this way is different compared to $W_{\chi}^{'\pm}$ obtained by taking the other route in diagram \eqref{noncommutativediagram}, namely start from $\chi$ on the top left corner of the diagram, first map $\chi$ into its Weierstrass form $W_{\chi}$ using the map on top of \eqref{noncommutativediagram}, and then use the map on the right of \eqref{noncommutativediagram} to split $W_{\chi}$ into 
\begin{equation}
W_{\chi}^{'\pm}: y^2 =x^3 + f_\chi^{'\pm} x z^4 +g_\chi^{'\pm} z^6 \,. 
\end{equation}
Indeed,  
\begin{equation}
W_{\chi}^{+} \neq W_{\chi}^{'+}\,, \qquad W_{\chi}^{-} \neq W_{\chi}^{'-}.
\end{equation}
To be precise, 
\begin{equation}
f_{\chi}^{\pm} = f_{\chi}^{'\pm} \quad \text{ but } \quad g_{\chi}^{\pm} \neq g_{\chi}^{'\pm}\,, \quad g_{\chi}^{+} - g_{\chi}^{'+} = \frac{2}{3}U^6 s_{13}s_{31}s_{72}^2s_{82}^2\,, \quad g_{\chi}^{-} - g_{\chi}^{'-} = \frac{2}{3}V^6 s_{13}s_{31}s_{72}^2s_{82}^2\,.
\end{equation}

\bibliographystyle{utphys}	
\bibliography{ref}

\providecommand{\href}[2]{#2}\begingroup\raggedright\begin{thebibliography}{10}

\bibitem{green2012superstring}
M.~B. Green, J.~H. Schwarz, and E.~Witten, {\em Superstring theory: volume 2,
  Loop amplitudes, anomalies and phenomenology}.
\newblock Cambridge university press, 2012.

\bibitem{polchinski1998string}
J.~Polchinski, {\em {String Theory}}, vol.~1\&~2.
\newblock Cambridge University Press, 1998.

\bibitem{Vafa:1996xn}
C.~Vafa, ``{Evidence for F theory},''
  \href{http://dx.doi.org/10.1016/0550-3213(96)00172-1}{{\em Nucl.Phys.}
  {\bfseries B469} (1996) 403--418},
\href{http://arxiv.org/abs/hep-th/9602022}{{\ttfamily arXiv:hep-th/9602022
  [hep-th]}}.

\bibitem{Morrison:1996pp}
D.~R. Morrison and C.~Vafa, ``{Compactifications of F theory on Calabi-Yau
  threefolds. 2.},'' \href{http://dx.doi.org/10.1016/0550-3213(96)00369-0}{{\em
  Nucl.Phys.} {\bfseries B476} (1996) 437--469},
\href{http://arxiv.org/abs/hep-th/9603161}{{\ttfamily arXiv:hep-th/9603161
  [hep-th]}}.

\bibitem{Morrison:1996na}
D.~R. Morrison and C.~Vafa, ``{Compactifications of F theory on Calabi-Yau
  threefolds. 1},'' \href{http://dx.doi.org/10.1016/0550-3213(96)00242-8}{{\em
  Nucl.Phys.} {\bfseries B473} (1996) 74--92},
\href{http://arxiv.org/abs/hep-th/9602114}{{\ttfamily arXiv:hep-th/9602114
  [hep-th]}}.

\bibitem{Aspinwall:1997ye}
P.~S. Aspinwall and D.~R. Morrison, ``{Point - like instantons on K3
  orbifolds},'' \href{http://dx.doi.org/10.1016/S0550-3213(97)00516-6}{{\em
  Nucl.Phys.} {\bfseries B503} (1997) 533--564},
  \href{http://arxiv.org/abs/hep-th/9705104}{{\ttfamily arXiv:hep-th/9705104
  [hep-th]}}.

\bibitem{Friedman:1997yq}
R.~Friedman, J.~Morgan, and E.~Witten, ``{Vector bundles and F theory},''
  \href{http://dx.doi.org/10.1007/s002200050154}{{\em Commun. Math. Phys.}
  {\bfseries 187} (1997) 679--743},
\href{http://arxiv.org/abs/hep-th/9701162}{{\ttfamily arXiv:hep-th/9701162
  [hep-th]}}.

\bibitem{Donagi:1997dp}
R.~Y. Donagi, ``{Principal bundles on elliptic fibrations},'' {\em Asian J.
  Math.} {\bfseries 1} (1997) 214--223,
\href{http://arxiv.org/abs/alg-geom/9702002}{{\ttfamily arXiv:alg-geom/9702002
  [alg-geom]}}.

\bibitem{Andreas:1998zf}
B.~Andreas, ``{N=1 heterotic / F theory duality},''
  \href{http://dx.doi.org/10.1002/(SICI)1521-3978(199906)47:6<587::AID-PROP587>3.0.CO;2-5}{{\em
  Fortsch. Phys.} {\bfseries 47} (1999) 587--642},
\href{http://arxiv.org/abs/hep-th/9808159}{{\ttfamily arXiv:hep-th/9808159
  [hep-th]}}.

\bibitem{Candelas:1996su}
P.~Candelas and A.~Font, ``{Duality between the webs of heterotic and type II
  vacua},'' \href{http://dx.doi.org/10.1016/S0550-3213(96)00410-5}{{\em
  Nucl.Phys.} {\bfseries B511} (1998) 295--325},
\href{http://arxiv.org/abs/hep-th/9603170}{{\ttfamily arXiv:hep-th/9603170
  [hep-th]}}.

\bibitem{Berglund:1998ej}
P.~Berglund and P.~Mayr, ``{Heterotic string / F theory duality from mirror
  symmetry},'' {\em Adv. Theor. Math. Phys.} {\bfseries 2} (1999) 1307--1372,
\href{http://arxiv.org/abs/hep-th/9811217}{{\ttfamily arXiv:hep-th/9811217
  [hep-th]}}.

\bibitem{Anderson:2014gla}
L.~B. Anderson and W.~Taylor, ``{Geometric constraints in dual F-theory and
  heterotic string compactifications},''
  \href{http://dx.doi.org/10.1007/JHEP08(2014)025}{{\em JHEP} {\bfseries 08}
  (2014) 025},
\href{http://arxiv.org/abs/1405.2074}{{\ttfamily arXiv:1405.2074 [hep-th]}}.

\bibitem{Hayashi:2008ba}
H.~Hayashi, R.~Tatar, Y.~Toda, T.~Watari, and M.~Yamazaki, ``{New Aspects of
  Heterotic--F Theory Duality},''
  \href{http://dx.doi.org/10.1016/j.nuclphysb.2008.07.031}{{\em Nucl.Phys.}
  {\bfseries B806} (2009) 224--299},
\href{http://arxiv.org/abs/0805.1057}{{\ttfamily arXiv:0805.1057 [hep-th]}}.

\bibitem{Grimm:2009ef}
T.~W. Grimm, T.-W. Ha, A.~Klemm, and D.~Klevers, ``{Computing Brane and Flux
  Superpotentials in F-theory Compactifications},''
  \href{http://dx.doi.org/10.1007/JHEP04(2010)015}{{\em JHEP} {\bfseries 1004}
  (2010) 015},
\href{http://arxiv.org/abs/0909.2025}{{\ttfamily arXiv:0909.2025 [hep-th]}}.

\bibitem{Grimm:2009sy}
T.~W. Grimm, T.-W. Ha, A.~Klemm, and D.~Klevers, ``{Five-Brane Superpotentials
  and Heterotic / F-theory Duality},''
  \href{http://dx.doi.org/10.1016/j.nuclphysb.2010.06.011}{{\em Nucl. Phys.}
  {\bfseries B838} (2010) 458--491},
\href{http://arxiv.org/abs/0912.3250}{{\ttfamily arXiv:0912.3250 [hep-th]}}.

\bibitem{Jockers:2009ti}
H.~Jockers, P.~Mayr, and J.~Walcher, ``{On N=1 4d Effective Couplings for
  F-theory and Heterotic Vacua},''
  \href{http://dx.doi.org/10.4310/ATMP.2010.v14.n5.a3}{{\em Adv. Theor. Math.
  Phys.} {\bfseries 14} (2010) 1433--1514},
\href{http://arxiv.org/abs/0912.3265}{{\ttfamily arXiv:0912.3265 [hep-th]}}.

\bibitem{Cvetic:2011gp}
M.~Cvetic, I.~Garcia~Etxebarria, and J.~Halverson, ``{Three Looks at Instantons
  in F-theory -- New Insights from Anomaly Inflow, String Junctions and
  Heterotic Duality},'' \href{http://dx.doi.org/10.1007/JHEP11(2011)101}{{\em
  JHEP} {\bfseries 11} (2011) 101},
\href{http://arxiv.org/abs/1107.2388}{{\ttfamily arXiv:1107.2388 [hep-th]}}.

\bibitem{Cvetic:2012ts}
M.~Cvetic, R.~Donagi, J.~Halverson, and J.~Marsano, ``{On Seven-Brane Dependent
  Instanton Prefactors in F-theory},''
  \href{http://dx.doi.org/10.1007/JHEP11(2012)004}{{\em JHEP} {\bfseries 11}
  (2012) 004},
\href{http://arxiv.org/abs/1209.4906}{{\ttfamily arXiv:1209.4906 [hep-th]}}.

\bibitem{Aldazabal:1996du}
G.~Aldazabal, A.~Font, L.~E. Ibanez, and A.~M. Uranga, ``{New branches of
  string compactifications and their F theory duals},''
  \href{http://dx.doi.org/10.1016/S0550-3213(96)00699-2}{{\em Nucl. Phys.}
  {\bfseries B492} (1997) 119--151},
\href{http://arxiv.org/abs/hep-th/9607121}{{\ttfamily arXiv:hep-th/9607121
  [hep-th]}}.

\bibitem{Klemm:1996hh}
A.~Klemm, P.~Mayr, and C.~Vafa, ``{BPS states of exceptional noncritical
  strings},''
\href{http://arxiv.org/abs/hep-th/9607139}{{\ttfamily arXiv:hep-th/9607139
  [hep-th]}}.

\bibitem{GrassiIPMU}
A.~Grassi, ``{Toric Weierstrass models},'' {\em I.P.M.U. 2010 talk} (2010)
  http://www.ipmu.jp/node/552.

\bibitem{Morrison:2012ei}
D.~R. Morrison and D.~S. Park, ``{F-Theory and the Mordell-Weil Group of
  Elliptically-Fibered Calabi-Yau Threefolds},''
  \href{http://dx.doi.org/10.1007/JHEP10(2012)128}{{\em JHEP} {\bfseries 1210}
  (2012) 128},
\href{http://arxiv.org/abs/1208.2695}{{\ttfamily arXiv:1208.2695 [hep-th]}}.

\bibitem{Borchmann:2013jwa}
E.~P. J.~Borchmann, C.~Mayrhofer and T.~Weigand, ``{Elliptic fibrations for
  SU(5) x U(1) x U(1) F-theory vacua},''
\href{http://arxiv.org/abs/1303.5054}{{\ttfamily arXiv:1303.5054 [hep-th]}}.

\bibitem{Cvetic:2013nia}
M.~Cveti{\v c}, D.~Klevers, and H.~Piragua, ``{F-Theory Compactifications with
  Multiple U(1)-Factors: Constructing Elliptic Fibrations with Rational
  Sections},''
\href{http://arxiv.org/abs/1303.6970}{{\ttfamily arXiv:1303.6970 [hep-th]}}.

\bibitem{Cvetic:2013uta}
M.~Cveti{\v c}, A.~Grassi, D.~Klevers, and H.~Piragua, ``{Chiral
  Four-Dimensional F-Theory Compactifications With SU(5) and Multiple
  U(1)-Factors},''
\href{http://arxiv.org/abs/1306.3987}{{\ttfamily arXiv:1306.3987 [hep-th]}}.

\bibitem{Borchmann:2013hta}
J.~Borchmann, C.~Mayrhofer, E.~Palti, and T.~Weigand, ``{SU(5) Tops with
  Multiple U(1)s in F-theory},''
\href{http://arxiv.org/abs/1307.2902}{{\ttfamily arXiv:1307.2902 [hep-th]}}.

\bibitem{Cvetic:2013jta}
M.~Cveti{\v c}, D.~Klevers, and H.~Piragua, ``{F-Theory Compactifications with
  Multiple U(1)-Factors: Addendum},''
  \href{http://dx.doi.org/10.1007/JHEP12(2013)056}{{\em JHEP} {\bfseries 1312}
  (2013) 056},
\href{http://arxiv.org/abs/1307.6425}{{\ttfamily arXiv:1307.6425 [hep-th]}}.

\bibitem{Cvetic:2013qsa}
M.~Cveti{\v c}, D.~Klevers, H.~Piragua, and P.~Song, ``{Elliptic fibrations
  with rank three Mordell-Weil group: F-theory with U(1) x U(1) x U(1) gauge
  symmetry},'' \href{http://dx.doi.org/10.1007/JHEP03(2014)021}{{\em JHEP}
  {\bfseries 1403} (2014) 021},
\href{http://arxiv.org/abs/1310.0463}{{\ttfamily arXiv:1310.0463 [hep-th]}}.

\bibitem{Braun:2014qka}
V.~Braun, T.~W. Grimm, and J.~Keitel, ``{Complete Intersection Fibers in
  F-Theory},'' \href{http://dx.doi.org/10.1007/JHEP03(2015)125}{{\em JHEP}
  {\bfseries 03} (2015) 125},
\href{http://arxiv.org/abs/1411.2615}{{\ttfamily arXiv:1411.2615 [hep-th]}}.

\bibitem{Klevers:2014bqa}
D.~Klevers, D.~K. Mayorga~Pena, P.-K. Oehlmann, H.~Piragua, and J.~Reuter,
  ``{F-Theory on all Toric Hypersurface Fibrations and its Higgs Branches},''
  \href{http://dx.doi.org/10.1007/JHEP01(2015)142}{{\em JHEP} {\bfseries 1501}
  (2015) 142},
\href{http://arxiv.org/abs/1408.4808}{{\ttfamily arXiv:1408.4808 [hep-th]}}.

\bibitem{Cvetic:2015ioa}
M.~Cveti{\v c}, D.~Klevers, H.~Piragua, and W.~Taylor, ``{General U(1)xU(1)
  F-theory Compactifications and Beyond: Geometry of unHiggsings and novel
  Matter Structure},''
\href{http://arxiv.org/abs/1507.05954}{{\ttfamily arXiv:1507.05954 [hep-th]}}.

\bibitem{Grassi:2012qw}
A.~Grassi and V.~Perduca, ``{Weierstrass models of elliptic toric K3
  hypersurfaces and symplectic cuts},''
\href{http://arxiv.org/abs/1201.0930}{{\ttfamily arXiv:1201.0930 [math.AG]}}.

\bibitem{Blumenhagen:2005ga}
R.~Blumenhagen, G.~Honecker, and T.~Weigand, ``{Loop-corrected
  compactifications of the heterotic string with line bundles},''
  \href{http://dx.doi.org/10.1088/1126-6708/2005/06/020}{{\em JHEP} {\bfseries
  06} (2005) 020},
\href{http://arxiv.org/abs/hep-th/0504232}{{\ttfamily arXiv:hep-th/0504232
  [hep-th]}}.

\bibitem{Marsano:2009gv}
J.~Marsano, N.~Saulina, and S.~Schafer-Nameki, ``{Monodromies, Fluxes, and
  Compact Three-Generation F-theory GUTs},''
  \href{http://dx.doi.org/10.1088/1126-6708/2009/08/046}{{\em JHEP} {\bfseries
  08} (2009) 046},
\href{http://arxiv.org/abs/0906.4672}{{\ttfamily arXiv:0906.4672 [hep-th]}}.

\bibitem{Blumenhagen:2009yv}
R.~Blumenhagen, T.~W. Grimm, B.~Jurke, and T.~Weigand, ``{Global F-theory
  GUTs},'' \href{http://dx.doi.org/10.1016/j.nuclphysb.2009.12.013}{{\em
  Nucl.Phys.} {\bfseries B829} (2010) 325--369},
\href{http://arxiv.org/abs/0908.1784}{{\ttfamily arXiv:0908.1784 [hep-th]}}.

\bibitem{Grimm:2010ez}
T.~W. Grimm and T.~Weigand, ``{On Abelian Gauge Symmetries and Proton Decay in
  Global F-theory GUTs},''
  \href{http://dx.doi.org/10.1103/PhysRevD.82.086009}{{\em Phys.Rev.}
  {\bfseries D82} (2010) 086009},
  \href{http://arxiv.org/abs/1006.0226}{{\ttfamily arXiv:1006.0226 [hep-th]}}.

\bibitem{Aspinwall:1998xj}
P.~S. Aspinwall and D.~R. Morrison, ``{Nonsimply connected gauge groups and
  rational points on elliptic curves},'' {\em JHEP} {\bfseries 9807} (1998)
  012,
\href{http://arxiv.org/abs/hep-th/9805206}{{\ttfamily arXiv:hep-th/9805206
  [hep-th]}}.

\bibitem{Choi:2012pr}
K.-S. Choi and H.~Hayashi, ``{U(n) Spectral Covers from Decomposition},''
  \href{http://dx.doi.org/10.1007/JHEP06(2012)009}{{\em JHEP} {\bfseries 06}
  (2012) 009},
\href{http://arxiv.org/abs/1203.3812}{{\ttfamily arXiv:1203.3812 [hep-th]}}.

\bibitem{Aspinwall:1996mn}
P.~S. Aspinwall, ``{K3 surfaces and string duality},'' in {\em {Fields, strings
  and duality. Proceedings, Summer School, Theoretical Advanced Study Institute
  in Elementary Particle Physics, TASI'96, Boulder, USA, June 2-28, 1996}},
  pp.~421--540.
\newblock 1996.
\newblock
\href{http://arxiv.org/abs/hep-th/9611137}{{\ttfamily arXiv:hep-th/9611137
  [hep-th]}}.
\newblock

\bibitem{Vafa:1995gm}
C.~Vafa and E.~Witten, ``{Dual string pairs with N = 1 and N = 2 supersymmetry
  in four dimensions},''
  \href{http://dx.doi.org/10.1016/0920-5632(96)00025-4}{{\em Nucl. Phys. Proc.
  Suppl.} {\bfseries 46} (1996) 225--247},
\href{http://arxiv.org/abs/hep-th/9507050}{{\ttfamily arXiv:hep-th/9507050}}.

\bibitem{Bershadsky:1996nh}
M.~Bershadsky, K.~A. Intriligator, S.~Kachru, D.~R. Morrison, V.~Sadov, {\em
  et~al.}, ``{Geometric singularities and enhanced gauge symmetries},''
  \href{http://dx.doi.org/10.1016/S0550-3213(96)90131-5}{{\em Nucl.Phys.}
  {\bfseries B481} (1996) 215--252},
\href{http://arxiv.org/abs/hep-th/9605200}{{\ttfamily arXiv:hep-th/9605200
  [hep-th]}}.

\bibitem{kodaira1963compact}
K.~Kodaira, ``On compact analytic surfaces: Ii,'' {\em The Annals of
  Mathematics} {\bfseries 77} no.~3, (1963) 563--626.

\bibitem{Douglas:2014ywa}
M.~R. Douglas, D.~S. Park, and C.~Schnell, ``{The Cremmer-Scherk Mechanism in
  F-theory Compactifications on K3 Manifolds},''
  \href{http://dx.doi.org/10.1007/JHEP05(2014)135}{{\em JHEP} {\bfseries 05}
  (2014) 135},
\href{http://arxiv.org/abs/1403.1595}{{\ttfamily arXiv:1403.1595 [hep-th]}}.

\bibitem{Donagi:1998vw}
R.~Donagi, ``{Heterotic / F theory duality: ICMP lecture},'' in {\em
  {Mathematical physics. Proceedings, 12th International Congress, ICMP'97,
  Brisbane, Australia, July 13-19, 1997}}.
\newblock 1998.
\newblock \href{http://arxiv.org/abs/hep-th/9802093}{{\ttfamily
  arXiv:hep-th/9802093 [hep-th]}}.
\newblock
\url{http://alice.cern.ch/format/showfull?sysnb=0270001}.
\newblock

\bibitem{Atiyah01011957}
M.~F. Atiyah, ``Vector bundles over an elliptic curve,''
  \href{http://dx.doi.org/10.1112/plms/s3-7.1.414}{{\em Proceedings of the
  London Mathematical Society} {\bfseries s3-7} no.~1, (1957) 414--452},
  \href{http://arxiv.org/abs/http://plms.oxfordjournals.org/content/s3-7/1/414.full.pdf+html}{{\ttfamily
  http://plms.oxfordjournals.org/content/s3-7/1/414.full.pdf+html}}.
  \url{http://plms.oxfordjournals.org/content/s3-7/1/414.short}.

\bibitem{Friedman:1997ih}
R.~Friedman, J.~W. Morgan, and E.~Witten, ``{Vector bundles over elliptic
  fibrations},''
\href{http://arxiv.org/abs/alg-geom/9709029}{{\ttfamily arXiv:alg-geom/9709029
  [alg-geom]}}.

\bibitem{Aspinwall:2005qw}
P.~S. Aspinwall, ``{An analysis of fluxes by duality},''
\href{http://arxiv.org/abs/hep-th/0504036}{{\ttfamily arXiv:hep-th/0504036
  [hep-th]}}.

\bibitem{tate1975algorithm}
J.~Tate, ``Algorithm for determining the type of a singular fiber in an
  elliptic pencil,'' {\em Modular functions of one variable IV} (1975) 33--52.

\bibitem{Katz:2011qp}
S.~Katz, D.~R. Morrison, S.~Schafer-Nameki, and J.~Sully, ``{Tate's algorithm
  and F-theory},'' \href{http://dx.doi.org/10.1007/JHEP08(2011)094}{{\em JHEP}
  {\bfseries 1108} (2011) 094},
\href{http://arxiv.org/abs/1106.3854}{{\ttfamily arXiv:1106.3854 [hep-th]}}.

\bibitem{2010arXiv1011.1003B}
U.~{Bruzzo} and A.~{Grassi}, ``{Picard group of hypersurfaces in toric
  3-folds},'' {\em ArXiv e-prints} (Nov., 2010) ,
  \href{http://arxiv.org/abs/1011.1003}{{\ttfamily arXiv:1011.1003 [math.AG]}}.

\bibitem{Batyrev:1994hm}
V.~V. Batyrev, ``{Dual polyhedra and mirror symmetry for Calabi-Yau
  hypersurfaces in toric varieties},'' {\em J.Alg.Geom.} {\bfseries 3} (1994)
  493--545,
\href{http://arxiv.org/abs/alg-geom/9310003}{{\ttfamily arXiv:alg-geom/9310003
  [alg-geom]}}.

\bibitem{Morrison:2014era}
D.~R. Morrison and W.~Taylor, ``{Sections, multisections, and U(1) fields in
  F-theory},''
\href{http://arxiv.org/abs/1404.1527}{{\ttfamily arXiv:1404.1527 [hep-th]}}.

\bibitem{Donagi:1999ez}
R.~Donagi, B.~A. Ovrut, T.~Pantev, and D.~Waldram, ``{Standard models from
  heterotic M theory},'' {\em Adv. Theor. Math. Phys.} {\bfseries 5} (2002)
  93--137,
\href{http://arxiv.org/abs/hep-th/9912208}{{\ttfamily arXiv:hep-th/9912208
  [hep-th]}}.

\bibitem{OguisoShioda}
T.~Shioda and K.~Oguiso, ``{The Mordell-Weil lattice of a rational elliptic
  surface},''
{\em Comment. Math. Univ. St. Paul} {\bfseries 40} (1991) 83--99.

\bibitem{Mayrhofer:2014opa}
C.~Mayrhofer, D.~R. Morrison, O.~Till, and T.~Weigand, ``{Mordell-Weil Torsion
  and the Global Structure of Gauge Groups in F-theory},''
\href{http://arxiv.org/abs/1405.3656}{{\ttfamily arXiv:1405.3656 [hep-th]}}.

\bibitem{Grassi:2011hq}
A.~Grassi and D.~R. Morrison, ``{Anomalies and the Euler characteristic of
  elliptic Calabi-Yau threefolds},''
\href{http://arxiv.org/abs/1109.0042}{{\ttfamily arXiv:1109.0042 [hep-th]}}.

\end{thebibliography}\endgroup

\end{document}